\documentclass[11pt,letterpaper,singlespace,english]{article}

\usepackage[vlined,titlenotnumbered,noend]{algorithm2e}
\SetKwComment{Comment}{$\triangleright$\ }{}
\usepackage[margin=1in]{geometry}
\usepackage{amsmath}
\usepackage{enumerate}
\usepackage{setspace}
\usepackage{framed}
\usepackage{upgreek}
\usepackage[english]{babel}
\usepackage{latexsym,amssymb,amsmath,amsthm}
\usepackage{graphics,graphicx}
\usepackage{subfigure}
\usepackage{listings}

\usepackage{comment}
\usepackage{array}
\usepackage{relsize}
\usepackage{fullpage}
\usepackage[T1]{fontenc}
\usepackage{lmodern}
\usepackage{calrsfs}

\DeclareMathAlphabet{\mathcal}{OMS}{cmsy}{m}{n}

\usepackage{microtype}
\usepackage{multirow} 
\usepackage{cite}
\usepackage{amsfonts}
\usepackage{amstext}
\usepackage{mathtools}
\usepackage{float}
\usepackage{paralist}

\usepackage{tikz}
\usetikzlibrary{calc}
\usetikzlibrary{shapes,arrows,positioning,shadows,snakes}

\usepackage{footmisc}

\usepackage[hang,small,bf]{caption}

\usepackage{xcolor}
\usepackage{amsthm, thmtools}
\usepackage{nameref}
\definecolor{ForestGreen}{rgb}{0.1333,0.5451,0.1333}
\definecolor{DarkRed}{rgb}{0.8,0,0}
\definecolor{Red}{rgb}{1,0,0}
\usepackage[linktocpage=true,
pagebackref=true,colorlinks,
linkcolor=DarkRed,citecolor=ForestGreen,
bookmarks,bookmarksopen,bookmarksnumbered]
{hyperref}

\usepackage{cleveref}
\usepackage{thm-restate}

\declaretheorem[numberwithin=section,refname={Theorem,Theorems},Refname={Theorem,Theorems},name={Theorem}]{thm}
\declaretheorem[numberlike=thm,refname={Lemma,Lemmas},Refname={Lemma,Lemmas},name={Lemma}]{lem}
\declaretheorem[numberlike=thm,refname={Corollary,Corollaries},Refname={Corollary,Corollaries},name={Corollary}]{cor}

\declaretheorem[numberlike=thm,name={Conjecture},refname={Conjecture,Conjectures},Refname={Conjecture,Conjectures}]{conjecture}
\declaretheorem[numberlike=thm,refname={Proposition,Propositions},Refname={Proposition,Propositions},name={Proposition}]{prop}

\declaretheorem[numberlike=thm,refname={Definition,Definitions},Refname={Definition,Definitions},name={Definition}]{defn}
\declaretheorem[style=remark,numberlike=thm,refname={Remark,Remarks},Refname={Remark,Remarks},name={Remark}]{rem}
\declaretheorem[style=remark,numberlike=thm,refname={Claim,Claims},Refname={Claim,Claims}]{claim}

\AtBeginDocument{%
\let\oldref\ref
\let\ref\Cref
}

\newcommand{\reffontsize}{\fontsize{11}{12}\selectfont}

  \renewenvironment{thebibliography}[1]{%
    \begin{oldthebibliography}{#1}%
      \setlength{\parskip}{0.7ex}%
      \setlength{\itemsep}{0ex}%
  }%
  {%
    \end{oldthebibliography}%
  }

\newcommand{\link}{{\tt link}}

\newcommand{\NEXT}{\mbox{{\tt next}}}
\newcommand{\PREV}{\mbox{{\tt prev}}}
\newcommand{\LM}{\mbox{{\tt leftmost}}}
\newcommand{\RM}{\mbox{{\tt rightmost}}}
\newcommand{\KEY}{\mbox{{\tt key}}}
\newcommand{\NULL}{\mbox{{\tt null}}}

\newcommand\TombStone{\rule{.7ex}{1.7ex}}
\renewcommand{\qedsymbol}{\TombStone}
\newcommand{\qedd}{\let\qed\relax\quad\raisebox{-.1ex}{$\qedsymbol$}}

\begin{document}
	
\global\long\def\boxx{\mathtt{box}}
\global\long\def\base{\mathtt{base}}
\global\long\def\add{\mathtt{ADD}}
\global\long\def\greedy{\textsf{Greedy\,}}
\global\long\def\smooth{\textsf{Smooth\,}}
\global\long\def\greedyfuture{\textsf{GreedyFuture\,}}
\global\long\def\eqdef{\overset{\textnormal{def}}{=}}
\global\long\def\cA{\mathcal{A}}
\global\long\def\cB{\mathcal{B}}
\global\long\def\cF{\mathcal{F}}
\global\long\def\cP{\mathcal{P}}
\global\long\def\opt{\mathsf{OPT}}

\global\long\def\bst{\mathtt{bst}}
\global\long\def\bstins{\mathtt{ibst}}
\global\long\def\sat{\mathtt{sat}}
\global\long\def\satins{\mathtt{isat}}
\global\long\def\path{\mathtt{path}}
\global\long\def\stable{\mathtt{stable}}
\global\long\def\coord#1#2{\left\langle #1,#2\right\rangle }

\title{Smooth heaps and a dual view of self-adjusting data structures\thanks{A preliminary version of this paper appears in Symposium on Theory of Computing (STOC) 2018.}}


\author{L\'{a}szl\'{o} Kozma \thanks{Freie Universit\"at Berlin, Germany. 
Work done while at the Blavatnik School of Computer Science, Tel Aviv University, Israel, and TU Eindhoven, Netherlands. E-mail:
{\tt laszlo.kozma@fu-berlin.de}.} \and
\and  Thatchaphol Saranurak \thanks{KTH Royal Institute of Technology, Stockholm, Sweden. E-mail: {\tt
    thasar@kth.se}.}}

\date{}

\pagenumbering{gobble}
\maketitle

\begin{abstract}
	We present a new connection between self-adjusting binary search trees
	(BSTs) and heaps, two fundamental, extensively studied, and practically
	relevant families of data structures (Allen, Munro, 1978; Sleator,
	Tarjan, 1983; Fredman, Sedgewick, Sleator, Tarjan, 1986; Wilber, 1989;
	Fredman, 1999; Iacono, \"Ozkan, 2014). Roughly speaking, we map an
	\emph{arbitrary} heap algorithm within a natural model, to a corresponding BST algorithm with the same cost on a \emph{dual} sequence of operations (i.e.\ the same sequence with the roles of
	\emph{time} and \emph{key-space} switched). This is the first general
	transformation between the two families of data structures. 
	
There is a rich theory of dynamic optimality for BSTs (i.e.\ the theory of competitiveness between BST algorithms). The lack of an analogous theory for heaps has been noted in the literature (e.g.\ Pettie; 2005, 2008). Through our connection, we transfer all instance-specific lower bounds known for BSTs to a general model of heaps, initiating a  
 theory of \emph{dynamic optimality for heaps}. 

On the algorithmic side, we obtain a new, simple and efficient heap algorithm, which we call the \emph{smooth heap}. We show the smooth heap to be the heap-counterpart of Greedy, the BST algorithm with the strongest proven and conjectured properties from the literature, widely believed to be instance-optimal (Lucas, 1988; Munro, 2000; Demaine et al., 2009). Assuming the optimality of Greedy, the smooth heap is also optimal within our model of heap algorithms. As corollaries of results known for Greedy, we obtain instance-specific upper bounds for the smooth heap, with applications in adaptive sorting.

Intriguingly, the smooth heap, although derived from a non-practical BST algorithm, is simple and easy to implement (e.g.\ it stores no auxiliary data besides the keys and tree pointers). It can be seen as a variation on
the popular \emph{pairing heap} data structure, extending it with a ``power-of-two-choices'' type of heuristic. 

\end{abstract}

\begin{flushleft}
	{The paper is dedicated to \textit{Raimund Seidel} on the occasion of his sixtieth birthday.}
\end{flushleft}

\newpage

\newpage

{\small \tableofcontents{}}
\pagebreak{}
\pagenumbering{arabic}

\section{Introduction}\label{sec:intro}

We revisit two popular families of comparison-based data structures that implement the fundamental \emph{dictionary} and \emph{priority queue} abstract data types.

A dictionary stores a set of keys from an ordered universe, supporting the operations \emph{search}, \emph{insert}, and \emph{delete}. Other operations required in some applications include finding the smallest or the largest key in a dictionary, finding the successor or predecessor of a key, merging two dictionaries, or listing the keys of a dictionary in sorted order.
 
A priority queue stores a set of keys from an ordered universe (often referred to as priorities), with typical operations such as finding and deleting the minimum key (``extract-min''), inserting a new key, decreasing a given key, and merging two priority queues. The fact that efficient search is \emph{not} required allows priority queues to perform some operations faster than dictionaries. 

Perhaps the best known and most natural implementations of dictionaries and priority queues are binary search trees (BSTs), respectively, multiway heaps. In both cases, a key is associated\footnote{Typically, a data record is also stored with each key, whose retrieval is the main purpose of the data structure; we ignore this aspect throughout the paper, as it does not affect the data structuring mechanisms we study.} to each node of the underlying tree structure, according to the standard in-order, respectively, (min-)heap order conditions. 
Both structures need an occasional re-structuring in order to remain efficient over the long term. In BSTs, the costs of the main operations are proportional to the depths of the affected nodes, the goal is therefore to keep the tree reasonably balanced. For heaps, the desired structure is, in some sense, the opposite; the crucial extract-min operations are easiest to perform if the heap is fully sorted, i.e.\ if it is a path. 
The design and analysis of efficient algorithms for heap- and BST-maintenance continues to be a central topic of algorithmic research, and a large number of ingenious techniques have been proposed in the literature.

Our contribution is a new connection between the two tasks, showing that, with some reasonable restrictions, an arbitrary algorithm for heap-maintenance encodes an algorithm for BST-maintenance on a related (dual) sequence of inputs. 
This is the first general transformation between the two families of data structures. 
The connection allows us to transfer results (both algorithms and lower bounds) between the two settings. 

On the algorithmic side, we obtain a new, simple and efficient heap data structure, which we call the \emph{smooth heap}. We show the smooth heap to be the heap-equivalent of Greedy, the BST algorithm with the strongest proven and conjectured properties from the literature. We thus obtain, for the smooth heap, instance-specific guarantees (upper bounds) that go beyond the logarithmic amortized guarantees known for existing heap implementations. Key to our result is a new, non-deterministic interpretation of Greedy, which may be of independent interest.

On the complexity side, we transfer all instance-specific lower bounds known for BSTs to a natural class of heap algorithms. Pettie observed~\cite{Pettie,deque_Pet08} that a theory of dynamic optimality for heaps, analogous to the rich set of results known for BSTs, is missing.  
We see the instance-specific lower and upper bounds that we transfer to heaps as essential components of such a theory.

\paragraph{Self-adjusting BSTs.}

The study of dictionaries based on BSTs goes back to the beginnings of computer science.\footnote{Arguably, as an abstraction of searching in ordered sets, binary search trees have been studied (explicitly or implicitly) well before the existence of computers, see e.g.\ methods for root-finding in mathematical analysis.} Standard balanced trees such as AVL-trees~\cite{avl}, red-black trees~\cite{Bayer1972, Guibas78}, or randomized treaps~\cite{SeidelA96} keep the tree balanced by storing and enforcing explicit bookkeeping information.   
We refer to Knuth~\cite[\S\,6.2]{Knuth3} for an extensive overview of classical methods. 

By contrast, splay trees, invented by Sleator and Tarjan~\cite{ST85}, do not store any auxiliary information (either in the nodes of the tree or globally) besides the keys of the dictionary and the pointers representing the tree. Instead, they employ an elegant local re-structuring whenever a key is accessed. In short, splay trees bring the accessed key to the root via a sequence of \emph{double-rotations}. Splay trees can be seen as a member of a more general family of \emph{self-adjusting} data structures.\footnote{Two simpler heuristics proposed earlier by Allen and Munro~\cite{Allen-Munro} also fall into this category, but are known to be inefficient on certain inputs.} 
Splay trees support all main operations in (optimal) logarithmic amortized time. Furthermore, they adapt to various usage patterns, providing stronger than logarithmic guarantees for certain structured sequences of operations; we call such guarantees \emph{instance-specific}. Several instance-specific upper bounds are known for splay trees (see e.g.\ \cite{ST85, tarjan_sequential, deque_Sun92, finger1, finger2, deque_Elm04, deque_Pet08}) and others have been conjectured. In fact, Sleator and Tarjan conjectured splay trees to be \emph{instance-optimal}, i.e.\ to match, up to a constant factor, the optimum among all BST algorithms.\footnote{The term ``instance-optimal'' is, in our context, equivalent with ``$O(1)$-competitive''. An algorithm is $c$-competitive, if its running time is at most $c$ times the running time of the optimal algorithm (for every given input).} This is the famous \emph{dynamic optimality conjecture}, still open. It is not known whether \emph{any} polynomial-time BST algorithm has such a property, even with a priori knowledge of the entire sequence of operations. 

Lucas~\cite{Luc88} proposed, as a theoretical construction, a BST algorithm that takes into account \emph{future} operations; in her algorithm, the search path is re-arranged such as to bring the soonest-to-be-accessed key to the root (and similarly in a recursive way in the left and right subtrees of the root). Following Demaine, Harmon, Iacono, Kane, and P\v{a}tra\c{s}cu\ \cite{DHIKP09}, we refer to this algorithm as GreedyFuture. GreedyFuture is widely believed to be $O(1)$-competitive~\cite{Luc88, Mun00, DHIKP09}, and is known to match essentially all instance-specific guarantees proven for splay trees, as well as some stronger guarantees not known to hold for splay trees~\cite{FOCS15, LI16}.  Unfortunately, $o(\log{n})$-competitiveness has not been shown for either Splay or GreedyFuture.

It was shown in~\cite{DHIKP09} that GreedyFuture can be simulated by an \emph{online} algorithm (which we refer to simply as Greedy), with only a constant factor increase in cost. In a different line of work, Demaine et al.\ described an online BST algorithm with a competitive ratio of $O(\log\log{n})$, called Tango tree~\cite{Tango}. Strictly speaking, neither Greedy (in the online version), nor Tango are self-adjusting, as they both store auxiliary  information in the nodes of the tree.  
Furthermore, due to their rather intricate details, neither can be considered a practical alternative to splay trees. It would be of great interest to find simple and practical BST algorithms that match the theoretical guarantees of Greedy or Tango.

\paragraph{Self-adjusting heaps.}

In 1984, Fredman and Tarjan introduced the Fibonacci heap~\cite{Fibonacci}, a data structure that achieves the theoretically optimal amortized bounds of $O(\log{n})$ for extract-min, and $O(1)$ for all other heap operations. Fibonacci heaps are rather complicated to implement and simpler alternatives tend to outperform them in practice~\cite{Tarjan14}. In the past three decades, a central goal of research in data structures has thus been to find a simpler heap implementation that would match the theoretical guarantees of Fibonacci heaps (see e.g.~\cite{Stasko, Rheap, Theap, ThinThick,  Violation, RankPairing,Strict, BrodalSurvey, Quake, Hollow}). 

Arguably, this goal has not been fully achieved. In particular, Fibonacci heaps, as well as most of their proposed alternatives store auxiliary bookkeeping information, or perform operations outside the comparison/pointer model, or maintain a pointer structure that is not ``forest-like'' (not following, as Iacono and \"Ozkan~\cite{IaconoOzkan} put it, \emph{``the spirit of what we usually think of as a heap''}.) The following general question is still open.

\vspace{0.1in}
\textit{
\textbf{Question 1.} Is there, by analogy to search trees, a simple, self-adjusting heap data structure with optimal amortized cost, possibly achieving instance-optimality?
}
\vspace{0.1in}

The closest in simplicity and elegance to self-adjusting trees are \emph{pairing heaps}, introduced by Fredman, Sedgewick, Sleator, and Tarjan~\cite{pairing}. Pairing heaps do not store auxiliary information besides the keys and the pointers of a single (multiway) heap, are easy to implement and efficient in practice~\cite{Tarjan14}. They implement all operations using the unit-cost \emph{link} primitive (\ref{fig1}). Pairing heaps perform extract-min, the only operation with a non-trivial implementation, in a \emph{two-pass re-structuring} of the children of the root.  
(We describe pairing heaps in more detail in \S\,\oldref{sub:heap model}.)

\begin{figure}[H]
	\begin{center}
\includegraphics[width=0.99\textwidth]{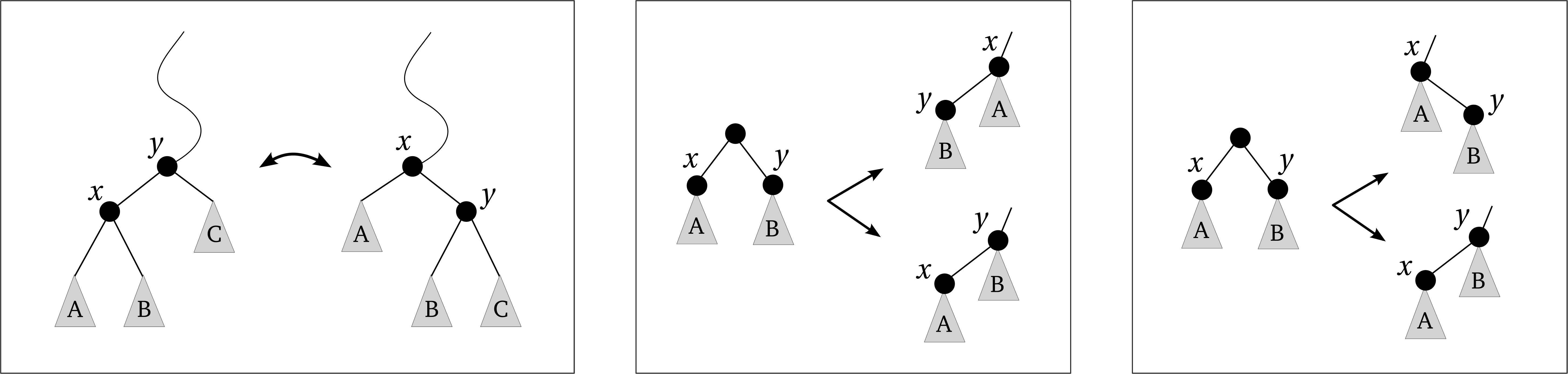}
	\end{center}
\caption{Elementary operations in BSTs and heaps. Letters $x$ and $y$ denote keys, $A$, $B$, $C$ are arbitrary subtrees. (i) Rotation in a binary search tree. (ii) Classical (``unstable'') link between neighboring siblings in a heap. The larger of $x$ and $y$ becomes the leftmost child of the smaller. (iii) Stable link between neighboring siblings in a heap. The larger of $x$ and $y$ becomes the leftmost child of the smaller, if it was originally to the left; otherwise it becomes the rightmost child.
\label{fig1}}
\end{figure}

Fredman et al.\ showed that pairing heaps support all operations in logarithmic amortized time. They also conjectured that pairing heaps match the optimal bounds of Fibonacci heaps. This conjecture was disproved by Fredman~\cite{FredmanLB}, who showed that in certain sequences, \emph{decrease-key} may cost $\Omega(\log\log{n})$. A similar lower bound was later shown by Iacono and \"{O}zkan~\cite{IaconoOzkan, IaconoOzkan2}.\footnote{We remark in passing that the complexity of the standard paring heap is not fully understood. Iacono~\cite{IaconoUB} showed that assuming $O(\log{n})$ for decrease-key, insert takes constant amortized time. Pettie~\cite{Pettie} showed that $O(4^{\sqrt{\log\log{n}}})$ can be simultaneously achieved for both insert and decrease-key. For natural variants of pairing heaps, even the logarithmic cost of extract-min has only been conjectured~\cite{pairing, forward_variant, multipass}. Improving these bounds remains a challenging open problem.}

The Fredman-, and Iacono-\"{O}zkan lower bounds hold for broad classes of ``pairing-heap-like'' data structures that include all natural variants of pairing heaps, as well as some other heap data structures from the literature. The heap models defined in these works are ``in the spirit'' of self-adjusting trees; to the extent that these models capture the idea of a ``self-adjusting heap'', the answer to Question 1 has to be negative. We argue, however, that there exist natural self-adjusting heaps that fall outside these models, leaving Question 1 open. 

For some concrete heaps, certain instance-specific \emph{upper bounds} were studied by Iacono and Langerman~\cite{queaps}, Elmasry~\cite{ElmasryWS}, and Elmasry et al.\ \cite{ElmasryEtc}.   
On the other hand, a theory of instance-specific \emph{lower bounds} (in a sense similar to BSTs) has not yet been proposed for heaps in \emph{any heap model}. We argue that the existence of such lower bounds is \emph{necessary} for dynamic optimality to be possible (even approximately). Informally speaking, even a candidate instance-optimal heap must have high cost on certain input sequences (by information-theoretic arguments); these input sequences must therefore be hard \emph{for every heap}. Put differently, some restriction on the allowed comparisons and links is necessary, to prevent a hypothetical optimal algorithm from ``guessing'' the correct permutation of the input.

\paragraph{Connecting BSTs and heaps.}
Self-adjusting BSTs and heaps serve different purposes, and the underlying trees they maintain have, in general, different characteristics.  
Yet, they are similar in their approach of performing local re-adjustments (using rotations, respectively links), seemingly ignoring global structure. The two families of data structures have been introduced and studied in the past decades by largely the same authors.

The following technical connection between splay trees and the standard variant of pairing heaps has been observed by Fredman et al.~\cite{pairing}. If we view a heap as a binary tree (interpreting ``leftmost-child'' and ``right-sibling'' pointers as ``left-child'' and ``right-child'', see e.g.\ \cite[\S\,2.3.2]{Knuth}), the re-structuring of pairing heaps after an extract-min operation resembles the re-adjustment of splay trees during an access. 
Using this observation, Fredman et al.\ adapt the splay tree potential function~\cite{ST85} to show that pairing heap operations take logarithmic time.\footnote{We remark that Fredman et al.'s analogy between pairing heaps and splay trees is far from trivial; to maintain it, the identities of nodes may need to be permuted, and some left-right children pairs may need to be swapped.} Despite the productive use of this connection, it seems rather specific to splay trees and pairing heaps. Given the intriguing similarity between self-adjusting data structures, the following question suggests itself.

\vspace{0.1in}
\textit{
\textbf{Question 2.} Is there a fundamental, general correspondence between self-adjusting BST and self-adjusting heap data structures?
}
\vspace{0.1in}

\subsection{Our results}

We propose a general model of heaps, which we call the \emph{stable heap model}.\footnote{By \emph{model} we understand the description of the pointer-based structure and the set of primitive operations that may be used to implement the priority queue.} This model allows us to approach both \emph{\textbf{Question~1}} and \emph{\textbf{Question~2}} in surprising ways. The model aims to capture self-adjusting heap algorithms, and is somewhat similar to the ``pure heap'' model of Iacono and \"Ozkan~\cite{IaconoOzkan, IaconoOzkan2}. 

The new element of our model is the \emph{stable link} operation that replaces the \emph{link} operation of existing heap models. Whereas the standard link operation makes the larger of the two linked items the leftmost child of the smaller, our stable link makes the larger item either the leftmost or the rightmost child of the smaller (see \ref{fig1}). More precisely, if $x$ and $y$ are linked, where $x$ is the left neighbor of $y$, then, a stable link operation makes $x$ the \emph{leftmost} child of $y$, or it makes $y$ the \emph{rightmost} child of $x$, respecting the (min-)heap order condition.

An intuitive a priori reason for stable links is that they ``better preserve'' the original left-to-right ordering of keys, which we expect to make algorithms that use stable links more amenable to instance-specific analysis.\footnote{The name is inspired by the idea of ``stable sorting'' (see e.g.\ \cite[\S\,5]{Knuth3}). Stable sorting preserves the original ordering of key-pairs that are \emph{equal}. As seen later, stable linking preserves the original ordering of key-pairs that are \emph{incomparable}, according to our current knowledge, i.e.\ in the partial order described by the heap.} A second justification is that the lower bounds known in existing heap models~\cite{FredmanLB, IaconoOzkan, IaconoOzkan2} crucially exploit that the links are not stable, raising the possibility that an algorithm that uses stable links may circumvent the existing lower bounds.   
The main motivation, however, is that with stable links, the heap model turns out to be deeply connected to the BST model. 

\paragraph{General transformation.}
Our connection addresses \emph{\textbf{Question~2}} as follows. We show that within the stable heap model, \emph{every} algorithm $\mathcal{A}$ for heap re-structuring corresponds to some algorithm $\mathcal{B}$ for BST re-arrangement. This is the first general connection between \emph{models} of heaps and BSTs. The costs of the two algorithms in \emph{sorting-mode} may differ only by a constant factor, when executing ``dual'' sequences of operations. By \emph{sorting-mode} execution we mean the following (see \ref{fig2__} for illustration). 

\noindent\textit{Sorting-mode for a heap:} A sequence of $n$ keys $(x_1, \dots, x_n)$ is inserted into an initially empty heap, followed by $n$ extract-min operations. For simplicity, assume that $\{x_1, \dots, x_n\} = [n]$. The execution can be interpreted as sorting the permutation $(x_1, \dots, x_n)$ via \emph{selection-sort}, using a particular heap-based method for selecting the minimum. 

\noindent\textit{Sorting-mode for a BST:} A sequence of $n$ keys $(x_1, \dots, x_n)$ is inserted into an originally empty BST. Again, assume $\{x_1, \dots, x_n\} = [n]$. The execution can be interpreted as sorting $(x_1, \dots, x_n)$ via \emph{insertion-sort}, using a particular BST-based method for insertion. (The sorted keys can be read out in an $O(n)$-time traversal of the final tree.)

The connection between $\mathcal{A}$ and $\mathcal{B}$ is the following. The number of elementary operations performed by $\mathcal{A}$ when sorting $X = (x_1, \dots, x_n)$ equals (up to a constant factor) the number of elementary operations performed by $\mathcal{B}$ when sorting the \emph{inverse permutation} $X'$, i.e.\ the permutation $X' = (y_1, \dots, y_n)$, where $y_i = j$ iff $x_j = i$. (See \ref{fig2__} for an illustration.) 

We say that the roles of \emph{time} and \emph{key-space} are switched between the two problems. The time of insertion, a.k.a.\ the \emph{index} of a key in the heap input is mapped to the \emph{rank} of a key in the BST input. Similarly, the \emph{rank} of a key in the heap input is mapped to the time of insertion in the BST input.

\begin{figure}
	\begin{center}
\includegraphics[width=0.99\textwidth]{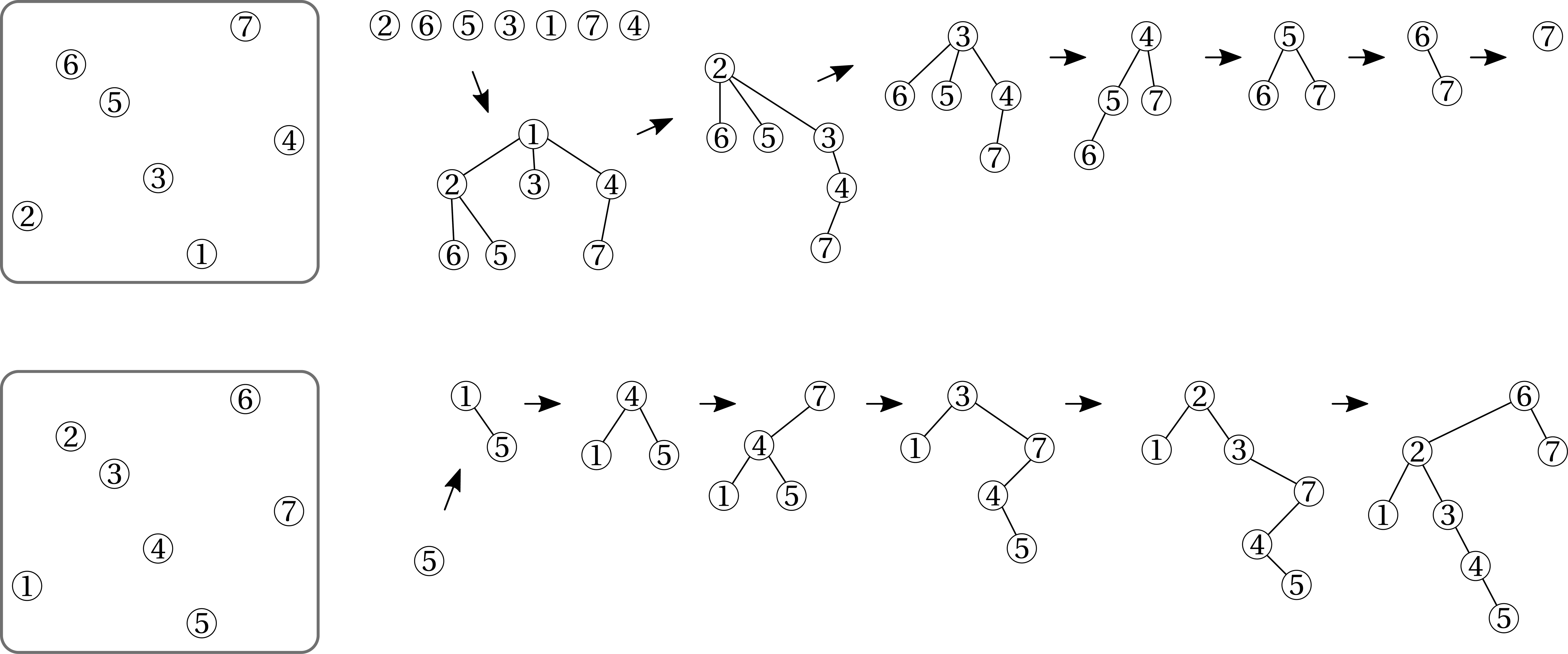}
	\end{center}
\caption{Sorting-mode execution for heaps and BSTs. Input permutation $P = (2,6,5,3,1,7,4)$ and its inverse $P' = (5,1,4,7,3,2,6)$. Above left, $P$ shown with values on $y$-axis, below left, $P'$ shown with values on $x$-axis. Above: sequence of extract-mins in a heap, following insertion sequence $P$, using an unspecified stable heap algorithm. Below, insertion sequence $P'$ in an initially empty BST, using the ``move-to-root'' restructuring (the inserted items are rotated to the root).
\label{fig2__}}
\end{figure}

We remark that the described duality between heaps and trees is quite subtle. As we execute the heap algorithm, the individual rotations that make up the BST execution are revealed in an order that is neither the temporal, nor the spatial order; the order in which the BST execution is revealed depends on the internal structure of the input.  
It is instructive to compare our connection with the connection of Fredman et al.\ between pairing heaps and splay trees. There, individual link operations map (essentially) to individual rotations; the pairing heap and the splay tree evolve in parallel and the correspondence between them is maintained throughout the execution. This is emphatically not the case in our connection. In fact, we do not even map individual operations such as \emph{search}, \emph{extract-min}, or \emph{insert} to each other. Instead, we ``globally'' match the execution trace of a sequence of operations in a heap with the execution trace in a BST for the ``dual'' sequence of operations. An intermediate state of one structure, in our connection, corresponds to an \emph{abstract state} of the other structure, where some decisions pertaining to future operations have been comitted to, others are still left undecided.

\paragraph{Smooth heaps.}
	Perhaps the most interesting consequence of our connection is a new, simple heap algorithm that we call the \emph{smooth heap}. We show it to be the heap-counterpart of Greedy, which is conjectured to be instance-optimal for BSTs. 
	Interpreting the proven guarantees for Greedy, we obtain a logarithmic bound on the cost of extract-min in smooth heaps, as well as a number of instance-specific upper bounds, not previously known for any heap implementation (e.g.\ we show that smooth heaps adapt to locality of reference, and to pattern-avoidance in the input). Assuming the conjectured optimality of Greedy, smooth heaps are also optimal in sorting-mode among all stable heap algorithms. Key to our result is a new, non-deterministic interpretation of Greedy which may be of independent interest.

The intriguing aspect of this connection is that although Greedy is rather complicated to implement as an online BST algorithm, the smooth heap is not; it is comparable in simplicity to pairing heaps. There are \emph{several equivalent descriptions} of the smooth heap. In one view, it appears, roughly speaking, as a pairing heap equipped with a ``power-of-two-choices'' type of heuristic. In another view, it appears as a structure built of nested treaps. The smooth heap is self-adjusting (no auxiliary information), and uses only the standard pointer structure, thus it partially addresses \emph{\textbf{Question~1}}.

In the following, we briefly describe the smooth heap data structure. (See \S\,\oldref{sec:smooth_desc} for more details and variants.) In this section, we view the smooth heap as a single multiway (min-)heap (i.e.\ the key of each node is greater than its parent). The crucial operation is extract-min, implemented as follows. After deleting the root, we are left with a list of nodes that need to be consolidated into a single tree, using link operations.  
We proceed by repeatedly finding a \emph{local maximum} (i.e.\ an item that is larger than its neighboring siblings). We then link the local maximum with \emph{the larger} of its two neighbors (or with its only neighbor, if it is the leftmost or the rightmost item in the list). This operation, in effect, removes the local maximum, ``smoothing'' the sequence of items; this gives the name of the data structure. Linking with the larger of the two neighbors can be seen as a locally greedy choice: as we are moving towards the sorted order, intuitively the smaller the rank-difference between linked items, the more progress we make.\footnote{There are limits to this intuition, and we do not claim smooth heaps to be exactly, globally optimal, see~\ref{app:nonoptimal}.} 

A high-level description of the algorithm is given in \ref{fig2_}. More explicit descriptions are given in \S\,\oldref{sec:smooth_desc} and \ref{smooth_heap_ps}.

\begin{figure}[H]
\begin{center}
\parbox{5.5in}{
\begin{algorithm}[H]
\newcommand{\llWhile}[2]{{\let\par\relax\lWhile{#1}{#2}}}
\newcommand{\llIf}[2]{{\let\par\relax\lIf{#1}{#2}}}
\newcommand{\llElse}[1]{{\let\par\relax\lElse{#1}}}
\DontPrintSemicolon
 \SetAlgoLined
\NoCaptionOfAlgo
\vspace{0.03in}
\KwIn{a list of siblings $X$}
\vspace{0.03in}
\textbf{\texttt{while}} there is an item $x$ in $X$ larger than its neighbors: \\
\quad \textbf{\texttt{let}} $y$ be the largest of the neighbors of $x$.\\
\quad \textbf{\texttt{link}} $x$ and $y$.
\end{algorithm}

\caption{Smooth heap re-structuring for extract-min. \label{fig2_}}
}
\end{center}
\end{figure}

Finding local maxima and linking them with a neighbor can be done in a single left-to-right pass (similarly to the first pass of pairing heaps). The remaining top-level items after this pass are in sorted order, they can therefore be collected in a second, right-to-left accumulation pass (again, similar to the second pass of pairing heaps, except that here, no more comparisons need to be made). The total number of comparisons is easily seen to be at most twice the number of link operations performed. See \ref{fig3} for an illustration. We remark that the link operations used in the smooth heap are all stable links.

\begin{figure}[H]
	\begin{center}
\includegraphics[width=0.99\textwidth]{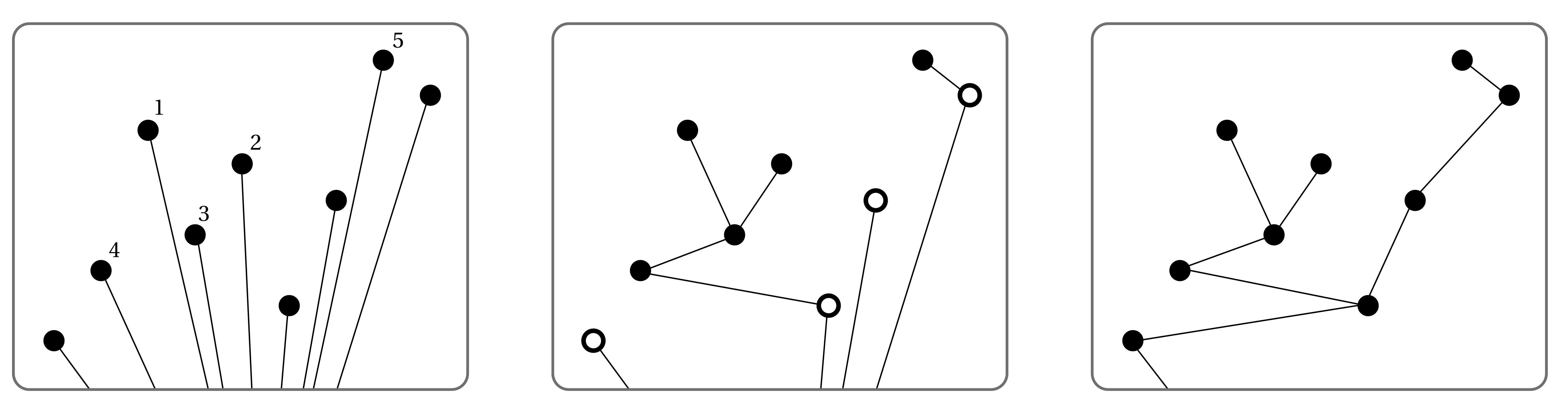}
	\end{center}
\caption{Smooth heap re-structuring. Top-level siblings, left-to-right with keys $(1,3,7,4,6,2,5,9,8)$, key values shown on $y$-axis, subtrees of nodes not shown. (i) Initial state. Small numbers indicate order of linking. (ii) State after left-to-right smoothing round. Remaining top-level nodes appear hollow. (iii) State after final, right-to-left round. Lowest node is new root after extract-min.
\label{fig3}}
\end{figure}

As our focus is on sorting-mode, we postpone the description of other operations to \S\,\oldref{sec:smooth_desc}. It remains an intriguing open question whether the smooth heap (or any other self-adjusting heap) can match Fibonacci heaps in the asymptotic cost of all operations (i.e.\ the remainder of \emph{\textbf{Question~1}}). The arguments that rule out constant-time decrease-key in the case of pairing heaps do not immediately apply to stable heaps. It remains to be seen whether this is a mere technicality, or if an efficient decrease-key can in fact be added to smooth heaps. There are several natural ways to implement this operation, but their analysis does not appear easy. 

\paragraph{Consequences of the transformation.}
Several standard heap algorithms can be transferred to the stable heap model, including the variants of pairing heaps described in \cite{pairing} (see \S\,\oldref{sub:heap model}). Through our duality between heap and tree algorithms, these new heap algorithms also yield new \emph{offline} BST algorithms that have not been previously described. 
(In fact, for every stable heap algorithm we obtain a \emph{family} of offline BST algorithms, with cost at most a constant factor larger; e.g.\ we obtain a set of (non-trivially) different BST executions that are $O(1)$-competitive with Greedy.)
In the other direction, we can transfer strong instance-specific \emph{lower bounds}~\cite{Wilber, DHIKP09} from BSTs to heaps.  
We thus obtain  
instance-specific lower bounds for all stable heap algorithms in sorting-mode, which we see as a necessary step towards a theory of dynamic optimality for heaps.\footnote{The lack of such a theory was noted e.g.\ by Pettie~\cite{Pettie,deque_Pet08}: \emph{``Despite the connection
between splay trees and pairing heaps, there is no obvious analogue
of dynamic optimality for priority queues.''}}

\paragraph{Adaptive sorting.} 

Our connection between self-adjusting BSTs and heaps can be interpreted as an equivalence between broad classes of selection-sort and insertion-sort algorithms on inverted input permutations. Sorting has been extensively studied, not just in the worst-case, but also in an instance-specific sense. (Sorting with some ``pre-sortedness'' in the input is often called ``adaptive sorting''). As an entry point into the literature on adaptive sorting, we refer to \cite{MunroSpira, Mehlhorn79, Mannila, Levcopoulos89, Levcopoulos,estivill,  Moffat, Barbay1},\cite[\textsection\,5.3]{Knuth3}, and references therein.

The proven instance-specific properties of Splay and GreedyFuture subsume many of the structural properties in the adaptive sorting literature (\ref{app_sort}). Until now, however, there has been no easy way to implement GreedyFuture as an insertion-sort algorithm. (As mentioned, the online Greedy algorithm is rather complicated.) By using smooth heaps, we obtain an easy-to-implement selection-sort algorithm\footnote{We remark that \emph{every} stable heap algorithm yields a \emph{stable sorting} algorithm in the classical sense, assuming that ties are broken according to the ``left is smaller'' rule.} that inherits all the instance-specific properties of Greedy (on the inverse of the input permutation).

As a concrete example, we mention sorting sequences with an avoided permutation pattern (pattern-avoidance is a well-studied combinatorial property of sequences, see e.g.~\cite{knuth68,tarjan_sorting, pratt_queues, Kitaev, MarxPattern, Newman}). Sorting pattern-avoiding permutations was studied by Arthur~\cite{Arthur}. As a corollary of our pattern-avoiding bounds for Greedy~\cite{FOCS15}, we obtain a practical sorting algorithm that sorts permutations with an arbitrary fixed avoided pattern in quasi-linear time (the exact bound involves the inverse Ackerman function), improving the earlier bound of~\cite{Arthur}. Furthermore, in contrast to~\cite{Arthur}, our algorithm need not know the avoided pattern in advance; pattern-avoidance of the input is used only in the analysis of Greedy~\cite{FOCS15}. (We use here the easy fact that if a permutation avoids a pattern, then the inverse permutation also avoids a certain pattern of the same size.)

\subsection{Related work and structure of the paper} 
A ``dual'' view of binary search trees and heaps brings to mind the Cartesian tree data structure~\cite{Vuillemin}, also known as treap~\cite{SeidelA96}. A treap is a binary tree built over pairs of values, respecting the in-order condition with respect to the first, and the heap-order condition with respect to the second entry. Assuming distinct priorities, the treap is easily seen to be unique.  
Swapping the roles of the two entries yields a ``dual treap''.

The operation performed by smooth heaps on a list of items can be seen as a linear-time ``on-the-fly treapification'', where the left-to-right indices of items play the role of key-values. Such a description of the smooth heap is reminiscent of the Cartesian tree sorting algorithm of Levcopoulos and Petersson~\cite{Levcopoulos89}. There too, a treap of the input permutation is built. Afterwards, a secondary heap structure (e.g.\ a binary heap or a Fibonacci-heap) is used to store the children of the root; in this secondary heap, the minimum is repeatedly deleted and its children inserted, resulting in a sorted deletion-order. Cartesian tree sorting was not proposed as a general-purpose data structure, but it could be turned into one, by implementing insert and other operations. (A natural way to insert is perhaps to add a new key into the auxiliary heap.)  

Despite the apparent similarity, the behavior and performance of smooth heaps and Cartesian tree sorting are different. Smooth heaps do not require a secondary heap structure, instead, they (implicitly) build a new treap from the top-level list of nodes after every deletion of the minimum.  
The running time of Cartesian tree sorting can be expressed in closed form. It was shown by Eppstein that this quantity is subsumed by known instance-specific bounds for splay trees~\cite{Eppstein-blog}.  
It is easy to construct examples where the running time of sorting by smooth heaps is $O(n)$, whereas the running time of Cartesian tree sorting is $\Omega(n \log{n})$. An example where Cartesian tree sorting is faster than sorting by smooth heaps would disprove the conjectured dynamic optimality of Greedy (\ref{app_cartesian}).

In \S\,\oldref{sec:bst_model} we define the self-adjusting BST model that we use in the rest of the paper.  
In \S\,\oldref{sub:heap model} we describe the stable heap model and we adapt known algorithms to this model. 
In \S\,\oldref{sec:smooth_desc} we describe the smooth heap, our new heap data structure. 
In \S\,\oldref{sec:produce sat set} we present the main connection between the two models, and describe some of its immediate consequences.  
As a key ingredient of our results, we present a nondeterministic description of \greedy in \S\,\oldref{sec:non-det greedy}. 
We then give the full proofs in \S\,\oldref{sec:proofs}.  
In \S\,\oldref{sec:open} we conclude with an extensive list of open questions and possible further directions.

\section{BST model}\label{sec:bst_model}

We adopt a standard model of binary search trees (BSTs), see e.g.\ \cite{Wilber, Luc88, Tango, DHIKP09}, which can be seen as a restricted pointer machine model. Each node of the BST is associated with a key, respecting the in-order condition (each item is larger than those in its left subtree and smaller than those in its right subtree). Dictionary operations are performed using a cursor that points, at the beginning of every operation, to the root of the tree. The elementary (unit-cost) operations allowed are moving the cursor to the parent, left-child, or right-child of the current node, and performing a rotation on the edge between the node at the cursor and its parent (\ref{fig1}). 

In this paper we consider only search and insert operations; we further assume that all searches are successful. Throughout, we assume that all key-values are distinct.\footnote{See~\cite{DHIKP09,FOCS15} for discussion of these standard assumptions.} 

Search and insert are implemented in the standard way. For search, it is required that the cursor visits, at some point during the operation, the node with the given key. For insert, the cursor must visit both the predecessor and the successor of the new key (or only one of them, in case the key is the new minimum or maximum). A node containing the new key is then attached as the right child of its predecessor, or as the left child of its successor (whichever slot is free). In addition, during both search and insert, a BST algorithm may perform arbitrary re-arrangement of the tree, using rotations at the cursor. The cost of a search or insert operation is the number of nodes that are \emph{visited} by the cursor (``touched'') during the operation. Observe that the touched nodes form a subtree containing the root. This cost model is easily seen to be equivalent up-to constant factors, with most other reasonable cost models~\cite{Luc88,Wilber, DHIKP09}.

In this paper we only consider ``offline'' BST algorithms, i.e.\ we assume that the entire sequence of insert and search operations is known in advance.\footnote{For online algorithms, operations are revealed one by one.} 
We mainly consider two possible modes of operation: a \emph{search-only mode}, in which, starting from some initial tree a sequence of $n$ distinct searches are performed, and an \emph{insert-only mode}, in which, starting from an empty tree, a sequence of $n$ distinct keys are inserted. (We also call the insert-only mode the \emph{sorting-mode}, as it can be seen as a particular implementation of insertion-sort; the $n$ keys can be read out in sorted order by the in-order traversal of the resulting tree.) 

As the operations are performed on distinct keys (in both search-only and insert-only modes), we view a sequence of operations of length $n$ as a permutation over $[n]$. The \emph{execution trace} of a BST algorithm $\cA$ serving the sequence $X$, denoted $\cA(X)$, is the set of touched keys for all operations. That is, $\coord i j \in \cA(X)$, if algorithm $\cA$ touches node $i$ when executing the $j$-th operation of $X$ (``at time $j$''). The cost of the execution is $|\cA(X)|$, i.e.\ the total number of nodes touched at all times. 
For every permutation $X$ over $[n]$, we denote the \emph{offline optimum} in search-only mode and insert-only mode as $\opt_{\bst}(X)$ and
$\opt_{\bstins}(X)$ respectively, where $\opt_{\bst}(X)=\min_{\cA}\{|\cA(X)|\}$ when $X$ is treated as a sequence of search operations, and similarly $\opt_{\bstins}(X)=\min_{\cA}\{|\cA(X)|\}$ when $X$ is treated as a sequence of insert operations. 

Demaine et al.~\cite{DHIKP09} introduced an elegant geometric view for the study of BST algorithms, which we use throughout the paper. 
In the remainder of the section we review some of their definitions and results, and we slightly extend the model, to handle insertion-sequences.

\paragraph{Satisfied point sets.}

Let $P\subseteq\mathbb{Z}\times\mathbb{Z}$ be a finite set of points.
For every point $p\in P$, we write $p=\coord{p.x}{p.y}$ where $p.x$
is the $x$-coordinate of $p$ (the first coordinate) and $p.y$ is
the $y$-coordinate of $p$ (the second coordinate). We denote $P_{x\le i}=\{p\in P\mid p.x\le i\}$
and $P_{y\le i}=\{p\in P\mid p.y\le i\}$, and similarly,
$P_{x=i}$, $P_{x<i}$, and so on. We call $P_{x=i}$ and $P_{y=i}$
the $i$-th column and the $i$-th row of $P$ respectively. We say
that $P$ is a \emph{permutation} (of size $n$) if $|P_{x=i}|=|P_{y=i}|=1$ iff $i \in [n]$.

Given two points $p, q \in P$, we denote as $\square_{pq} \subseteq\mathbb{Z}\times\mathbb{Z}$ the rectangle (possibly degenerate) whose opposite corners are $p$ and $q$. We say that $\square_{pq}$ is a \emph{satisfied rectangle} if $\square_{pq}$ contains another point $r\in P\setminus\{p,q\}$ (where $r$ can be at a corner or a border of $\square_{pq}$), or if $p$ and $q$ are aligned vertically or horizontally.  

A point set $P$ is \emph{satisfied} if, for every pair of points $p,q\in P$, the rectangle $\square_{pq}$ is satisfied.

\begin{defn}
	[Satisfied superset problem]Given a set of points $P$, find a satisfied set $Q \supseteq P$.\label{def:sat problem}
\end{defn}
Let $\cA_{\sat}$ be some algorithm for solving the satisfied superset
problem. Given a point set $P$, we denote by $\cA_{\sat}(P)$ the output of $\cA_{\sat}$ for input $P$, i.e.\ a satisfied superset of $P$. The \emph{cost }of $\cA_{\sat}$ on $P$ is defined to be $|\cA_{\sat}(P)|$. 
Let $\opt_{\sat}(P)$ be the size of the smallest satisfied superset of $P$. 

The following definition is new.

\begin{defn}
	[Insertion-compatible superset] Given a set of points $P$, a superset $Q \supseteq P$ is \emph{insertion-compatible}, if for every $i \in [n]$ we have $\min\{p.y \mid p \in P_{x=i}\} = \min\{q.y \mid q \in Q_{x=i}\}$. In words, $Q$ does not add a point \emph{below} a point of $P$.
\end{defn}

We also consider the variant of the satisfied superset problem where the goal is to find an \emph{insertion-compatible} satisfied superset of the input. 
For every point set $P$, let $\opt_{\satins}(P)$ be the size of the smallest insertion-compatible satisfied superset of $P$. 
Clearly, $\opt_{\satins}(P) \ge \opt_{\sat}(P)$.

Let $X=(x_{1},\dots,x_{n}) \in [n]^n$ be a permutation, i.e.\ $x_{i}\neq x_{j}$ for every $i\neq j$. 
We call $P^{X}=\{\coord{x_{i}}i\mid i \in [n] \}$ the \emph{point set of $X$}. 

Given $X$, we also consider the \emph{inverse} permutation $X'$. (The $i$-th element of $X$ is $x_{i}$ iff the $x_{i}$-th element of $X'$ is $i$.) Treating $P^{X}$ as an $n$-by-$n$
matrix, we can define the transpose $(P^{X})^{'}$. Observe that $(P^{X})^{'}=P^{X^{'}}$. 
We also consider the \emph{reverse} permutation $X^r$. (The $i$-th element of $X^r$ is $x_{n-i+1}$.)
For a point set $P \subseteq [n]\times[n]$, its reverse $P^r$ is such that 
$P^r_{y=i} = \{ \coord{x}{i} \mid \coord{x}{n-i+1} \in P_{y=n-i+1} \}$ for each $i$. Note that $(P^X)^r = (P^{X^r})$. 
Throughout this paper, we usually use $X$ to denote a sequence of numbers, and use $P$ or $Q$ to denote a point set.

Demaine et al.~\cite{DHIKP09} show a precise equivalence between execution traces of BST algorithms and satisfied sets.

\begin{thm}[Geometry of BSTs, search-only \cite{DHIKP09}]\label{thm:geo bst}
	Let $X \in [n]^n$ be an arbitrary permutation. A set $Q \subset [n]^2$ is a satisfied superset of $P^X$ iff there is an offline BST algorithm $\cA_{\bst}$ in search-only mode with some initial tree 
	such that $\cA_{\bst}(X) = Q$.
\end{thm}

\begin{cor}
	For every permutation $X \in [n]^n$, $\opt_{\bst}(X) = \opt_{\bst}(X') = \opt_{\bst}(X^r)$.
	\label{cor:bst symmetric}
\end{cor}
\begin{proof}
	By \Cref{thm:geo bst}, $\opt_{\bst}(X) = \opt_{\sat}(P^X)$. It is clear that $\opt_{\sat}(P^X) = \opt_{\sat}(P^{X'}) = \opt_{\sat}(P^{X^r})$. Another application of \Cref{thm:geo bst} implies the claim. \qedd
\end{proof}

The following theorem is new. Its proof closely follows the proof of \ref{thm:geo bst} from \cite{DHIKP09}. For completeness, we give it in \ref{app_insert}.

\begin{thm}
	[Geometry of BSTs, insertion-compatible]\label{thm:geo ins}
	Let $X \in [n]^n$ be an arbitrary permutation. A set $Q \subset [n]^2$ is an insertion-compatible satisfied superset of $P^X$ iff there is an offline BST algorithm $\cA_{\bst}$ in insert-only mode 
	such that $\cA_{\bst}(X) = Q$.
\end{thm}

\begin{rem}\label{rem:bst insert stronger}
	By \Cref{thm:geo bst,thm:geo ins}, given any BST algorithm $\cA$ in insert-only mode (a.k.a.\ sorting-mode), there is a BST algortihm $\cA'$ in search-only mode, such that $|\cA'(X)| \le |\cA(X)|$ holds for all permutations $X$.\footnote{
		With the same reasoning, we can even show that there is a BST algorithm $\cA''$ that handles a sequence of arbitrarily intermixed search and insert operations, where $|\cA''(X)| \le |\cA(X)|$ for every permutation $X$.}
	That is, an upper bound for the insert-only mode is also an upper bound for the search-only mode. 
	It is not known whether the reverse statement holds. In this paper we always prove a stronger statement by upper bounding the cost of BST algorithms in insert-only mode.
\end{rem} 

A BST algorithm $\cA$ is $c$-competitive (w.r.\ to the offline optimal BST algorithm) if
for all permutations $X \in [n]^n$, we have $|\cA(X)| \leq c \cdot \opt_{\bst}(X)$.
If $\cA$ is $O(1)$-competitive, we also say that $\cA$ is \emph{instance-optimal}.

\paragraph{\greedyfuture~and $\protect\greedy$.}

Perhaps the most natural algorithm for the satisfied superset problem is \greedy, a straightforward geometric sweepline algorithm.

\begin{defn}
	[\greedy \cite{DHIKP09}]\label{def:greedy}Given a set $P\subseteq[n]\times[n]$,
	\greedy\,works in $n$ steps as follows. Initially, $Q=P$. At the
	$i$-th step, if there is a point $p\in P_{y=i}$ and $q\in Q_{y<i}$
	where $\square_{pq}$ is not satisfied, it adds a point $r=\coord{q.x}i$
	to $Q$ (so now $\square_{pq}$ is satisfied). The output of \greedy\,is denoted $\greedy(P)=Q$.
\end{defn}

Observe that \greedy\,always produces an insertion-compatible satisfied superset.
It is known that \greedy\,is equivalent, in the sense of \ref{thm:geo bst,thm:geo ins},
to the offline BST algorithm \greedyfuture. 
\begin{thm}
	[\cite{DHIKP09}]For all permutations $X$ it holds that
 \greedyfuture$(X)=\greedy(P^{X})$.\label{thm:greedy and greedyfuture}
\end{thm}
Similarly to Splay~\cite{ST85}, \greedyfuture is conjectured
to be instance-optimal.
\begin{conjecture}[\cite{DHIKP09,Mun00,Luc88}]
	For all permutations $X \in [n]^n$ it holds that\\
 \greedyfuture$(X)=O(\opt_{\bst}(X))$.\label{conj:greedy opt}
\end{conjecture}

Demaine et al.\ also define \emph{online} satisfied superset algorithms and show a stronger form of \ref{thm:geo bst} that applies to online algorithms. In this way, they obtain an online BST algorithm competitive with GreedyFuture. This online equivalence is not our concern in the current paper.

\newpage
\section{Stable heap model}
\label{sub:heap model}

In this section we introduce a family of heap data structures that include close variants of several of the previously proposed heaps. The key feature of our model is the stable link operation that replaces the standard (unstable) link (\ref{fig1}). We call the new model the stable heap model.

\subsection{Description of the model}
\paragraph{Structure.} A heap in the stable heap model is organized as a \emph{forest} of multiway min-heaps.\footnote{We opt for a forest-based representation in order to simplify the implementation of insert, meld, and decrease-key operations, which can now be performed lazily. It is not hard to change our model such that it is based on a single heap. Fredman~\cite{FredmanTransform} describes a general transformation that turns various heap implementations into the forest-of-heaps representation.} Every node has an associated key (we refer interchangeably to a node and its key), and every non-root key is larger than its parent. For simplicity, we assume keys to be distinct. Internally, a forest of multiway heaps is represented in the standard way by storing the children of every node as a linked list, see e.g.\ \cite[\S\,2.3.2]{Knuth}. The collection of top-level roots in the forest is also stored in a linked list. We assume that all lists are \emph{doubly-linked}, providing the pointers \NEXT~and \PREV~between neighboring siblings, with the appropriate markers at the two ends of the list. We further assume that every node has a pointer to its \emph{leftmost} and to its \emph{rightmost} child, and that we have a global pointer to the \emph{leftmost} and \emph{rightmost} among the top-level roots. The presence of parent pointers is not assumed (\ref{fig_tree_struct}).

\begin{figure}[H]
	\begin{center}
\includegraphics[width=0.4\textwidth]{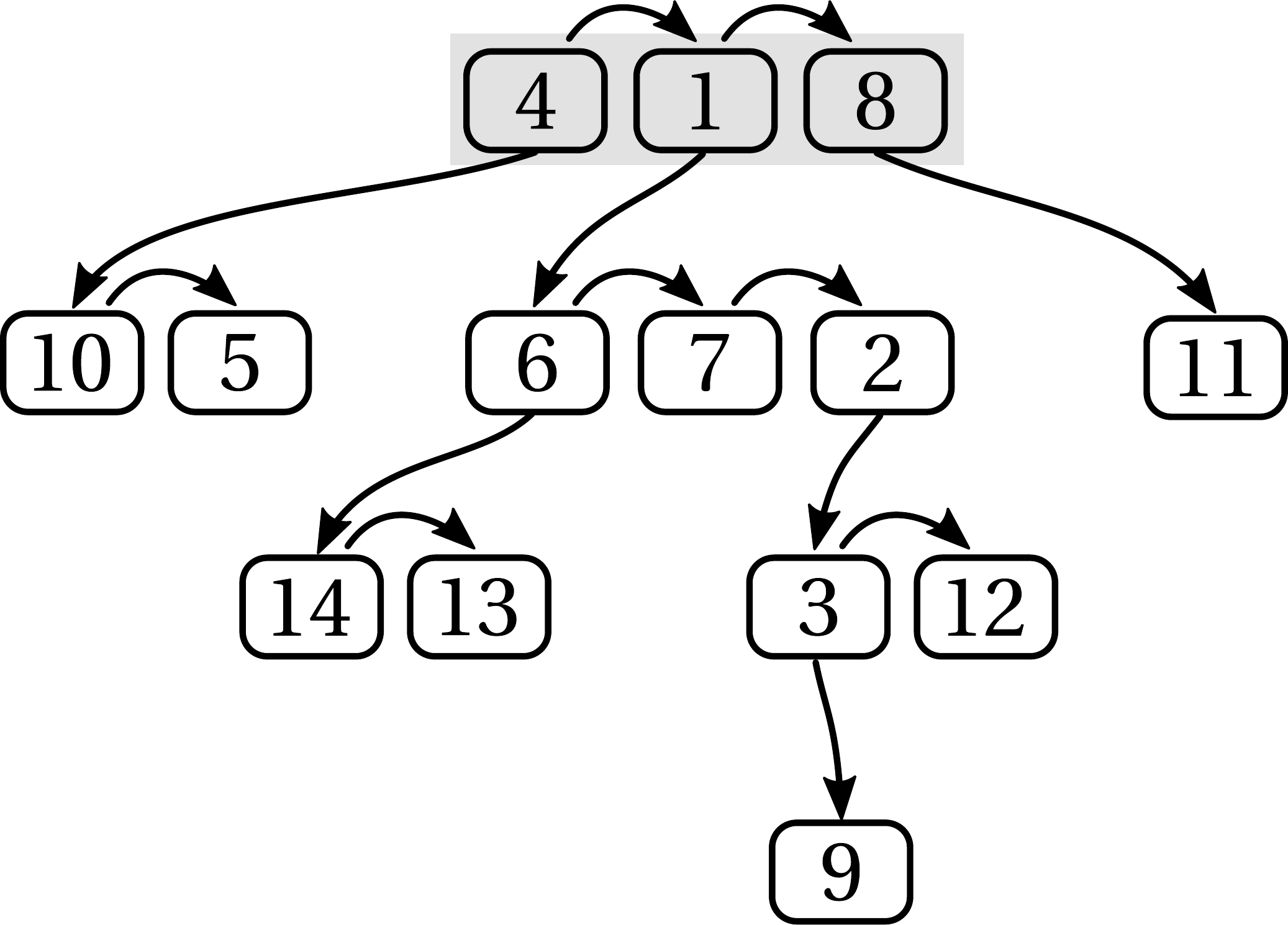}
	\end{center}
\caption{Forest-of-heaps structure with lists of siblings. Arrows show \emph{leftmost} and \emph{next} pointers. 
	(The pointers \emph{rightmost} and \emph{prev} are not shown.)
	Top-level list appears shaded. 
\label{fig_tree_struct}}
\end{figure}

In practice, to make the heap more space-efficient, it is desirable to have only two pointers per node, instead of four, as described. To this effect, algorithms in the stable heap model can also be implemented using only one of \PREV~or \NEXT. Furthermore, one of the pointers \LM~and \RM~can be simulated in constant time by making the lists circularly linked. The discussion of similar issues of implementation by Fredman et al.\ \cite{pairing} for pairing heaps also applies to our model. 

\paragraph{Stable link.} The defining feature of the stable heap model is the stable link operation (\ref{fig1}). We denote by $\link(x)$ the operation of linking $x$ and its right sibling $y = x.\NEXT$ with a stable link. If the key of $x$ is smaller than the key of $y$, then $y$ becomes the rightmost child of $x$. Otherwise, $x$ becomes the leftmost child of $y$. (Contrast this again with the standard, ``unstable'', link operation where the larger item always becomes the leftmost child of the smaller.) Clearly, all necessary pointer changes can be performed in constant time. With a careful implementation, the stable link has, in practice, similar cost as the standard link. 

In the following we assume (without explicitly describing the low-level details) that the link operations correctly update all pointers to reflect the structure-change.

\paragraph{Operations.} The heap operations \emph{makeheap}, \emph{insert}, \emph{decrease-key}, and \emph{meld} can be implemented in a straightforward way, identically to pairing heaps, apart from our use of stable links instead of standard links. 

Makeheap creates a new, empty heap with the structure described above. Insert creates a singleton root with the new key, and appends it to the list of top-level roots of the heap. Melding two heaps concatenates their top-level root-lists. Decrease-key detaches the tree rooted at the node whose key is decreased and appends it to the top-level root list.\footnote{The easiest choice is to append the node at one of the ends of the top-level list. An alternative implementation (more in the spirit of stable heaps) would be to ``sift-up'' the node with decreased key, placing it between its successor and predecessor in the top-level list, according to the insertion-times.} All four operations require constant number of pointer moves and pointer changes. 

We now describe extract-min, the most complex operation. To find the minimum, the roots in the top-level list are consolidated into a single tree, through a sequence of stable links. In our model \emph{only neighboring siblings} can be linked. Every link operation removes one root from the top-level list (the one with the larger key of the two linked). Thus, the number of links performed during extract-min is exactly one less than the initial size of the top-level list. After a single top-level root remains (the minimum), it is deleted, and its list of children becomes the new top-level list. 

Algorithms in the stable heap model differ only in the order in which they link items during the extract-min operation. At the start of extract-min, a cursor is assumed to point to the leftmost node in the top-level list. Algorithms are allowed to move the cursor to the right or to the left, make comparisons on keys of visited nodes, and perform stable links at the cursor.

We call algorithms in the stable heap model \emph{stable heap algorithms}, and we call the entire structure a \emph{stable heap data structure}.

\paragraph{Cost model.}
We define the cost of operations to be the \emph{link-only cost}, i.e.\ the number of stable link operations performed. Thus, the worst-case cost of operations other than extract-min is constant, whereas the cost of extract-min equals the number of top-level nodes before the operation.

It may seem unrealistic to ignore the cost of comparisons and pointer moves. We justify this choice as follows: (1) In all algorithms that we consider, the number of pointer moves and comparisons will in fact be proportional to the link-only cost, therefore, the use of link-only cost is accurate for these algorithms. (2) We can prove \emph{lower bounds} even for the link-only cost, that is, our lower bounds hold even if pointer moves and comparisons are free. In particular, our lower bounds hold even if the outcomes of all comparisons are known in advance.\footnote{This assumption can be seen as the ``dual'' of the BST assumption of knowing in advance the order of operations.} \footnote{It may seem that knowing the sorted order of keys, a stable heap algorithm can arrange the nodes into a path with a linear number of links, contradicting the information-theoretic lower bound for sorting. Observe however, that we are limited to linking neighboring siblings, which removes the contradiction.}

In the above description, we only link nodes at the top level. This is only for simplicity, and our lower bounds also apply to algorithms that can link siblings at any level of the heap.\footnote{The assumption to link only at the top level is analogous to the BST assumption of only rotating the search path. The restriction does not seem significant, since links at lower levels can be postponed to a later time, when the siblings reach the top level. One case where this may make a difference is if decrease-key operations are involved; performing links at lower levels may allow some items to ``travel together'' with the item whose key is decreased.}

\begin{figure}[H]
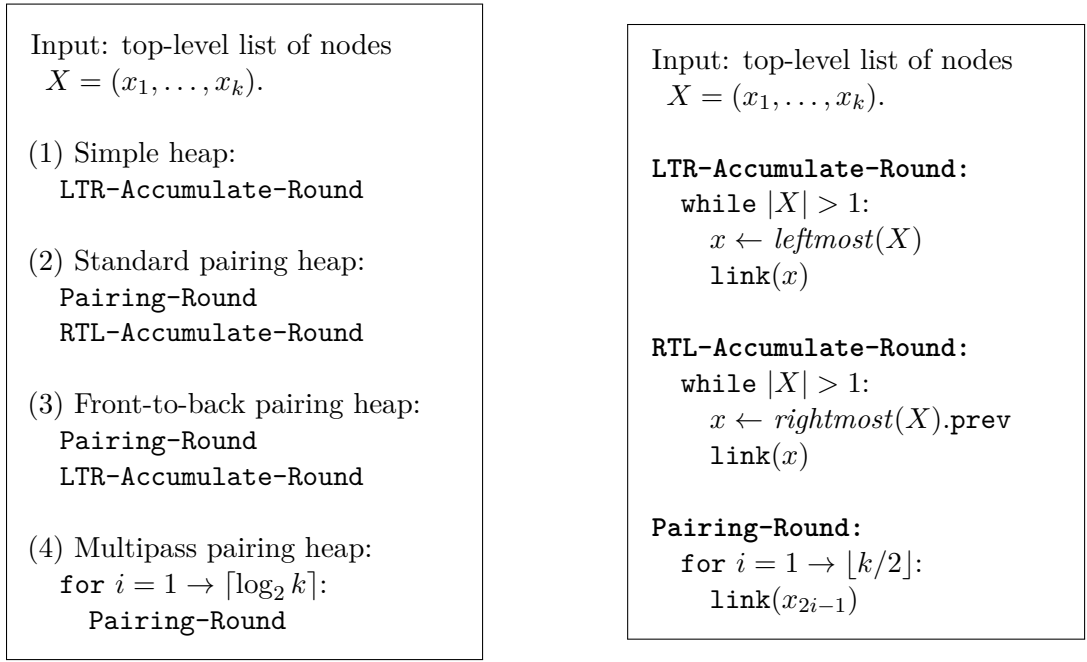

\begin{center}
\parbox{\columnwidth}{
  \parbox{2.85in}{
    \begin{algorithm}[H]
\begin{framed}
    \DontPrintSemicolon
    \SetAlgoLined
    \NoCaptionOfAlgo
    \vspace{0.03in}
    Input: top-level list of nodes $X = (x_1, \dots, x_k)$.\\
    \ \\

    {(1) Simple heap:}\\
    \quad {\texttt{LTR-Accumulate-Round}}\\
    \ \\
    {(2) Standard pairing heap:}\\
    \quad {\texttt{Pairing-Round}}\\
    \quad {\texttt{RTL-Accumulate-Round}}\\
    \ \\	
    {(3) Front-to-back pairing heap:}\\
    \quad {\texttt{Pairing-Round}}\\
    \quad {\texttt{LTR-Accumulate-Round}}\\
    \ \\
    {(4) Multipass pairing heap:}\\
    \quad {\texttt{for}} $i = 1  \rightarrow \lceil \log_2{k} \rceil$:\\
    \quad \quad {\texttt{Pairing-Round}}\\
\end{framed}
    \end{algorithm}
  } 
  \hspace{0.3in} 
  \parbox{2.75in}{
    \begin{algorithm}[H]
\begin{framed}
    \NoCaptionOfAlgo

    Input: top-level list of nodes $X = (x_1, \dots, x_k)$.\\
    \ \\
    \textbf{\texttt{LTR-Accumulate-Round:}}\\

\quad {\texttt{while}} $|X| > 1$:\\
    \quad \quad	$x \gets$ \textit{leftmost}$(X)$ \\
    \quad \quad $\link(x)$\\
    \ \\
    \textbf{\texttt{RTL-Accumulate-Round:}}\\
\quad {\texttt{while}} $|X| > 1$:\\
    \quad \quad $x \gets$ \textit{rightmost}$(X)$.\PREV \\
    \quad \quad $\link(x)$\\
    \ \\
    \textbf{\texttt{Pairing-Round:}}\\
    \quad {\texttt{for}} $i = 1 \rightarrow \lfloor k/2 \rfloor$:\\
    \quad \quad $\link(x_{2i-1})$
\end{framed}
    \end{algorithm}
  }

\caption{(left) Four heap algorithms. (1) A naive restructuring. (2)-(4) Variants of pairing heaps introduced by Fredman et al.~\cite{pairing}. (right) Subroutines implementing linking rounds on the list of nodes. The operation $\link(x)$ links the node $x$ with its right neighbor. To transform the algorithms into stable heap algorithms, $\link(x)$ needs to be implemented as stable link.\label{fig_heap_algos}}
}
\end{center}
\end{figure}

\paragraph{Sorting-mode.} In this paper we look at stable heap algorithms in \emph{sorting-mode} (see \S\,\oldref{sec:intro}), and analyse the amortized cost of smooth heap operations in this mode only. That is, we assume that a makeheap operation is followed by $n$ insert operations with distinct keys, and finally, by $n$ extract-min operations.\footnote{While sorting-mode is somewhat restrictive, we remark that the complexity of classical pairing heap variants such as the front-to-back or multipass heuristics is not known, even in sorting-mode~\cite{pairing,forward_variant,multipass}.}

\paragraph{Algorithms.}

We mention four algorithms for heap re-structuring: a \emph{simple} (folklore) algorithm, and the three \emph{pairing heap variants} introduced by Fredman et al.\ \cite{pairing}. All four algorithms can be implemented with the standard (unstable) link operation, and can be turned into stable heap algorithms if stable link is used instead of the standard link. 

We describe the implementation of extract-min in the four algorithms (see Figure~\oldref{fig_heap_algos}).

The first, \textbf{simple heap} collects the roots in a left-to-right accumulation round, repeatedly linking the (current) leftmost item with its right neighbor. It is easy to construct examples where the amortized cost of extract-min operations is linear, i.e.\ prohibitively large (e.g.\ $1,2,3,\dots$ for the stable variant and $1,3,5,\dots,6,4,2$ for the unstable variant).

The \textbf{standard pairing heap} works in two passes: a left-to-right pairing round in which neighboring items are linked in pairs and a subsequent right-to-left accumulation round. The amortized cost of operations is $O(\log{n})$~\cite{pairing}.

The \textbf{``front-to-back'' variant of pairing heaps} differs from the standard variant only in the fact that the second round is performed left-to-right, rather than right-to-left. In~\cite{pairing} only an $O(\sqrt{n})$ bound is given for the amortized cost of extract-min using this method. This was recently improved in~\cite{forward_variant} to $O(\log{n} \cdot 4^{\sqrt{\log{n}}})$, still far from the conjectured logarithmic bound.

The \textbf{``multipass'' variant of pairing heaps} repeatedly executes pairing rounds (identical to the first round of the standard and front-to-back variants), until a single root remains. In~\cite{pairing} the bound $O(\log{n} \cdot \log\log{n} / \log\log\log{n})$ is given for the cost of extract-min using this method. This was recently improved in~\cite{multipass} to $O(\log{n} \cdot 2^{{\log^{\ast}{n}}} \cdot \log^{\ast}{n})$, again, not known to be tight. 

The mentioned bounds for the three pairing heap variants refer to the implementations with the standard link operation. Only the analysis of~\cite{forward_variant} transfers readily to the stable setting, giving a bound of $O(\log{n} \cdot 2^{O(\sqrt{\log{n}})})$ for all three variants~\cite[\S\,2.2]{forward_variant}. It may well be that the true amortized costs of all three pairing heap variants (both stable and unstable) are $O(\log{n})$.

\subsection{Comparison with other heap models}

Fredman~\cite{FredmanLB} and Iacono and \"Ozkan~\cite{IaconoOzkan,IaconoOzkan2} define general heap models, for the purpose of proving lower bounds for all algorithms within the model. The models of Fredman and Iacono and \"Ozkan are similar to each other and to our stable heap model in that they work on a forest-of-heaps structure, allowing traversal by pointer moves, comparisons, and link operations. 

In Fredman's ``generalized pairing heap'' model, the amount of information stored in the nodes of the underlying tree is restricted. Comparisons between keys are only allowed together with a subsequent link operation. Pointer moves (for reaching the keys to be compared) are free. 

In the ``pure heap'' model of Iacono and \"Ozkan, comparisons are decoupled from links. A wider set of pointer moves and pointer changes are allowed, but their cost is accounted for. 

Our stable heap model borrows elements from both models, but is not directly comparable to them. It is less restrictive in that it allows arbitrary information to be stored in the nodes, and allows arbitrary pointer moves and comparisons for free. It is more restrictive in that changing the structure can happen only via link operations, and only neighboring siblings can be linked. 

The main difference between our model and the other two models is in the link operation: whereas the other two models use the classical, unstable, ``link-as-leftmost'' method, the stable heap model uses stable links. We find the restrictions of our model fairly natural (in hindsight), and in the spirit of classical algorithms, such as pairing heaps -- apart from the use of stable links, which we see as a natural replacement of standard links.  
The main advantage of our model is its surprising connection to the standard BST model.

The lower bounds in our model are of a different kind than the lower bounds shown in the existing models. The lower bounds of Fredman and Iacono-\"Ozkan are concrete constructions of particularly costly sequences of operations, i.e.\ they are lower bounds \emph{for the worst-case}. The lower bounds we show in the stable heap model are \emph{instance-specific}, i.e.\ they apply to all sequences of operations, and describe structural properties that capture the difficulty of these sequences. (In this sense, they are similar to the lower bounds in the BST model~\cite{Wilber, DHIKP09}; in fact, they are the \emph{exact same} bounds, transferred from the BST to the stable heap model.)

\subsection{Stable heaps in sorting-mode}\label{sec:stable heap sort}
Let $X \in [n]^n$ be a permutation sequence. We consider the \emph{execution trace} of a stable heap algorithm for $n$ extract-min operations, after inserting the keys in $[n]$ in the order given by $X$ into an initially empty heap (i.e.\ in sorting-mode). 

For a stable heap algorithm $\cA$, we denote the set of links performed during the sorting-mode execution of $\cA$ on $X$ as $\cA(X)$, i.e.\ $(a,b) \in \cA(X)$ if, at some point during the execution of $\cA$ on $X$, the nodes with keys $a$ and $b$ are linked. Observe that a pair of nodes can be linked only once during the execution (once a node is in the subtree of the other, it stays there). Thus, the cost of the execution is $|\cA(X)|$. We denote by $\opt_{\stable}(X)$ the minimum cost of a stable heap algorithm when serving $X$ in sorting-mode. 

In the following we describe  
the combinatorial \emph{star-path problem} as an intermediate step towards proving the formal connection between the stable heap model and the BST model (\S\,\oldref{sec:produce sat set}).

Recall that a point set $P\subseteq[n]\times[n]$ is a permutation
if $|P_{x=i}|=|P_{y=i}|=1$ for each $i$. In this case, we denote
as $p_{y=i}$ the unique point in $P_{y=i}$, and similarly for $p_{x=i}$. Let $P_0$ denote the set $\{\coord 00\}\cup P$, i.e.\ the set $P$ augmented with the origin.

We consider trees with node-set $P_0$. We call such a tree a \emph{monotone tree}, if $\coord 00$ is the root, and for every edge $(u,v)$ of the tree, $u.y < v.y$ iff $u$ is closer to the root (in graph-theoretical sense) than $v$. Two particular monotone trees are important: $star(P)$ is the tree in which every point in $P$ is the child of $\coord 00$, and $path(P)$ is the path $(\coord 00,p_{y=1},p_{y=2},\dots,p_{y=n})$. See \ref{fig:star path} for illustration. 

\begin{figure}[H]
	\begin{centering}
		\includegraphics[width=0.8\textwidth]{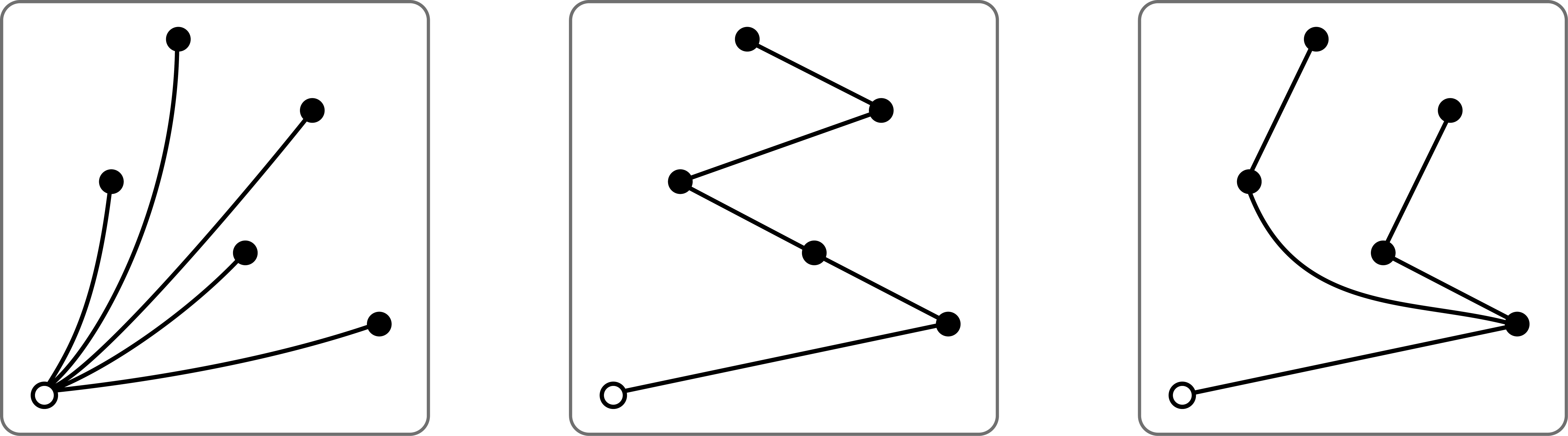}
		\par\end{centering}
	
	\caption{Sequence $X=(3,5,2,4,1)$ and corresponding point set $P = P^{X'}$ (values shown on $y$-coordinate). (i) $star(P)$, (ii) $path(P)$, (iii) an arbitrary monotone tree with points in $P \cup \{\coord 00\}$. \label{fig:star path}}
\end{figure}

We define a link operation in monotone trees as follows. Let $a$ and $b$ be two \emph{neighboring} siblings in the tree (according to $x$-coordinate), and let $u$ be their parent. Suppose $a.y>b.y$. Then $\link(a,b)$ deletes the edge $(u,a)$ and adds the edge $(b,a)$ (i.e.\ changes the parent of $a$ from $u$ to $b$). Otherwise, if $a.y<b.y$, we delete $(u,b)$ and add $(a,b)$. (See~\ref{fig:geometry of stable link}.) 

A link can be performed on any monotone tree that has a node with at least two children, e.g. $star(P)$. Starting from $star(P)$, after at most $O(n^2)$ links we reach $path(P)$, and no more links are possible. (To see this, we can argue that the total length of the $y$-components of edges decreases with every link.)

\begin{figure}[H]
	\begin{centering}
		\includegraphics[width=0.4\textwidth]{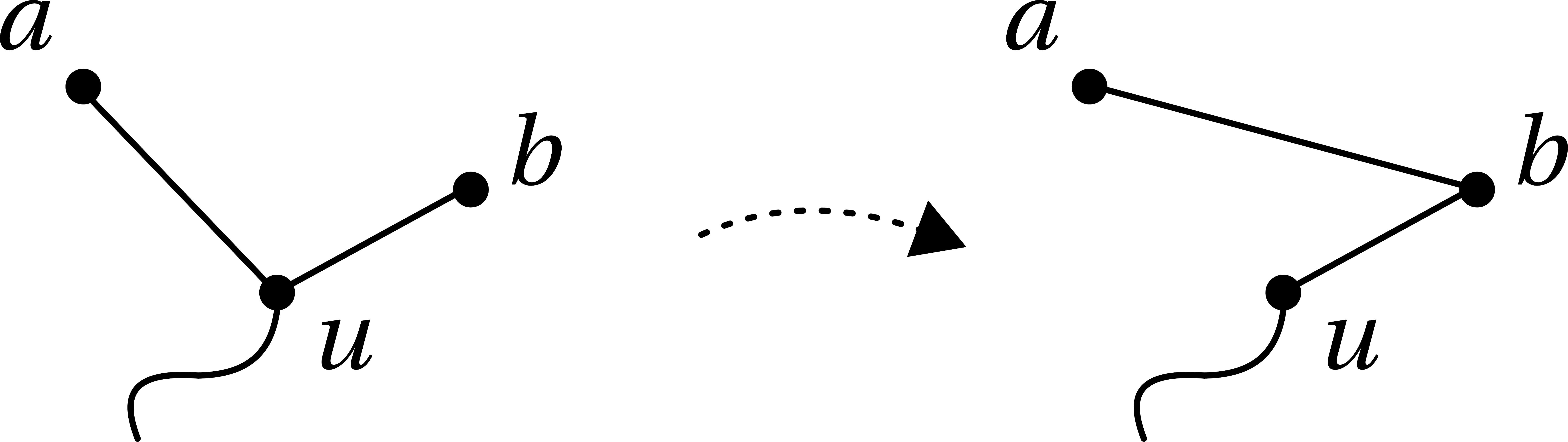}
		\par\end{centering}
	
	\caption{Operation of linking $a$ and $b$ when $a.y>b.y$.\label{fig:geometry of stable link}}
\end{figure}

\begin{defn}
	[Star-path problem] Given a permutation point set $P$, transform $star(P)$ into $path(P)$ using a sequence of link operations.\label{def:path problem}
\end{defn} 

For a point set $P$, we denote by $\cB(P)$ the set of links performed during the execution of an algorithm $\cB$ for the star-path problem with input $P$, i.e.\ $(a,b) \in \cB(P)$, if at some point, $\cB$ links $a$ and $b$ (again, observe that this can happen at most once).
We show that stable heap executions for a permutation sequence $X$ and star-path executions for point set $P^{X'}$ (from $star(P^{X'})$ to $path(P^{X'})$) are, in a precise sense, equivalent.

\begin{thm}\label{prop:heap and path}
	Let $X=(x_{1},\dots,x_{n})$ be an arbitrary permutation of $[n]$, and let $\cP \subset [n]^2$. Then there is a stable heap algorithm $\cA$, such that $\cA(X) = \cP$, iff there is an algorithm $\cB$ for the star-path problem, such that $\cB(P^{X'}) = \cP$.
\end{thm}

\begin{proof}
Let $\cA$ be an arbitrary stable heap algorithm executed in sorting mode for $X$.  
After inserting $X$, the top-level roots are $x_1, \dots, x_n$ (in this order). We assume these roots to be the children of a virtual root with key $0$.  
Thus, at this stage, the heap structure is exactly $star(P^{X'})$, i.e.\ the $i$-th child of the root is $\coord i {x_i}$.
Furthermore, we assume that after each extract-min, the old root is ``kept around'', considering the children of the old root the new top level roots. Thus, after $n$ extract-min operations, the structure of the heap is exactly $path(P^{X'})$, i.e.\ the sorted path of $X$, with a dummy root in front.
These adjustments to $\cA$ are only for convenience, and do not affect its behavior.

To show the equivalence between the two executions, we identify a node with key $x_i$ in the heap with the point $\coord i{x_i}$ in the star-path problem, and we describe a bijection between link operations in the two executions. We maintain the following two invariants:

\textbf{(Invariant 1):} The current heap state is the same as the current monotone tree in the star-path execution (with the virtual roots added in both cases). The keys in the heap view correspond to the $y$-coordinates in the star-path view, and the left-to-right ordering of siblings in the heap view corresponds to the left-to-right ordering by $x$-coordinates in the star-path view. 

\textbf{(Invariant 2):} For every node $q$ in the monotone tree, we have $parent(q).\PREV.x < q.x < parent(q).\NEXT.x$. Here, $\PREV$ and $\NEXT$ denote the left and right neighboring siblings of a node (by $x$-coordinate) in the monotone tree. For convenience we assume that if $q$ has no left neighbor, then $q.\PREV.x = -\infty$, and if $q$ has no right neighbor, then $q.\NEXT.x = +\infty$. 

Initially, in the case of a star, both invariants clearly hold. 
We need to show that linking maintains these invariants.

Consider a link between $a$ and $b$ with parent $u$, where $a$ is to the left of $b$. Assume $a.y > b.y$, as in \ref{fig:geometry of stable link}. After the link, $a$ becomes the child of $b$. 

In the heap, by Invariant (1), the corresponding items $x_a$ and $x_b$ are neighboring siblings ($a<b$), and $x_a > x_b$. Therefore, linking $x_a$ and $x_b$ is a valid operation and $x_a$ becomes the leftmost child of $x_b$. (Conversely, if $x_a$ and $x_b$ are neighboring siblings in the heap, then linking $a$ and $b$ is a valid star-path link.)

We need to show that $a$ becomes the leftmost child of $b$, and thus, the ordering of siblings is by $x$-coordinate. By Invariant (2), before the link operation for all children $c$ of $b$ we have $c.x > parent(c).\PREV.x = a.x$. Thus, Invariant (1) is maintained.

We need to show that Invariant (2) is not violated. This could only happen if, after the link operation, $a.x$ were smaller than $b.\PREV.x$. But this is impossible, since the left neighbor of $b$ is the earlier left neighbor of $a$. The case $a.y < b.y$ is symmetric, and omitted.
\qedd
\end{proof}

Observe that in the star-path problem only the set of edges changes during the execution, the locations of points remain the same. 
In order to maintain this simple geometric model, the use of stable links is essential. (The classical link operation would move entire subtrees from one place to another.) In the remainder of the paper we view stable heap executions mostly in the ``geometric view'' of the star-path problem.

\section{The smooth heap}\label{sec:smooth_desc}

In this section we describe the smooth heap, our new heap data structure. The smooth heap conforms to the stable heap model described in \S\,\oldref{sub:heap model} and is based on a forest-of-heaps representation. The implementations of makeheap, insert, decrease-key and meld are those from the description of stable heap algorithms in \S\,\oldref{sub:heap model}. 

The crucial operation is the restructuring of the top-level list of nodes during extract-min. In the following we give three \emph{equivalent} descriptions of this operation;
a \emph{non-deterministic} description, a \emph{treap-based} description, and a \emph{two-pass} description. We refer to~\ref{fig:smooth2} for the pseudocodes of the three variants.

\begin{figure}
\begin{center}
\textbf{Input: top-level list of roots $X=(x_1, \dots, x_k)$}
\ \\
\parbox{5.5in}{

\begin{algorithm}[H]
\begin{framed}
\newcommand{\llWhile}[2]{{\let\par\relax\lWhile{#1}{#2}}}
\newcommand{\llIf}[2]{{\let\par\relax\lIf{#1}{#2}}}
\newcommand{\llElse}[1]{{\let\par\relax\lElse{#1}}}
\DontPrintSemicolon
\SetAlgoLined
\NoCaptionOfAlgo
\caption{Smooth heap (\emph{non-deterministic} view)}

\textbf{\texttt{while}} $|X| > 1$:\\
\quad \quad \textbf{\texttt{let}} $x$ be an arbitrary \textbf{local maximum} in $X$\\ 
\quad \quad \textbf{\texttt{if}} $x.\PREV.key > x.\NEXT.key$\hfill{(assuming \emph{null.key }$ =-\infty$)} \\
\quad \quad \quad \quad \quad \textbf{\texttt{link}}$(x.\PREV)$ \\
\quad \quad \textbf{\texttt{else}} \\
\quad \quad \quad \quad \quad \textbf{\texttt{link}}$(x)$

\end{framed}
\end{algorithm}

}
\ \\ \ \\ \ \\

\parbox{5.8in}{
\begin{algorithm}[H]
\begin{framed}
\newcommand{\llWhile}[2]{{\let\par\relax\lWhile{#1}{#2}}}
\newcommand{\llIf}[2]{{\let\par\relax\lIf{#1}{#2}}}
\newcommand{\llElse}[1]{{\let\par\relax\lElse{#1}}}
\DontPrintSemicolon
 \SetAlgoLined
\NoCaptionOfAlgo
\caption{Smooth heap (\emph{treap} view)}
Transform $X$ into a treap with keys $(1,2,\dots,k)$ and priorities $(x_1, \dots, x_k)$. 
\end{framed}
\end{algorithm}

}

\ \\ \ \\

\parbox{\columnwidth}{
  \parbox{2.6in}{
    \begin{algorithm}[H]
\begin{framed}
    \DontPrintSemicolon
    \SetAlgoLined
    \NoCaptionOfAlgo
    \caption{Smooth heap (\emph{two-pass} view)}

    {Smooth heap:}\\
    \quad {\texttt{Smoothing-Round}}\\
    \quad {\texttt{RTL-Accumulate-Round}}\\
    
\end{framed}
    \end{algorithm}
  } 
  \hspace{0.2in} 
  \parbox{3.8in}{
    \begin{algorithm}[H]
\begin{framed}
    \NoCaptionOfAlgo
    \caption{Linking rounds}

\vspace{0.03in}
\vspace{0.03in}
    \textbf{\texttt{Smoothing-Round:}}\\
    \textbf{\texttt{while}} there is $x$ in $X$ s.t.\ $x.key > x.\NEXT.key$:\\
    \quad \textbf{\texttt{let}} $x$ be the \emph{leftmost} such node\\
    \quad \textbf{\texttt{if}} ($x.\PREV = null$) or ($x.\PREV.key < x.\NEXT.key$)\\
    \quad \quad $\link(x)$\\
    \quad \textbf{\texttt{else}}\\
	\quad \quad  $\link(x.\PREV)$\\    
	\ \\
    \textbf{\texttt{RTL-Accumulate-Round:}}\\
    \quad	\textbf{\texttt{while}} $|X| > 1$:\\
    \quad \quad $x \gets$ \textit{rightmost}$(X)$.\PREV \\
    \quad \quad $\link(x)$
\end{framed}
    \end{algorithm}
  }

\caption{Smooth heap re-structuring (three different views). Recall that $\link(y)$ denotes a stable link between $y$ and its \emph{right} neighbor $y.\NEXT$.\label{fig:smooth2}}
}
\end{center}
\end{figure}

\paragraph{Non-deterministic view.} In the top-level list $X$ of items, we repeatedly find an arbitrary \emph{local maximum} $x$. A local maximum is an item $x$ that is larger than both its neighboring siblings, or, in case $x$ is the leftmost or rightmost item, larger than its only neighbor.

We link $x$ with the \emph{larger} of its two neighboring siblings (or its only sibling, in case it is the leftmost or rightmost item). 
As $x$ becomes the child of one of its neighbors, it drops out of $X$, reducing the size of $X$ by one. As long as $X$ has at least two elements, there is a suitable next choice of $x$ (for instance, the global maximum of $X$). When $X$ becomes a singleton, we are done; as this item is the minimum, it can be deleted.

The non-deterministic view is useful in analysing smooth heaps, because of its resemblance to our non-deterministic view of Greedy (\S\,\oldref{sec:non-det greedy}).

\paragraph{Two-pass view.} This description differs from the non-deterministic description only in the choice of the local maximum $x$: we always choose the \emph{leftmost} such item. Observe that if the only remaining local maximum is the rightmost item, then the items in the list are \emph{sorted}. In this case, the remaining items can be linked in a single right-to-left pass, which we can execute without further comparisons. It is thus convenient to view the execution in two passes that resemble the description of pairing heaps: a left-to-right \emph{smoothing pass} followed by a right-to-left accumulation pass. The two-pass view of smooth heaps is perhaps the most convenient to implement. We describe this implementation in more detail in~\ref{smooth_heap_ps}.

\paragraph{Treap view.} We associate each item $x_i$ in $X$ with a pair of values $(i,x_i)$, and we transform the list into a treap over the pairs (using an arbitrary method for treap-building). Recall that a treap is a binary tree with a pair of values in every node, respecting the in-order (i.e.\ search tree order) according to the first entry, and the min-heap order according to the second entry of every pair. As mentioned, such a tree is unique. (The item with the unique minimum priority is the root, and the items with smaller, resp.\ larger key values form its recursively-built left and right subtrees). We use the treap view in order to connect the non-deterministic and two-pass views of the smooth heap.

A remark is in order: as a treap is a binary tree, each node may have a left and a right child. For every node in $X$, its left child in the treap will become its \emph{leftmost} child in the underlying tree, and its right child in the treap will become its \emph{rightmost} child. In other words, the existing children of $x_i$ end up \emph{between} the two new children possibly gained during the treap-building.

\begin{thm}
The non-deterministic, two-pass, and treap-based descriptions of smooth heaps describe the same transformation.
\end{thm}

We show that the non-deterministic algorithm constructs a treap regardless of the order in which local maxima are chosen, and thus the non-deterministic and treap-based descriptions are equivalent. As the two-pass description is just a particular way of choosing the local maxima, it follows that it also produces a treap, and since the treap is unique, all three views are equivalent.

Although there are several linear-time algorithms known for constructing a treap from an array, e.g.\ \cite{TarjanTreap}, we are not aware of a previous mention in the literature of the particular method implicit in the description of the smooth heap. Due to its simplicity, we find it interesting in its own right; in the following, we prove its correctness.  

\begin{proof}
We refer to the non-deterministic description in~\ref{fig:smooth2}, and denote by $X=(x_1, \dots, x_k)$ the initial top-level list of items.

Let $T$ be the tree built from $X$ through repeatedly linking local maxima of $X$ with their larger neighbor, until a single item remains. Three properties of $T$ together imply that $T$ is a treap. 

\noindent(1) $T$ is a \emph{heap} according to the key-values $x_i$,\\
\noindent(2) $T$ is a \emph{binary} tree,\\
\noindent(3) $T$ is a \emph{search tree}, according to the indices $i$.\\

Property (1) follows from the definition of linking (links only create edges where the parent is smaller than the child).

To see Property (2), observe that once an item $x$ is larger than its left neighbor $x.\PREV$, it continues to be larger as long as it stays in $X$ (this is because if the left neighbor of $x$ drops out of $X$, then it must have become the child of even smaller item, which is the new left neighbor of $x$). Similarly, once $x$ is larger than its right neighbor $x.\NEXT$, it will remain so, as long as it stays in $X$. Suppose $x$ and $y$ are linked, and $x$ gains $y$ as a child. Then, according to our ``two-choices'' strategy, $x$ must have gained a new neighbor that is smaller than $x$ (the old neighbor of $y$), unless $x$ is already the leftmost or rightmost in the list. It follows that $x$ can gain at most two children (one as a new leftmost child, and one as a new rightmost child).

Property (3) follows from the fact that the list is initially ordered by the indices $i$ and stable links maintain this order. (If a node has a single child, we consider it a right child or left child, depending on whether the child was linked as rightmost or leftmost.) \qedd
\end{proof} 

\textit{Remark}: As mentioned, the smooth heap is implemented using stable links. In the non-deterministic and two-pass descriptions it is in fact also possible to use the classical, unstable link. In this case, however, the three descriptions are no longer equivalent. It is an interesting open question how efficient the resulting unstable two-pass algorithm is.

It is instructive to compare smooth heaps and pairing heaps. Besides the implementation of the link operation (stable or unstable) there are several differences between the two algorithms. 

In pairing heaps, as in Fredman's model of generalized pairing heaps, comparisons only occur within a link operation, i.e.\ a comparison is always followed by the corresponding link.
By contrast, in smooth heaps, comparisons are decoupled from links, and are also used to decide \emph{which pair of nodes} to link. In fact, comparisons are used \emph{only} for this purpose: it can be observed in the two-pass description in~\ref{fig:smooth2} or in the description of~\ref{smooth_heap_ps}, that at the time of linking, it is always already known, which of the two linked items is greater, therefore, the link operation can proceed without any comparison.

Recall that in our model, the cost of the algorithm is its ``link cost'', i.e.\ the number of link operations performed. In the detailed description of~\ref{smooth_heap_ps} it can be observed that the number of other elementary operations is proportional to the link cost. In particular, every comparison is followed by moving the cursor to the right, or by a link at the cursor. It follows that the number of comparisons is at most twice the number of links. (The number of items to the right of the cursor \emph{plus} the total number of items decreases after every comparison.)

\begin{figure}
	\begin{centering}
		\includegraphics[width=0.6\textwidth]{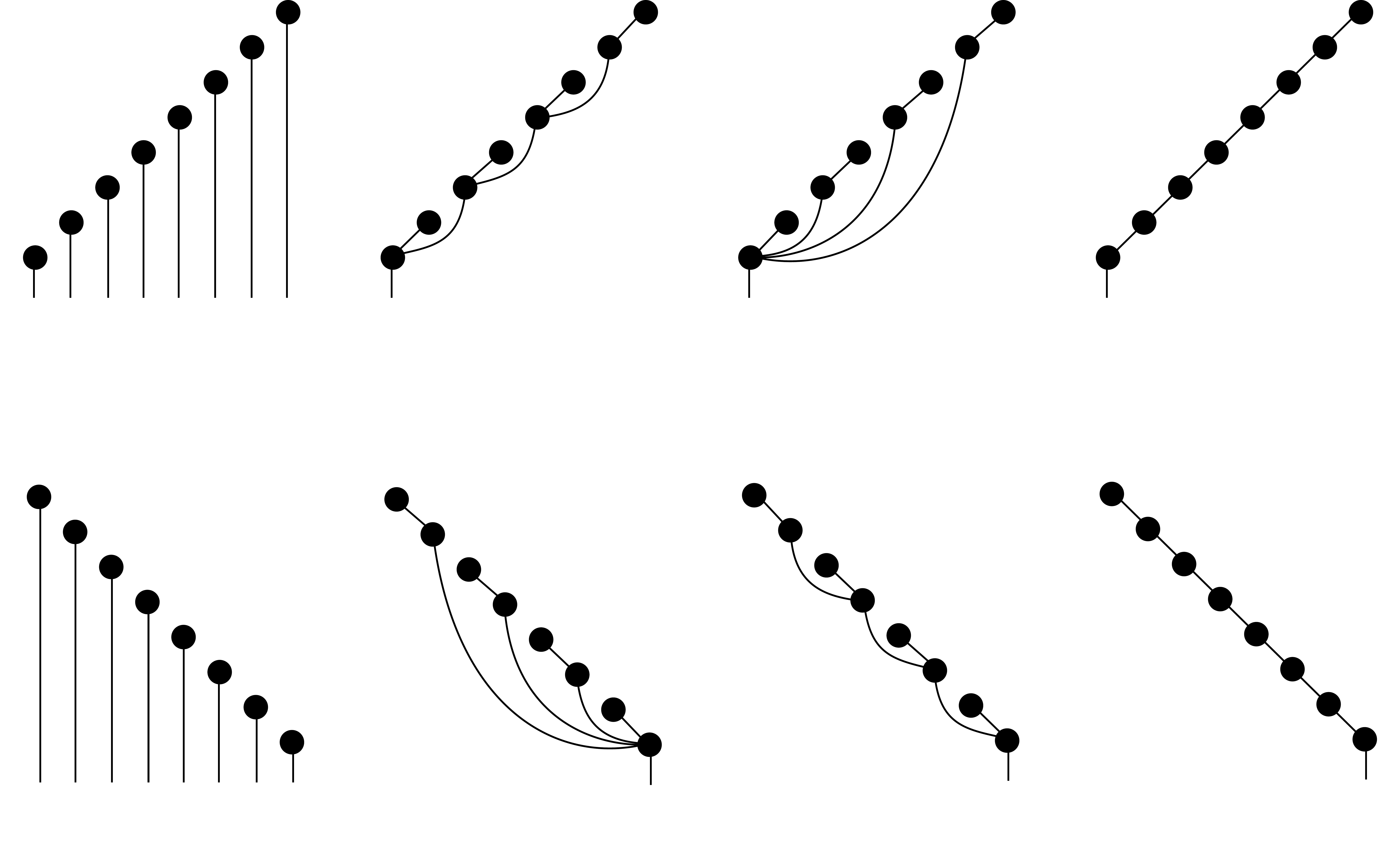}
		\par\end{centering}
	
	\caption{Extract-min operation illustrating the difference between smooth heap and pairing heap with increasing (above) and decreasing (below) initial state. From left-to-right: initial state, standard pairing heap, front-to-back pairing heap, smooth heap. The behavior of smooth heap is similar to the \emph{better} of the two other algorithms, in both cases. A transformation is better, if its result is closer to the fully sorted state, i.e.\ to a path. This is the case for the standard pairing heap for increasing, for the front-to-back pairing heap for decreasing, and the smooth heap for both increasing and decreasing input. In the other cases, the resulting heaps resemble the initial, unsorted state.\label{fig_positive}}
\end{figure}

The distinctive feature of smooth heaps is their \emph{power-of-two-choices} linking, whereby an item $x$ is linked with the larger of its two neighbors (both of which are smaller than $x$). In other words, of the two possible edges we could create, we choose the one with smaller rank-difference. Intuitively, with this choice, we expect to move closer to the totally ordered state (i.e.\ a path in which the rank-difference along every edge is one). Depending on the subtrees of the nodes, this choice may, of course, not be globally optimal, as shown in~\ref{fig_counter} in~\ref{app:nonoptimal}. As shown in \S\,\oldref{sec:produce sat set} there is good reason to believe that the smooth heap is not far from being optimal; a large gap from the optimum would disprove the conjectured optimality of Greedy. In~\ref{fig_positive} we show examples that illustrate the behavior of smooth heaps compared to pairing heap variants. An intuitive reason for the efficiency of the smooth heap is that it does the optimal transformation for both increasing and decreasing (sub)sequences, see~\ref{fig_positive}. (By contrast, the standard and back-to-front pairing heap variants are efficient in one and inefficient in the other case.)

Finally, we give some remarks on how smooth heaps (in the two-pass view) can be implemented with only two pointers per node. We refer to the description in~\ref{smooth_heap_ps}. First, we make the lists circularly linked, such that the \NEXT~pointer of the rightmost item points to the leftmost item in every list, marking the rightmost item appropriately. We can now omit the \LM~pointer, as it can be simulated via the \RM~and \NEXT~pointers. Getting rid of the \PREV~pointer is trickier. Observe that in the first, left-to-right pass we may need, after a link, to make a step to the left with the cursor, if an item was linked with its left neighbor. Finally, in the second, right-to-left pass we need to make several steps to the left. In order to simulate these steps without using \PREV~pointers, we can \emph{reverse} the links of the list as we traverse it left-to-right, such that the \NEXT~pointers of the already visited items point to their left neighbors in the list. (If we move left, we undo the reversal.) This allows us to move the cursor both left and right from its current position using only \NEXT~pointers, and a constant-size buffer to hold the two items at the cursor (and a constant factor overhead in cost). We omit further details.

\newpage
\section{The transformations and their consequences\label{sec:produce sat set}}

In this section we present the main connections between stable heap
algorithms and BST algorithms and we discuss the consequences of these connections.

We show that the cost of smooth heaps in sorting-mode, once the role of key-space and time are swapped, matches the cost of \greedyfuture, conjectured to be instance-optimal in the BST model. 
 
\begin{thm}
[Smooth-Greedy transformation]\label{thm:smooth tran}For every permutation
$X$, 
\[
|\greedyfuture(X')|=\Theta(|\smooth(X)|).
\]

\end{thm}

As an immediate consequence of \ref{thm:smooth tran} and the amortized analysis of \greedyfuture~\cite{Fox11}, we have the following.

\begin{cor}[Amortized analysis of smooth heaps]
For every permutation $X \in [n]^n$, 
\[
|\smooth(X)| = O(n \log{n}).
\]
\end{cor}

For the amortized cost of $\greedyfuture$ several stronger, instance-specific bounds are known. We describe two from the literature.

The \emph{weighted dynamic finger} $WDF(X)$ of an input sequence $X$ is a quantity that describes its \emph{locality of reference}. It was introduced in~\cite{LI16}, and it subsumes the earlier \emph{dynamic finger} bound~\cite{ST85, finger1, finger2}, as well as other bounds (see e.g.~\cite{landscape}). For the exact definition of $WDF(\cdot)$ we refer to~\cite{LI16}.

A permutation $X \in [n]^n$ \emph{avoids} a pattern permutation $\pi \in [k]^k$ if there is no subsequence of $X$ (not necessarily contiguous) that is \emph{order-isomorphic} to $\pi$. (We refer the reader to~\cite{knuth68,tarjan_sorting, pratt_queues, Kitaev, MarxPattern, Newman} for more information on this extensively studied property.) As a simple observation, we mention that if $X$ avoids $\pi$ then $X'$ avoids $\pi'$.

We have the following results.
\begin{thm} For every permutation $X \in [n]^n$: 
\begin{itemize}
\item {\cite{LI16}}  $|\greedy(X)| = O(WDF(X))$.
\item \cite{FOCS15} If $X$ avoids some permutation $\pi \in [k]^k$, then $|\greedy(X)| = n \cdot 2^{{\alpha(n)}^{O(k)}}$, where $\alpha(\cdot)$ is the slowly growing inverse Ackermann function.
\end{itemize}
\end{thm}

We immediately obtain the following.

\begin{cor}[Instance-specific upper bounds for smooth heaps]
For every permutation $X \in [n]^n$:
\begin{itemize}
\item $|\smooth(X)| = O(WDF(X'))$.
\item If $X$ avoids some permutation $\pi \in [k]^k$, then $|\smooth(X)| = n \cdot 2^{{\alpha(n)}^{O(k)}}$.
\end{itemize}
\end{cor}

More generally, we show that there is a general transformation from \emph{an arbitrary} stable
heap algorithm to an offline BST algorithm:
\begin{thm}
[General transformation]\label{thm:general tran}For every stable
heap algorithm $\cA_{\stable}$ there is an offline BST algorithm
$\cA_{\bst}$, such that for every permutation $X$,
the costs of $\cA_{\stable}$ and $\cA_{\bst}$ for sorting $X$ are asymptotically the same, i.e. 
\[
|\cA_{\bst}((X')^{r})|=\Theta(|\cA_{\stable}(X)|),
\]
and in particular, 
\[
\opt_{\bstins}((X')^{r})=O(\opt_{\stable}(X)).
\]

\end{thm}
Observe that the permutation input of $\cA_{\bst}$ is inverted
\emph{and} reversed. From \ref{rem:bst insert stronger}, we note
again that bounding from above the running time of a BST algorithm in sorting mode is stronger than bounding it in search-only mode. 

For search-only mode, we can obtain infinitely many offline BST
algorithms from a single stable heap algorithm, all of which have at most a constant factor larger cost. (Here, \emph{infinitely many} is understood as $n$ tends to infinity, i.e.\ the number of BST algorithms generated depends on the input size $n$.) In particular, we obtain multiple (non-trivial) offline BST executions that are $O(1)$-competitive with $\greedyfuture$.

The proofs of these results are presented in \S\,\oldref{sec:proofs}.
Below we discuss some of their interesting consequences.

\subsection{Consequences of the optimality of $\protect\greedyfuture$\label{sub:consequence greedy}}

$\greedyfuture$ is widely conjectured to be instance-optimal \cite{DHIKP09,Mun00,Luc88}.
In case the conjecture is true, we obtain the following two
statements. First, our smooth heap algorithm is an instance-optimal
stable heap algorithm (at least) for sorting. Second, the optimal
costs of (1) selection-sort with stable heaps ($\opt_{\stable}$), (2) insertion-sort with BSTs ($\opt_{\bstins}$), and (3) searching with BSTs ($\opt_{\bst}$) are the same within some constant factor. Moreover, each of these
quantities are invariant under applying inversion and/or reversion to the input. 

Without the assumption, we only know that, for every permutation $X$, we have
$\opt_{\bstins}(X)\ge\opt_{\bst}(X)$ by \ref{rem:bst insert stronger},
and $\opt_{\stable}(X)\ge\Omega(\opt_{\bstins}((X')^{r}))$ by \ref{thm:general tran}.
It is not clear how to prove the inequalities in the other direction,
or how to compare $\opt_{\stable}(X)$
and $\opt_{\bstins}(X)$.
\begin{cor}
Assuming \ref{conj:greedy opt}, for every permutation $X$,
\begin{enumerate}[(i)]
\item 
$$\smooth(X)=\Theta(\opt_{\stable}(X)),$$  
\item 
$$
\opt_{\stable}(X)=\Theta(\opt_{\bstins}(X))=\Theta(\opt_{\bst}(X)),
$$

and, for each $model\in\{\stable,\bstins,\bst\}$,
\[
\opt_{model}(X)=\Theta(\opt_{model}(X'))=\Theta(\opt_{model}(X^{r})).
\]
 
\end{enumerate}
\end{cor}
\begin{proof}
We use $f\lesssim g$ to denote $f=O(g)$. Consider the following
inequalities, for an arbitrary permutation $X$.

\begin{align*}
\opt_{\bst}(X) & \le\opt_{\bstins}(X) & \mbox{by \ref{rem:bst insert stronger}}\\
 & \lesssim\opt_{\stable}((X')^{r}) & \mbox{by \ref{thm:general tran}}\\
 & \le|\smooth((X')^{r})|\\
 & \lesssim|\greedyfuture(X^{r})| & \mbox{by \ref{thm:smooth tran}}\\
 & \lesssim\opt_{\bst}(X^{r}) & \mbox{by \ref{conj:greedy opt}}\\
 & =\opt_{\bst}(X) & \mbox{by \ref{cor:bst symmetric}.}
\end{align*}
Assuming \ref{conj:greedy opt}, the inequalities collapse,
implying (i). 

For (ii), we have $\opt_{\bst}(X)=\Theta(\opt_{\bstins}(X))$ and $\opt_{\stable}((X')^{r})=\Theta(\opt_{\bst}(X))$
from the above inequalities. By \ref{cor:bst symmetric}, $\opt_{\bst}(X)=\Theta(\opt_{\bst}(X'))=\Theta(\opt_{\bst}(X^{r}))$
for every permutation $X$. \qedd 
\end{proof}

\subsection{Dynamic optimality for stable heaps\label{sub:dynopt for heap}}

The theory of instance-specific lower bounds for BSTs is much richer
than the corresponding theory for heaps. There are several \emph{concrete} lower bounds known for $\opt_{\bst}$, Wilber's
first bound (a.k.a.\ the ``interleave bound''), Wilber's second bound (a.k.a.\ the ``funnel bound'')~\cite{Wilber} and the maximum independent rectangle (MIR) bound~\cite{DHIKP09}. (See also~\cite{in_pursuit} for the precise definitions of these quantities.)
It is typically easier to reason about these lower bounds than to analyse $\opt_{\bst}$ directly. Currently, all BST algorithms shown to be $o(\log n)$-competitive (i.e.\ Tango \cite{Tango}, Multi-splay \cite{multisplay} and Chain-splay \cite{chain_splay}) have been shown to be competitive with Wilber's first bound. Some connections between
these bounds are known. In~\cite{DHIKP09} it is shown that both
Wilber's bounds are subsumed by the MIR bound, which is computable by a sweepline algorithm (intriguingly) similar to \greedy. Harmon~\cite{Harmon} also shows that the MIR bound is captured (up to a constant factor) by the optimal cost of a network-design problem called \emph{small Manhattan networks}~\cite{manhattan}.\footnote{Wilber's first bound is in fact, a family of bounds, as it is defined with respect to a reference tree. Here, by $W_1(X)$ we denote the best bound in the family, i.e.\ using the reference tree that maximizes the bound for given $X$.}

For heaps, lower bounds have so far been studied in the \emph{worst-case} setting
\cite{FredmanLB,IaconoOzkan,IaconoOzkan2}, and a theory of instance-specific
lower bounds has not yet been proposed (in any heap model). 

Our results connect the two models, and yield an analogous theory for heaps (at least with the restrictions implied by the stable heap model). 
By~\ref{thm:general tran}, for \emph{every} stable heap algorithm, there is a corresponding BST algorithm  with the same asymptotic cost. By this result, all known instance-specific lower bounds are immediately transferred from BSTs to stable heaps.
\begin{cor}
Let $W_{1}(X)$, $W_{2}(X)$, $MIR(X)$ be Wilber's two bounds~\cite{Wilber}, and the maximum independent
rectangle (MIR) bound~\cite{DHIKP09} for an arbitrary permutation $X$. Then we have
\[
W_{1}(X),W_{2}(X)\le MIR(X) \le O(\opt_{\stable}(X)).
\]
\end{cor}
\begin{proof}
In \cite{DHIKP09} it is shown that $W_{1}(X),W_{2}(X)\le MIR(X) \le O(\opt_{\bst}(X))$.
The following concludes the claim:
\begin{align*}
\opt_{\bst}(X) & =\opt_{\bst}((X')^{r}) & \mbox{by \ref{cor:bst symmetric}}\\
 & \le\opt_{\bstins}((X')^{r}) & \mbox{by \ref{rem:bst insert stronger}}\\
 & =O(\opt_{\stable}(X)). & \mbox{by \ref{thm:general tran}} \\ & &\qedd
\end{align*}
\let\qed\relax
\end{proof}

\section{A non-deterministic description of $\protect\greedy$\label{sec:non-det greedy}}

In this section we describe a non-deterministic algorithm for the
satisfied superset problem whose output is exactly the same as the
output of \greedy, described in \S\,\oldref{sec:bst_model}.
We call this algorithm ``non-deterministic \greedy'' and denote it $\greedy_{nondet}$. This alternative description of \greedy\,can be of independent interest; we use it to show a connection between the smooth heap and \greedy\,in \S\,\oldref{sec:proofs}.
The equivalence between \greedy~and $\greedy_{nondet}$ is shown in \ref{thm:greedy nondet}. We first define the $\add$ gadget. 
\begin{defn}
[$\add$ gadget]\label{def:add gadget}Let $P\subseteq\mathbb{Z}\times\mathbb{Z}$
be a point set. We say that $G=(a,b,c,d,e)$ is an $\add$ gadget
in $P$ if 
\begin{itemize}
\item $\{a,b,c,d,e\}\subseteq P$ 
\item The relative positions of $a,b,c,d,e$ are according to \ref{fig:add gadget}
or its horizontal reflection. That is, $a.y>b.y>c.y=d.y=e.y$ and
$e.x<a.x=c.x<b.x=d.x$, or $e.x>a.x=c.x>b.x=d.x$.
\item Let $f=\coord{b.x}{a.y}$. Then $f\notin P$ 
\item Let $\square_{G}=\square_{ef}\setminus\bigl(\{\coord{e.x}i\mid i\}\cup\{\coord i{e.y}\mid i\}\cup\{\coord{b.x}i\mid i\}\cup\{\coord i{a.y}\mid i\}\bigr)$.
In words, $\square_{G}$ is $\square_{ef}$ with the borders removed.
Then, $\square_{G}\cap P=\emptyset$. See \ref{fig:add gadget}. We
call $\square_{G}$ the \emph{rectangle inside the gadget} $G$. 
\end{itemize}
\end{defn}
\begin{figure}[H]
\begin{centering}
\includegraphics[width=0.19\textwidth]{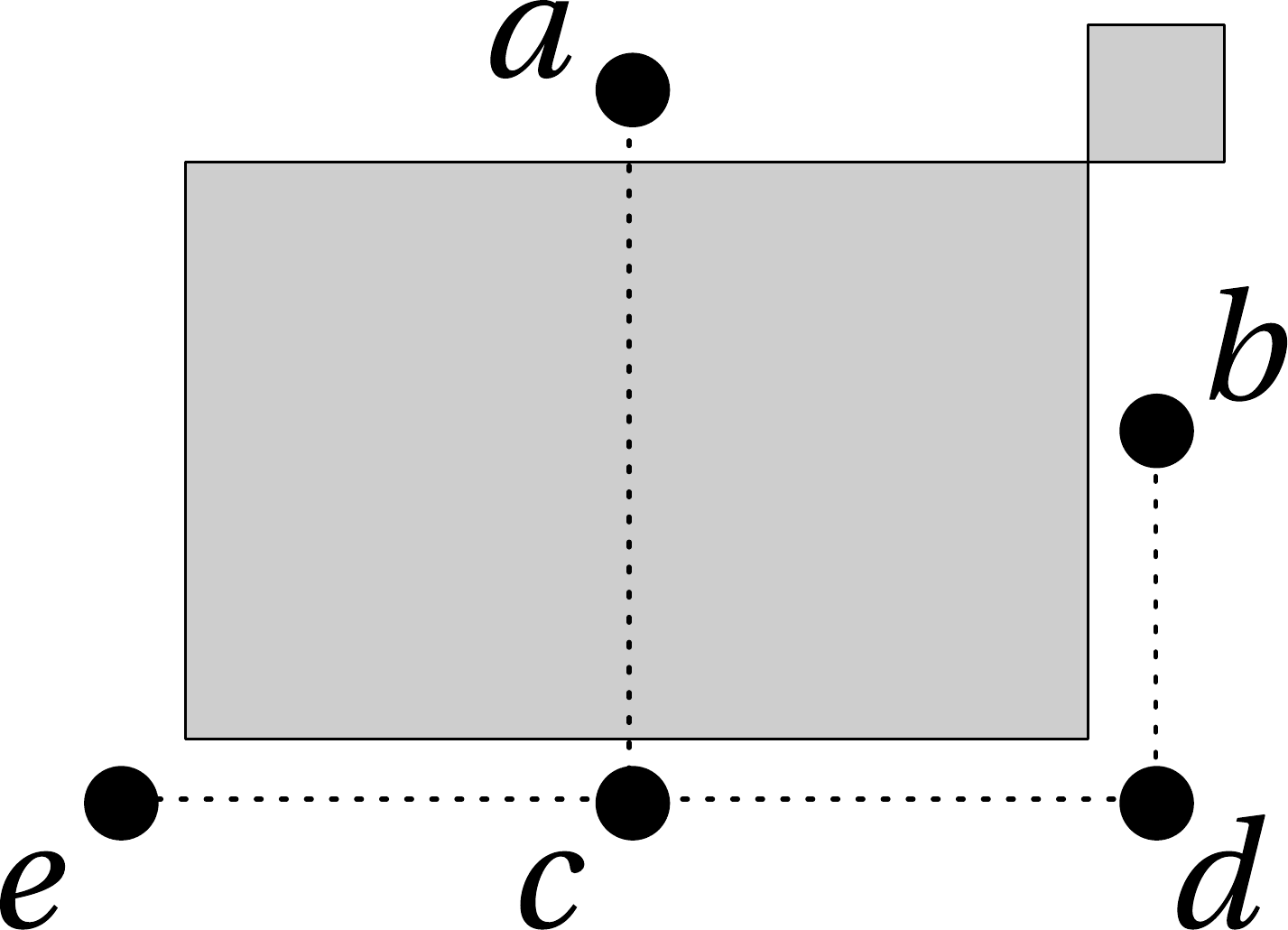} 
\par\end{centering}

\caption{$\protect\add$ gadget $G=(a,b,c,d,e)$. The larger shaded rectangle
is $\square_{G}$. The small square is $f = \protect\coord{b.x}{a.y}$. Dotted lines indicate horizontal or vertical alignment. 
Shaded areas are free of points. \label{fig:add gadget}}
\end{figure}

We begin by proving some basic properties of the $\add$ gadget during the execution of \greedy. Recall the process of \greedy~from \ref{def:greedy}. First, we show that \greedy never adds a point inside $\square_{G}$. This claim follows from~\cite[Lemma 2.4(iii)]{FOCS15}, but we include the argument for completeness. 

\begin{lem}
	If there is an $\add$ gadget $G=(a,b,c,d,e)$ in point set $P$, then $\square_{G} \cap \greedy(P) = \emptyset$.
	\label{lem:inside gadget empty}\end{lem}
\begin{proof}
	We assume that $e.x<d.x$ (the horizontally reflected case is analysed similarly).

	Suppose for contradiction that when \greedy~runs on $P$, it adds a point $r$ inside $\square_{G}$, and let $r$ be the lowest such point. Observe that $e.y < r.y \leq a.y$, and $e.x < r.x < d.x$.
	Let $Q$ be the point set before processing row $r.y$.
	
	
	From the definition of \greedy~it follows that there are two points $p$ and $q$ where $p\in P_{y=r.y}$
	and $q\in Q_{y \le e.y}$ such that $\square_{pq}$ is unsatisfied and $q.x=r.x$. See \Cref{fig:gadget_empty}.
	We know that either $p.x\le e.x$ or $d.x\le p.x$ because $p\notin\square_{G}$.
	But since $e.x<q.x<d.x$ (because $q.x=r.x$ and $r\in\square_{G}$), either $e$ or $d$ must be in $\square_{pq}$, contradicting the claim that $\square_{pq}$ is unsatisfied. \qedd
	
\end{proof}

\begin{figure}
	\begin{centering}
		\includegraphics[scale=0.22]{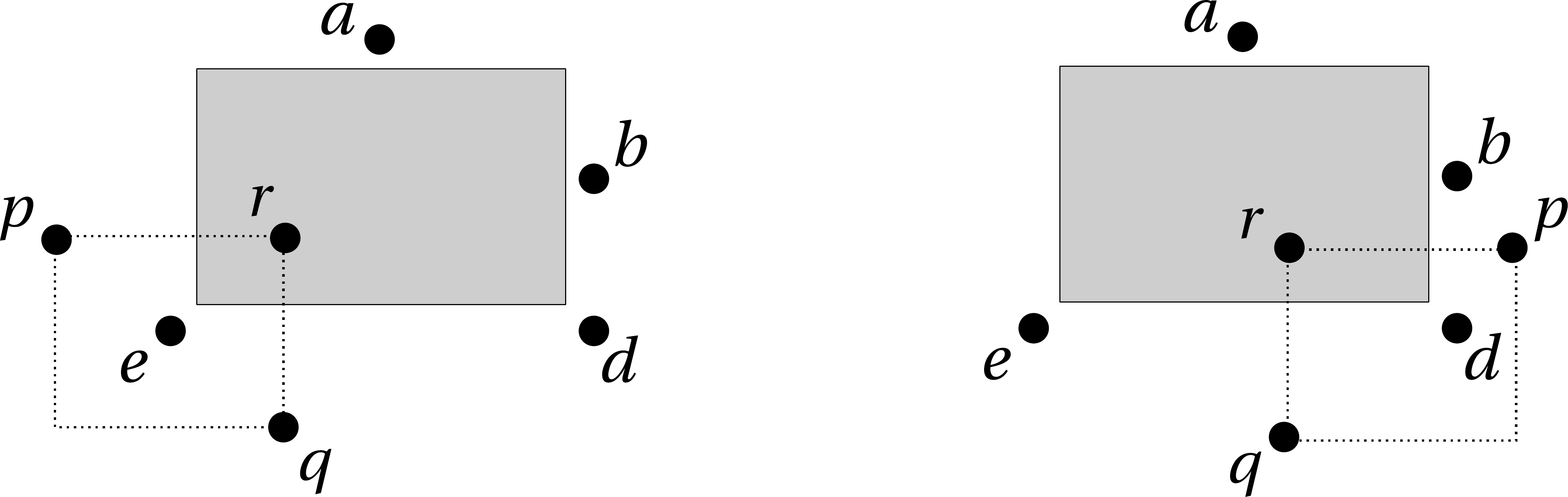}
		\par\end{centering}
	\caption{Illustration of the proof of \Cref{lem:inside gadget empty}. Suppose that $r$ is the lowest point in $\square_{G} \cap \greedy(P)$. Either $e$ or $d$ contradict the fact that $\square_{pq}$ is unsatisfied. \label{fig:gadget_empty}}
\end{figure}

By \emph{filling the gadget $G$} we mean the action of adding the
point $f = \coord{b.x}{a.y}$ into a point set $P$ containing the gadget
$G=(a,b,c,d,e)$. We observe that \greedy~always fills
an $\add$ gadget. 

\begin{lem}
If there is an $\add$ gadget $G=(a,b,c,d,e)$ in $P$, then $\coord{b.x}{a.y}\in\greedy(P)$. \label{lem:add gadget works}\end{lem}
\begin{proof}
By \Cref{lem:inside gadget empty}, it holds that $\square_{G}\cap\greedy(P)=\emptyset$,
i.e.\ \greedy\;never adds a point within $\square_{G}$. 
Consider the step when \greedy~processes row $a.y$. There must be $b'$ such that $b'.x=b.x$ and $b.y\le b'.y<a.y$, and $a'$ such that
$a'.y=a.y$ and $a.x \le a'.x < b.x$ 
 such that $\square_{a'b'}$ is unsatisfied ($b'$
may be $b$ and $a'$ may be $a$). Therefore, $f = \coord{b.x}{a.y}$ is added. \qedd
\end{proof}
A straightforward induction shows the following:
\begin{prop}
Let $P$ be a point set and $Q\subseteq\greedy(P)$. Then, $\greedy(Q)=\greedy(P)$.\label{prop:greedy monotone} 
\end{prop}
Let $\boxx=\{\coord i0|~i\in[n]\}\cup\{\coord 0i|~0\le i\le n\}\cup\{\coord{n+1}i|~0\le i\le n\}$. See \Cref{fig:boxx}.
We state two easy observations. 

\begin{figure}
	\begin{centering}
		\includegraphics[scale=0.2]{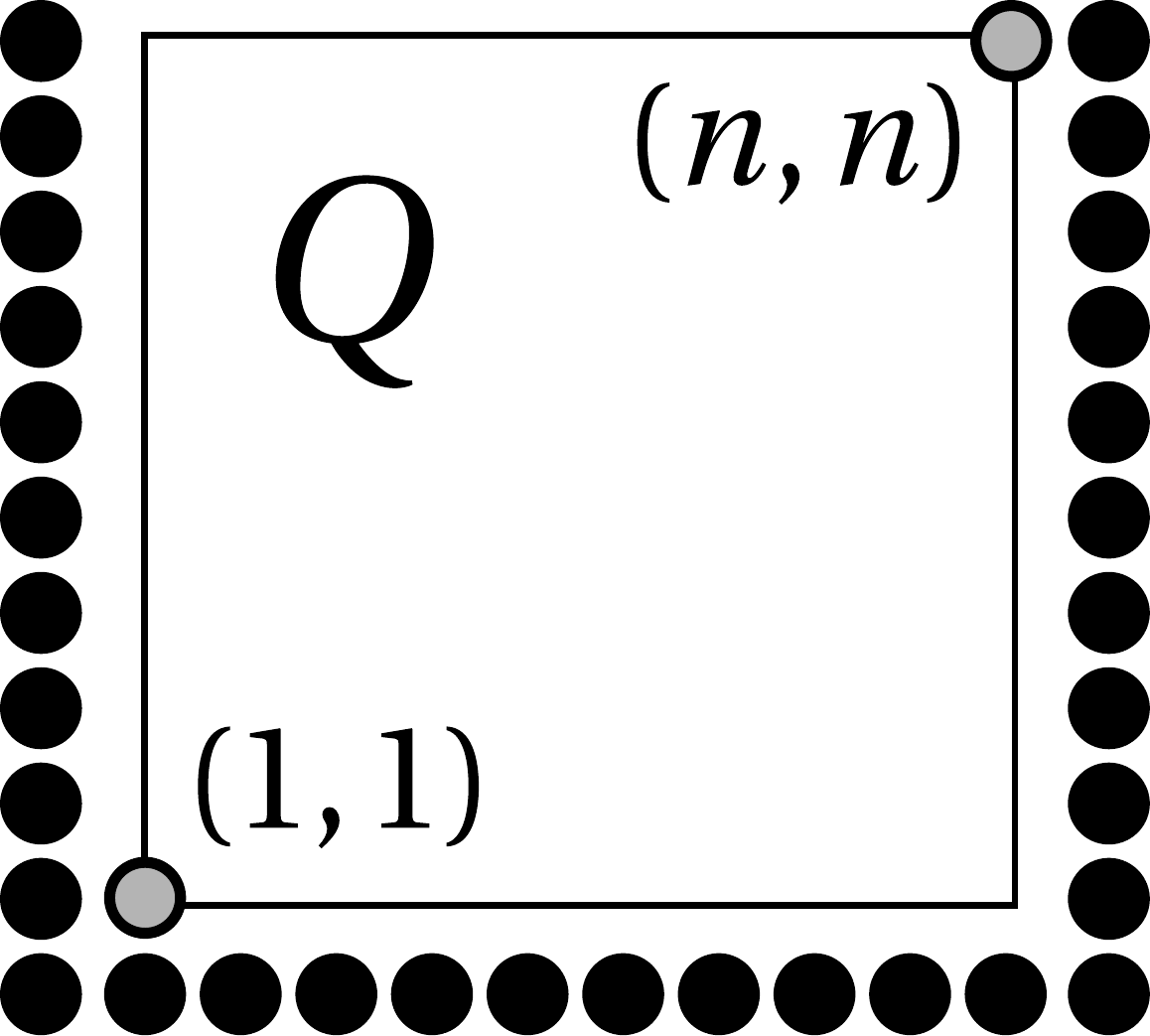}
		\par\end{centering}
	\caption{Illustration of $Q\cup\boxx$ where $Q \subseteq [n]\times[n]$ is a point set.\label{fig:boxx}}
\end{figure}

\begin{prop}
\label{prop:base}Let $Q\subseteq[n]\times[n]$ be an arbitrary point
set. Then, 
\begin{enumerate}[(i)]
\item $\greedy(Q\cup\boxx)=\greedy(Q)\cup\boxx$. 
\item $Q\cup\boxx$ is satisfied iff $Q$ is satisfied. 
\end{enumerate}
\end{prop}

\begin{figure}
	\begin{centering}
		\includegraphics[width=0.18\textwidth]{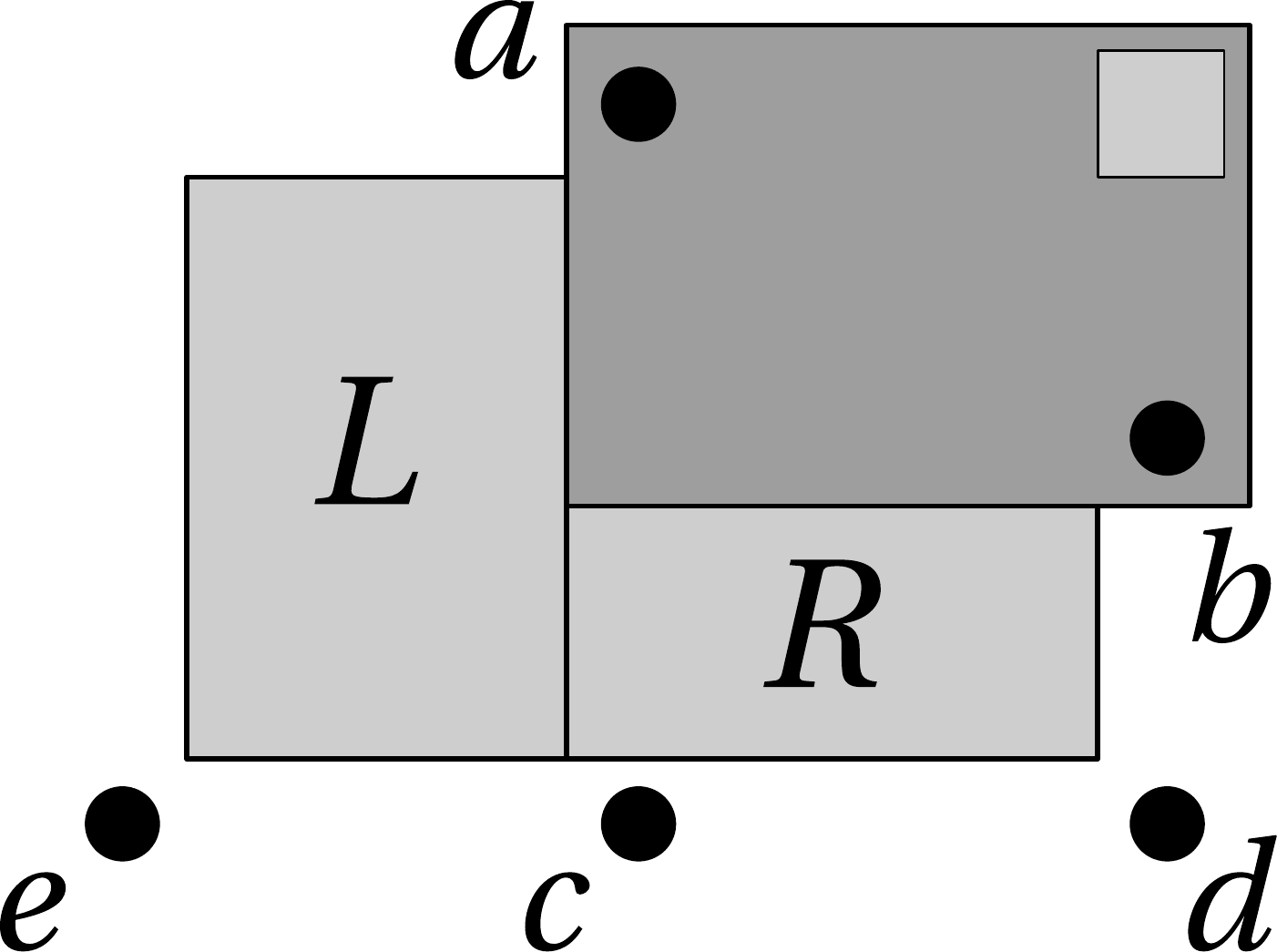} 
		\par\end{centering}
	
	\caption{$\protect\add$ gadget $G=(a,b,c,d,e)$ from the proof of \ref{lem:cannot add imply sat}, where $\square_{ab}$ is unsatisfied and $\square_{G}$ is partitioned
		into $L$ and $R$.\label{fig:add unsat}}
\end{figure}

\begin{lem}
Let $Q\subseteq[n]\times[n]$ be an arbitrary point set. There is
no $\add$ gadget in $(Q\cup\boxx)$ iff $Q$ is satisfied.\label{lem:cannot add imply sat}\end{lem}
\begin{proof}
One direction is trivial. Suppose there is an $\add$ gadget $G=(a,b,c,d,e)$
in $(Q\cup\boxx)$. Assume w.l.o.g.\ that $e.x<d.x$. Let $a'\in (Q\cup\boxx)$
be a point such that $a'.y=a.y$ and $a.x\le a'.x\le b.x$ where $a'.x$
is maximal. Let $b'\in Q\cup\boxx$ be a point such that $b'.x=b.x$
and $b.y\le b'.y\le a.y$ where $b'.y$ is maximal. Then 
$\square_{a'b'}$ is unsatisfied because $\square_{G}\cap Q=\emptyset$
and $\coord{b.x}{a.y}\notin Q$. Hence $(Q\cup\boxx)$ is unsatisfied,
and so is $Q$ by \ref{prop:base}(ii).

Now, suppose that $Q$ is unsatisfied, and thus, $(Q\cup\boxx)$ is unsatisfied
by \ref{prop:base}(ii). Let $a,b\in (Q\cup\boxx)$ be such that $\square_{ab}$
is unsatisfied and $b.y$ is minimal. Assume by symmetry
that $a.x<b.x$. Let $c,d,e\in (Q\cup\boxx)$ be such that the relative
positions of $a,b,c,d,e$ are according to \ref{fig:add gadget}.
When there is more than one choice, we first maximize $e.y$, and then maximize
$e.x$. We know that such $c,d,e$ exist because $\coord 00,\coord{a.x}0,\coord{b.x}0\in\boxx$. (This is the reason we consider $Q \cup \boxx)$ and not just $Q$.)
We claim that $G=(a,b,c,d,e)$ is an $\add$ gadget in $Q\cup\boxx$.

By construction (1) $\{a,b,c,d,e\}\subseteq (Q\cup\boxx)$,
(2) the relative positions of $a,b,c,d,e$ are according to \ref{fig:add gadget},
and (3) $\coord{b.x}{a.y}\notin (Q\cup\boxx)$ as $\square_{ab}$ is
unsatisfied. It remains to show that $\square_{G}$ is empty (i.e.\ contains no point from $Q\cup\boxx$). First note that $\square_{G}\cap\square_{ab}$
is empty because $\square_{ab}$ is unsatisfied. 
Partition $\square_{G}\setminus\square_{ab}$ into disjoint parts $L$ and $R$, where for each point $p\in\square_{G}\setminus\square_{ab}$, we have $p\in L$
if $p.x<c.x$, and $p\in R$ otherwise. (See \ref{fig:add unsat}.) Suppose
$\square_{G}\setminus\square_{ab}$ is not empty. Then we claim that $R$ is not empty. 

\begin{figure}
	\begin{centering}
		\includegraphics[scale=0.23]{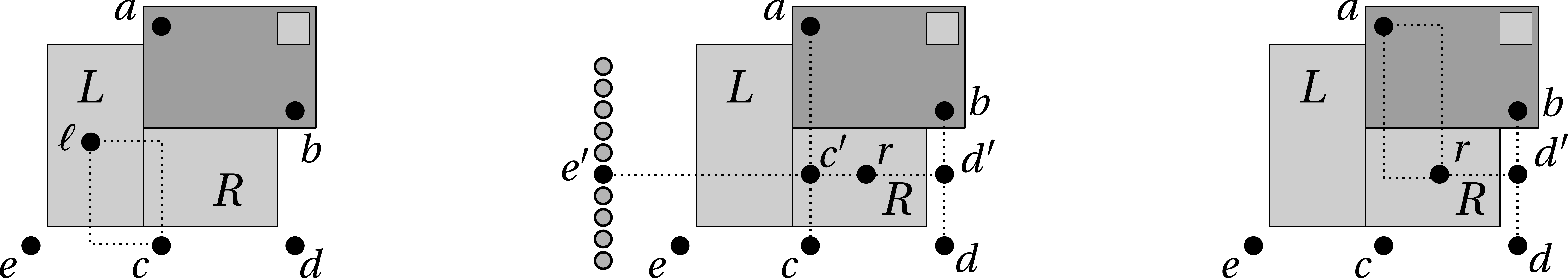}\caption{Illustration of the proof of \ref{lem:cannot add imply sat}. \textbf{(left)} Suppose that $L$ is not empty but $R$ is empty.
		Then $\square_{lc}$ is unsatisfied and $c.y<b.y$ which contradicts
		the choice of $a$ and $b$. \textbf{(middle)} If $c',d'\in Q\cup\boxx$,
		then $(c',d',e')$ contradicts the choice of $(c,d,e)$ where $e.y$
		is maximal. \textbf{(right) }If $c'\protect\notin Q\cup\boxx$, then
		$\square_{ar}$ is unsatisfied and contradicts the choice of $a$
		and $b$. If $d'\protect\notin Q\cup\boxx$, we reach a contradiction similarly.\label{fig:no_add_then_sat}}
	\end{centering}
\end{figure}

\begin{claim}
If $\square_{G}\setminus\square_{ab}$ contains a point in $Q\cup\boxx$,
then $R$ contains a point in $Q\cup\boxx$.\end{claim}
\begin{proof}
Suppose for contradiction that there is a point in $L$ but $R$ is empty. Observe that (1) $\{\coord i{e.y}\mid e.x<i<c.x\}$ is empty
because $e$ is such that $e.x$ is maximal, and (2) $\{\coord{c.x}i\mid e.y<i<a.y\}$
is empty because $R\cup\square_{ab}$ is empty. Consider a rightmost, bottommost point $\ell\in L\cap(Q\cup\boxx)$. We have that $\square_{\ell c}$
is unsatisfied and $c.y<b.y$. See \Cref{fig:no_add_then_sat} (left). But we chose $a$ and $b$ such that
$\square_{ab}$ is unsatisfied and $b.y$ is minimal. This is a contradiction. \qedd
\end{proof}
By the above claim, let $r\in R\cap(Q\cup\boxx)$ where $r.y$ is
maximal. We will show a contradiction. (See \Cref{fig:no_add_then_sat}.)

\begin{claim}
Either $c'=\coord{a.x}{r.y}\notin (Q\cup\boxx)$ or $d'=\coord{b.x}{r.y}\notin (Q\cup\boxx)$.\end{claim}
\begin{proof}
Suppose otherwise. See \Cref{fig:no_add_then_sat} (middle). Then there are points $e'=\coord 0{r.y},c',d'\in (Q\cup\boxx)$, such that $(a,b,c',d',e')$ are in the desired configuration, but $e'.y>e.y$. This contradicts the choice of $c,d,e$. \qedd
\end{proof}
By the above claim, there are two cases. If $\coord{a.x}{r.y}\notin (Q\cup\boxx)$,
then we choose $r$ such that $r.x$ is minimal (among those with the same $r.y$). Let $a'\in (Q\cup\boxx)$
be such that $a'.x=a.x$ and $r.y<a'.y\le a.y$ where $a'.y$ is minimal
(possibly $a'=a$). We have that $\square_{a'r}$ is unsatisfied and,
moreover, $r.y<b.y$. This contradicts the choice of $a$ and $b$. See \Cref{fig:no_add_then_sat} (right).
If $\coord{b.x}{r.y}\notin (Q\cup\boxx)$, then we choose $r$ such
that $r.x$ is maximal (among those with the same $r.y$). Then the argument is symmetric. 

We conclude that $R$ is empty, and therefore, $\square_G$ must be empty. \qedd
\end{proof}

The above lemmas already suggest the definition of $\greedy_{nondet}$. 
\begin{defn}
[Non-deterministic \greedy]Given a set $P\subseteq[n]\times[n]$,
$\greedy_{nondet}$ works as follows. Initially, set $Q=(P\cup\boxx)$.
Repeatedly find an \emph{arbitrary} $\add$ gadget $G=(a,b,c,d,e)$
in $Q$. Then, fill $G$ (i.e.\ add point $f = \coord{b.x}{a.y}$ to
$Q$). Let $Q$ be the final point set where there is no $\add$ gadget
in $Q$. Let $\greedy_{nondet}(P)=(Q\setminus\boxx)$.\label{def:nondet greedy} 
\end{defn}
Now we show that $\greedy_{nondet}$ and \greedy\;are identical. 
\begin{thm}
For every point set $P\subset[n]\times[n]$, we have $\greedy_{nondet}(P)=\greedy(P)$.\label{thm:greedy nondet}\end{thm}
\begin{proof}
We proceed via two claims: (1) $\greedy_{nondet}(P)$ is satisfied, and (2) $\greedy_{nondet}(P)\subseteq\greedy(P)$. These claims together imply
that $\greedy_{nondet}(P)=\greedy(P)$, because $\greedy(P)$ is a
\emph{minimally} satisfied set in the sense that every set $Q$ such that $P \subseteq Q \subset \greedy(P)$ is unsatisfied. (To see this, look at the first row $i$ where $Q$ and $\greedy(P)$ differ. By definition, the points that \greedy~adds in row $i$ are \emph{necessary and sufficient} to satisfy all rectangles defined by points in rows $1,\dots,i-1$ and $P_{y=i}$. As $Q_{y=i}$ is a proper subset of $\greedy(P)$ in row $i$, it leaves some rectangle unsatisfied, which remains unaffected by the points in rows $i+1, \dots, n$.)


For claim (1), by the definition of $\greedy_{nondet}$, there
is no $\add$ gadget in $\greedy_{nondet}(P)\cup\boxx$. By \ref{lem:cannot add imply sat},
$\greedy_{nondet}(P)\cup\boxx$ is satisfied, and so is $\greedy_{nondet}(P)$
by \ref{prop:base}(ii).

For claim (2), 
by~\ref{prop:base}(i), it suffices to show that $\greedy_{nondet}(P)\cup\boxx\subseteq\greedy(P\cup\boxx)$.
We show this by induction on the steps of $\greedy_{nondet}$. For
the base case, the initial point set is $(P\cup\boxx)\subseteq\greedy(P\cup\boxx)$.
For the inductive step, let $Q$ be the points added by $\greedy_{nondet}$
so far. We have $Q\subseteq\greedy(P\cup\boxx)$ by induction, and
so $\greedy(Q)=\greedy(P\cup\boxx)$ by \ref{prop:greedy monotone}.
Let $f$ be the new point added to $Q$ by $\greedy_{nondet}$.
Note that $f$ is added because of an $\add$ gadget in $Q$.
Thus, by \ref{lem:add gadget works}, $f\in\greedy(Q)=\greedy(P\cup\boxx)$.
It follows that $Q\cup\{f\}\subseteq\greedy(P\cup\boxx)$. This concludes
the proof. \qedd \end{proof}

\section{Proofs\label{sec:proofs}}

In this section we present the proofs of \ref{thm:smooth tran} and
\ref{thm:general tran}  
from \S\,\oldref{sec:produce sat set}.

\subsection{Reduction to geometric setting}

In this subsection we state some results about the connection between
the two geometric problems (the star-path problem and the satisfied
superset problem). Then, we prove that they imply the main results
from \S\,\oldref{sec:produce sat set}.

To prove \ref{thm:smooth tran}, let $\smooth_{\text{\ensuremath{\path}}}$
be an algorithm for the star-path problem obtained from the smooth
heap algorithm using \ref{prop:heap and path}. We have the following
key lemma: 
\begin{lem}
$|\greedy(P)|=\Theta(|\smooth_{\path}(P)|)$ for every permutation
point set $P$.\label{lem:smooth tran geo} 
\end{lem}
To prove \ref{thm:general tran}, we show the analogous statement
in geometric setting. 
\begin{lem}
For every algorithm $\cA_{\path}$ for the star-path problem, there
is an algorithm $\cA_{\sat}$ for the satisfied superset problem,
such that $|\cA_{\sat}(P)|=\Theta(|\cA_{\path}(P)|)$ for every permutation
$P$. Moreover, $(\cA_{\sat}(P))^{r}$, i.e.\ the reverse of the
output of $\cA_{\sat}$, 
is an insertion-compatible
superset of $P^{r}$.\label{lem:general tran geo} 
\end{lem}

Assuming the two lemmas above, we can prove \ref{thm:smooth tran}
and \ref{thm:general tran}.
\begin{proof}
We use $f\approx g$ to denote $f=\Theta(g)$. Fix some permutation
$X$ on $[n]$.

\textbf{(\ref{lem:smooth tran geo} implies \ref{thm:smooth tran}):}
\begin{align*}
|\greedyfuture(X')| & =|\greedy(P^{X'})| & \mbox{by \ref{thm:greedy and greedyfuture}}\\
 & \approx|\smooth_{\path}(P^{X'})| & \mbox{by \ref{lem:smooth tran geo}}\\
 & =|\smooth(X)|. & \mbox{by \ref{prop:heap and path}}
\end{align*}

\textbf{(\ref{lem:general tran geo} implies \ref{thm:general tran}):}
Given a stable heap algorithm $\cA_{\stable}$ from \ref{thm:general tran},
there is $\cA_{\path}$ where $\cA_{\stable}(X)=\cA_{\path}(P^{X'})$
by \ref{prop:heap and path}. Plugging $\cA_{\path}$ into \ref{lem:general tran geo},
there is an algorithm $\cA_{\sat}$ for the satisfied superset problem,
where $|\cA_{\sat}(P^{X'})|=\Theta(|\cA_{\path}(P^{X'})|)$, with
$(\cA_{\sat}(P^{X'}))^{r}$ being an insertion-compatible satisfied
superset of $(P^{X'})^{r}$. By treating $(P^{X'})^{r}$ as the input,
there is an algorithm $\cA'_{\sat}$ such that $\cA'_{\sat}((P^{X'})^{r})=\cA_{\sat}(P^{X'})$
and $\cA'_{\sat}((P^{X'})^{r})$ is insertion-compatible for its input
$(P^{X'})^{r}$. As $(P^{X'})^{r}=P^{(X')^{r}}$, there is an offline
BST algorithm $\cA_{\bst}$ in insert-only mode (sorting mode) such
that $\cA_{\bst}((X')^{r})=\cA'_{\sat}(P^{(X')^{r}})$, by \ref{thm:geo ins}.
To summarize, we have:

\begin{align*}
|\cA_{\bst}((X')^{r})| & =|\cA'_{\sat}(P^{(X')^{r}})| & \mbox{by \ref{thm:geo ins}}\\
 & =|\cA_{\sat}(P^{X'})|\\
 & \approx|\cA_{\path}(P^{X'})| & \mbox{by \ref{lem:general tran geo}}\\
 & =|\cA_{\stable}(X)|. & \mbox{by \ref{prop:heap and path}}
\end{align*}
\let\qed\relax
\end{proof}
It remains to prove \textbf{\ref{lem:smooth tran geo}} and \textbf{\ref{lem:general tran geo}}.
We prove them using a similar approach, but with different key ingredients
in the proof. The rest of the section is devoted to these proofs.

\subsection{The common framework\label{sub:common framework}}

From now on, we fix an algorithm $\cA_{\path}$ for the star-path
problem and a permutation $P$. We describe two different ways of
producing a satisfied superset $Q$ of $P$ such that $|Q|=\Theta(|\cA_{\path}(P)|)$.
The first method works for arbitrary $\cA_{\path}$. We will show
that $Q$ is insertion-compatible for $P^{r}$ which gives \ref{lem:general tran geo}.
The second method works only for $\cA_{\path}=\smooth_{\path}$. In
this case we will show that $Q=\greedy(P)$ which gives \ref{lem:smooth tran geo}. 

In this subsection, we describe a common framework for constructing
a satisfied set $Q$. 
Then, in the next
subsections, we show how to construct $Q$ for general $\cA_{\path}$
and for $\smooth_{\path}$.

Let $P_{0}=P\cup\{\coord 00\}$. Recall the definitions of $star(P)$
and $path(P)$ near \ref{def:path problem} (they are two particular
monotone trees whose nodes are $P_{0}$). Initially, we set $Q=Q^{init}$
where the definition of $Q^{init}$ will be described differently
in \S\,\oldref{sub:general step} and \S\,\oldref{sub:smooth greedy step}. For
each link that $\cA_{\path}$ performs to transform $star(P)$ to
$path(P)$, we add a constant number of points to $Q$ while maintaining
certain invariants. Once the execution of $\cA_{\path}$ is finished
(producing $path(P)$), the invariant will imply that $Q$ is satisfied.
To state the invariants precisely, we need some definitions.

Let $T$ be the current tree of $\cA_{\path}$. For convenience, we
write the set of nodes of $T$ as $P_{0}\eqdef\{u_{0},\dots,u_{n}\}$
where $u_{0}=\coord 00$ and $u_{i}$ is the unique point in $P_{y=i}$.
Let $Q\supseteq Q^{init}$ be the current point set. For every $i$,
we write $Q_{u_{i}}=Q_{y=i}$. For every node $u$ we define $I(u)\subset[n]$
as \emph{the interval of $u$} and we write $I(u)\eqdef(I(u).\min,I(u).\max)$.
We call $I$ the \emph{interval function}. The precise definition
of $I$, which is a key element of the proofs, will be given later.

The invariants are the following (illustrated in \ref{fig:inv_smooth-1}). 
\begin{itemize}
\item \textbf{Invariant 1 (Intervals respect tree-structure):} The set of
intervals $I(u)$ for all nodes $u\in T$ forms a laminar family whose
structure is exactly $T$. That is, let $u$ be an arbitrary node in $T$ with
children $v_{1},\dots,v_{k}$. Then $I(v_{j})\subseteq I(u)$ for
all $j$, and $I(v_{j})\cap I(v_{j'})=\emptyset$ for all $j\neq j'$. 
\item \textbf{Invariant 2 (Tree-structure respects satisfiability):} If $u$
is an ancestor of $v$ in $T$, then for every $p\in Q_{u}$ and $q\in Q_{v}$
the rectangle $\square_{pq}$ is satisfied. 
\item \textbf{Invariant 3 (No points between a node and its parent's interval):
}For every node $v$ with parent $u$, we have $\{\coord{v.x}i\mid u.y<i<v.y\}\cap Q=\emptyset$.
\item \textbf{Invariant 4 (Adding interval endpoints):} For every node
$v$ with parent $u$, 
\begin{eqnarray*}
\coord{I(v).\min}{v.y},\coord{I(v).\max}{v.y} & \in & Q\mbox{ (endpoints of \ensuremath{I(v)}), and}\\
\coord{I(v).\min}{u.y},\coord{I(v).\max}{u.y} & \in & Q\mbox{ (endpoints of \ensuremath{I(v)} ``projected'' to row of parent \ensuremath{u}).}
\end{eqnarray*}
\end{itemize}

\begin{figure}
	\begin{centering}
		\includegraphics[width=0.5\textwidth]{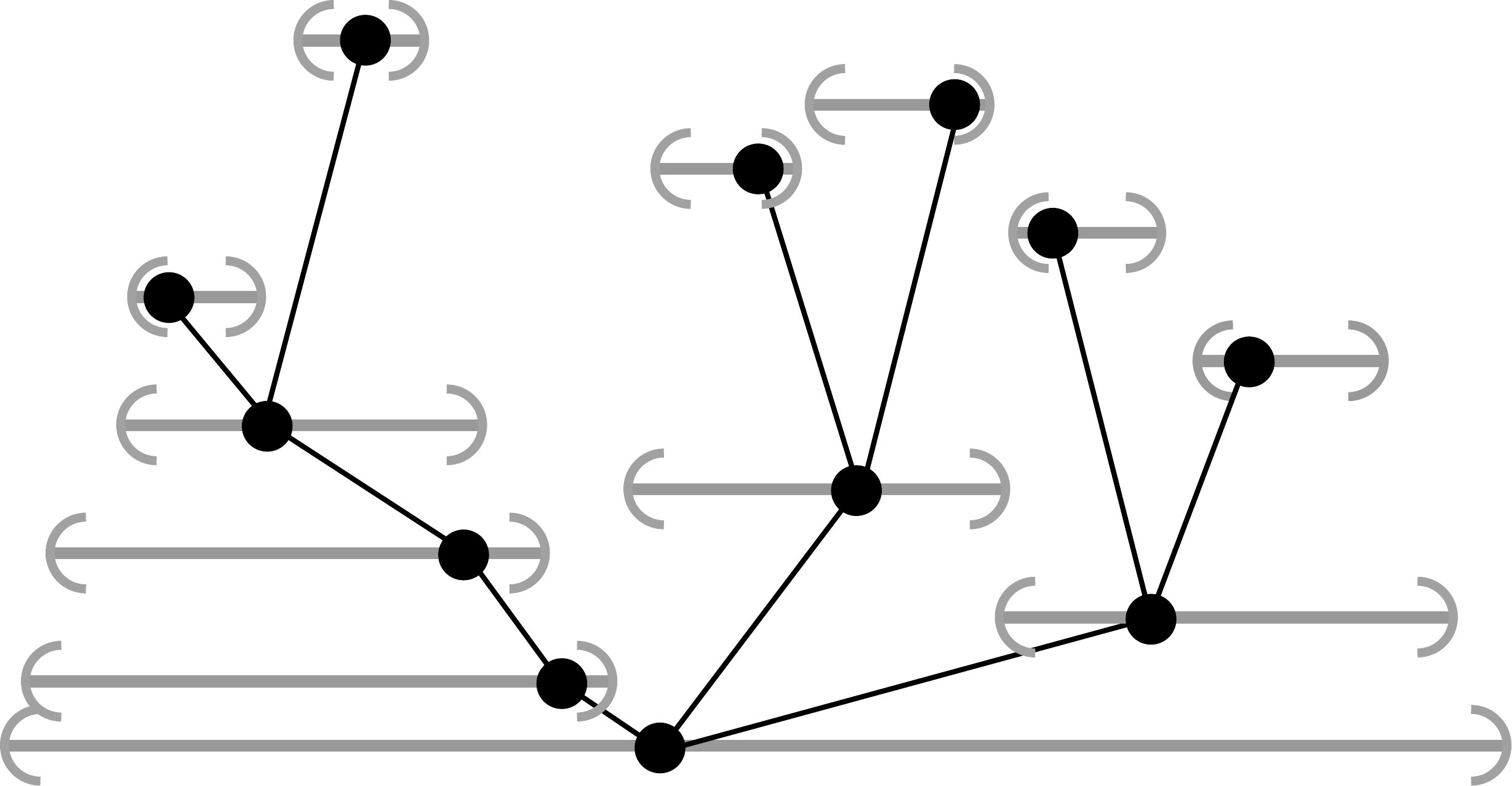}
		\par
	\end{centering}
	
	\caption{Illustration of Invariant (1). Horizontal segments indicate intervals of nodes. 
		\label{fig:inv_smooth-1}}
\end{figure}

We make a crucial observation. 
\begin{prop}
At the end, when $T=path(P)$, if the invariants hold, then $Q$ is
a satisfied set.\label{lem:inv end}\end{prop}
\begin{proof}
Let $p$ and $q$ be two arbitrary points in $Q$. W.l.o.g. $p\in Q_{u_{i}}$
and $q\in Q_{u_{j}}$ where $i\le j$. As $T=path(P)$, $u_{i}$ is
an ancestor of $u_{j}$. By Invariant (2), $\square_{pq}$ is satisfied.  \qedd
\end{proof}

\subsection{The general transformation\label{sub:general step}}

Let $\base=\{\coord i0~|~i\in[n]\}$. We define $Q^{init}\eqdef P\cup\base$,
and define the interval function as follows. For each node $u_{i}$,
\begin{equation}
I(u_{i})\eqdef(\min_{q\in Q_{u_{i}}}q.x,\max_{q\in Q_{u_{i}}}q.x).\label{eq:interval general}
\end{equation}
In words, $I(u_{i})$ is the open interval between the leftmost and
rightmost ``touched'' positions in the $i$-th row. With this definition
of $Q^{init}$ and $I$, we observe the following. 
\begin{lem}
The invariants hold initially.\label{lem:inv init general}\end{lem}
\begin{proof}
Initially, $Q=(P\cup\base)$, and $T=star(P)$. Observe that $I(u_{0})=(1,n)$
and $I(u_{i})=(i,i)=\emptyset$ for all $i\in[n]$ by definition of
$Q_{init}$. Moreover, these intervals form a laminar family whose
structure is exactly $star(P)$, so Invariant (1) holds. Next, for
every $i\in[n]$, $u_{i}$ has only $u_{0}$ as an ancestor. For every
point $p\in Q_{y=0}^{init}$ and $q\in Q_{y=i}^{init}$, we have that
$\square_{pq}$ is satisfied as $Q_{y=0}^{init}=\base=\{\coord i0\mid1\le i\le n\}$.
Thus, Invariant (2) holds. Invariants (3) and (4) hold trivially. \qedd
\end{proof}
Next, we specify how to add points to $Q$ for each link performed
by $\cA_{\path}$ and show that the invariants are maintained.

%

\begin{figure}
	\begin{centering}
		\includegraphics[scale=0.2]{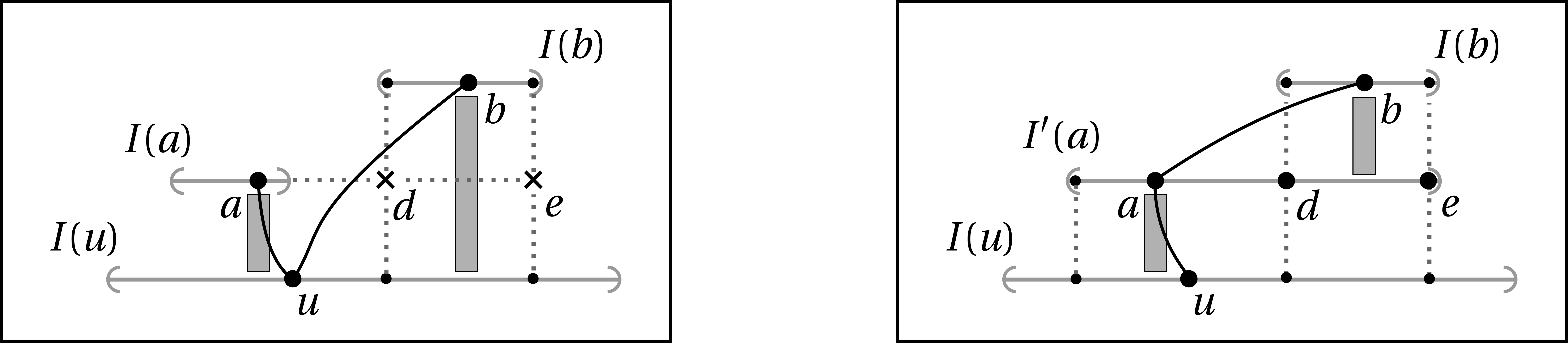}
		\par\end{centering}
	\caption{A step in the general transformation and the maintenance of invariants. \textbf{(left)} Before the transformation: $a$ and $b$ are children of $u$ in $T$. Crosses mark locations of newly added points $d$ and $e$. \textbf{(right)} After the transformation: $b$ is a child of $a$ in $T'$. Intervals shown as horizontal segments, respecting tree-structure (Invariant (1)). Shaded rectangles indicate empty areas (Invariant (3)). Smaller dots indicate interval endpoints and their projections (Invariant (4)), some points omitted for clarity. Non-emptiness of newly created rectangles (Invariant (2)) is guaranteed by newly added points and the points of Invariant (4). 
		\label{fig:inv_gen}}
\end{figure}

\begin{lem}
Let $T$ be the current tree and $Q$ be the current point set. Suppose
that the invariants hold for $T$ and $Q$. Let $T'$ be obtained from
$T$ by a stable link operation. We can add \textbf{one} or \textbf{two}
new points to $Q$ and obtain $Q'$ such that the invariants hold
for $T'$ and $Q'$.\label{lem:inv maintained general}\end{lem}
\begin{proof}
Suppose that $a$ and $b$ are neighboring nodes in $T$ with the
same parent $u$, and $T'$ is obtained from $T$ by changing the
parent of $b$ from $u$ to $a$. Assume $a.x<b.x$ (the other case
is symmetric). Let $Q'=Q\cup\{d,e\}$ where $d=\coord{I(b).\min}{a.y}$
and $e=\coord{I(b).\max}{a.y}$ (see \ref{fig:inv_gen}). We
note that it is possible that $d$ is not new, i.e.\ $d\in Q$ already.
But we have the following:
\begin{claim}
$e$ is a new point, i.e.\ $e\notin Q$. Hence, $1\le|Q'\setminus Q|\le2$.\end{claim}
\begin{proof}
There are two cases. First, suppose $I(b).\min=I(b).\max=b.x$. Then
by Invariant (3), we have that $\{\coord{b.x}i\mid u.y<i<b.y\}\cap Q=\emptyset$.
As $e=\coord{b.x}{a.y}$, so $e\notin Q$. Next, suppose $I(b).\min<I(b).\max$.
Then, by \ref{eq:interval general}, for every $p\in Q_{a}$, we have $p.x\le I(a).\max\le I(b).\min$.
As $e.x=I(b).\max$, so $e\notin Q_{a}$ and hence $e\notin Q$. \qedd
\end{proof}
It remains to show that the invariants hold for $T'$ and $Q'$.

For Invariant (1), observe that $I'(x)=I(x)$ for all $x$, except
for $x=a$. As $I(a).\max\le d.x$ by Invariant (1), we have $I'(b)=(d.x,e.x)\subseteq(I(a).\min,e.x)=I'(a)$.
For any child $c$ of $a$, $I'(c).\max=I(c).\max\le I(a).\max\le d.x=I'(b).\min$
(the first inequality holds by Invariant (1)), so we have $I'(b)\cap I'(c)=\emptyset$.
As we only change the parent of $b$ (from $u$ to $a$), all these
observations imply that Invariant (1) holds after adding the new points.

For Invariant (2), after linking, $a$ becomes a new ancestor of $b$
and of any descendant of $b$. It is sufficient to consider rectangles $\square_{pq}$ where $p\in Q'_{a}$, and $q\in Q'_{c}$, where $c$ is an ancestor or descendant of $a$.

First, let $c$ be a descendant of $b$ (possibly $b$ itself). 
There are three cases: (1) if
$p=d$, then $\coord{I(b).\min}{b.y}\in\square_{pq}$; (2) if
$p=e$, then $\coord{I(b).\max}{b.y}\in\square_{pq}$; (3) if $p\neq d,e$, then $p\in Q_{a}$ and $q\in Q_{c}$ and so $d\in\square_{pq}$.

Second, let $c$ be an old descendant of $a$ (possibly $a$ itself). If $p \in \{d,e\}$, then $\coord{I(a).\max}{a.y}\in\square_{pq}$. (If $p\neq d,e$, then $p\in Q_{a}$, and $\square_{pq}$ was already satisfied before adding the new points, due to Invariant (2).)

Third, let $c$ be an ancestor of $u$ (possibly $u$ itself). By Invariant (4) we have $d',e' \in Q'_u$, where $d' = \coord{d.x}{u.y}$, and $e' = \coord{e.x}{u.y}$, i.e.\ the projections of the endpoints of $I(b)$ to $I(u)$. One of the points $d'$, $e'$ is contained in every $\square_{pq}$, making it satisfied. 

As there are no other cases, Invariant (2) holds.

For Invariant (3), as only $a$ becomes a new parent of $b$, we only
need to show that $\{\coord{b.x}i\mid a.y<i<b.y\}\cap Q'=\emptyset$.
As we have $\{\coord{b.x}i\mid u.y<i<b.y\}\cap Q=\emptyset$
and $d,e\notin\{\coord{b.x}i\mid u.y<i<b.y\}$, we are done. 

For Invariant (4), we need to argue that $\coord{I(b).\min}{a.y}$, $\coord{I(b).\max}{a.y}$, $\coord{I(a).\max}{a.y}$, $\coord{I(a).\max}{u.y}$ $\in Q'$. The first three are true by construction (these are exactly $d$, $e$, and $e$, respectively). The fourth is true, because $\coord{I(a).\max}{u.y} = \coord{I(b).\max}{u.y} \in Q \subseteq Q'$ by Invariant (4) for $T$ and $Q$. \qedd
\end{proof}
We conclude with the proof of \ref{lem:general tran geo}. 
\begin{proof}
[Proof of \ref{lem:general tran geo}]Given $\cA_{\path}$ and $P$,
we let $\cA_{\sat}$ be simply the algorithm that returns $Q$ as
described above. By \ref{lem:inv maintained general}, for each link
of $\cA_{\path}$, one or two new points are added to $Q$. So $|Q|=\Theta(|\cA_{\path}(P)|)$.
By \ref{lem:inv end} and the fact that $Q\supseteq(P\cup\base)$,
$Q$ is a satisfied superset of $P$. Lastly, it is easy to see from
\ref{fig:inv_gen} that all points added to $Q$ are ``below''
$P$, i.e.\ for every $q\in Q_{x=i}$ we have $q.y\le(p_{x=i}).y$,
where $p_{x=i}$ is the unique point in $P_{y=i}$. (To
see this, consider for contradiction the first time a point would
be added above a point of $P$.) In other words, $Q$ is an insertion-compatible
superset of $P^{r}$. \qedd
\end{proof}

\subsubsection{The general transformation to a family of BST algorithms\label{sub:general family step}}

In this subsection we sketch the idea of generating multiple (different)
BST executions.

The argument goes in a similar way as in \S\,\oldref{sub:general step}, except that Invariant (3) is no longer maintained.
Every time
$\cA_{\path}$ performs a link operation and \ref{lem:inv maintained general}
is invoked, there is an alternative way to add points to $Q$ (see
\ref{fig:alt step}) such that Invariants (1), (2) and (4) are maintained, by a similar argument.
However, with this alternative, it is possible that the points $d$ and $e$ are already in the point set $Q$.  
It follows that $|Q|=O(|\cA_{\path}(P)|)$ (instead of $|Q|=\Theta(|\cA_{\path}(P)|)$ as before).
Note also that it is no longer the case that all points added to $Q$ are
below $P$.

\begin{figure}[H]
	\begin{centering}
		\includegraphics[width=0.4\textwidth]{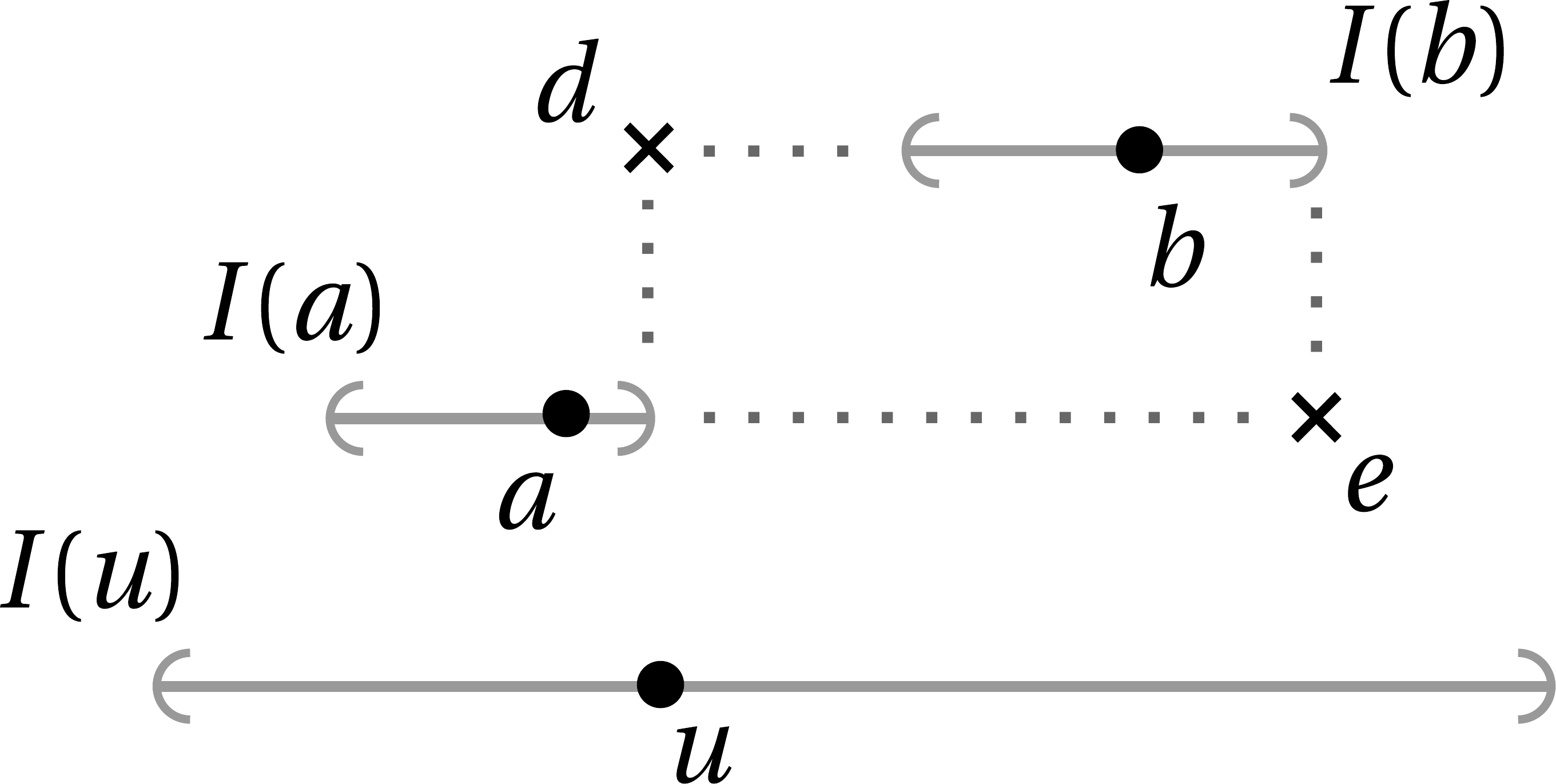}
		\par\end{centering}
	
	\caption{An alternative step in the general transformation. After linking $b$ to $a$,
		the new points $d$ and $e$ (marked as crosses) are added. Dotted lines show horizontal or vertical alignment. \label{fig:alt step}}
\end{figure}

The important observation is that, we can choose, independently for
each link operation, whether to add points as shown in \ref{fig:inv_gen}
or as shown in \ref{fig:alt step}. Thus, we obtain infinitely many
ways of constructing $Q$ (as $n$ goes to infinity). 

Applied to $\smooth$, we obtain various (non-trivial) satisfied supersets,
whose costs are at most (up to constant factors) the conjectured optimal cost
of $\greedy$. The surprising aspect of this result is that small
local changes tend to propagate and affect the future behavior of
algorithms in ways that are otherwise difficult to analyse.

For instance, it appears difficult even to show that $\greedy$\;(in
the geometric view) is competitive with itself, executed on the same
input point set, but with \emph{one additional point} added at time
zero. Understanding the effect of different initial/intermediate states is a bottleneck
in our current understanding of \greedy, and indeed, the BST model.
(See also~\cite{FOCS15}.)

\subsection{Structure of smooth heap after linking all two-neighbor nodes }

Before proving \ref{thm:smooth tran} in \S\,\oldref{sub:smooth
greedy step}, we make some observations about the tree that $\smooth_{\path}$
maintains.

Recall the non-deterministic view of the smooth heap described in \S\,\oldref{sec:smooth_desc}. In sorting mode,
the smooth heap performs $n$ \emph{rounds}, where each round corresponds
to one extract-min operation. In each round, as long as the current
root has more than one child, we repeatedly find an \emph{arbitrary}
local maximum $x$ among the children of the root, and link $x$ to
the neighbor with a larger key (if there are more than one neighbors). 

To facilitate the proof, we assume that, whenever possible, a local maximum with two neighbors
is chosen. If such a node does not exist, then we choose a local maximum
with one neighbor.

We translate this process to describe how $\smooth_{\path}$ works
as follows. Given a permutation $P$, $\smooth_{\path}$ works in
$n$ \emph{rounds. }In the beginning of the $r$-th round $(u_{0},\dots, u_{r-1})$
form a path in the current monotone tree. Let $v_{1},\dots,v_{k}$
be the current children of $u_{r-1}$ ordered from left to right.
We call a node $v_{i}$ a \emph{local maximum} (among the children
of $u_{i-1}$) if $(v_{i}).y>(v_{i-1}).y$ (if $v_{i-1}$ exists)
and $(v_{i}).y>(v_{i+1}).y$ (if $v_{i+1}$ exists). As long as $k>1$,
we choose an arbitrary \emph{local maximum} $v_{i}$, then link $v_{i}$ to
its higher neighbor. Again, we assume that, if possible, a local maximum
with two neighbors is chosen, before choosing a local maximum with
only one neighbor. The following observations are immediate. 
\begin{prop}
\label{prop:structure one-neighbor}Let $v_{1},\dots,v_{k}$ be the
current children of $u_{r-1}$ and suppose that there is no local
maximum node with two neighbors. Then: 
\begin{itemize}
\item $v_{1},\dots,v_{k}$ form a V-shape, i.e.\ for some $i$: 
\[
v_{1}.y>\dots>v_{i-1}.y>v_{i}.y<v_{i+1}.y<\dots<v_{k}.y.
\]

\item Only $v_{1}$ and $v_{k}$ are local maxima. Both have one neighbor. 
\item After linking $v_{1}$ or $v_{k}$, there is still no local maximum
with two neighbors. 
\item At the end of the round, $v_{j}$ is linked to $v_{j+1}$ for all
$j<i$, and $v_{j}$ is linked to $v_{j-1}$ for all $j>i$. 
\end{itemize}
\end{prop}
By this observation, it follows that in each round,
once we have linked local maxima with two neighbors until this is no longer possible, all the remaining links will be of one-neighbor local maxima. 
\begin{prop}
For each round, the sequence of link operations can be divided into
two consecutive phases: a \emph{two-neighbor phase} with links involving
a local maximum with two neighbors, and a \emph{one-neighbor phase}
consisting of the remaining link operations.
\end{prop}

\subsection{The Smooth-Greedy transformation\label{sub:smooth greedy step}}

In this subsection, we finally prove \ref{thm:smooth tran}. The approach
is similar to the one in \S\,\oldref{sub:general step}. The differences
are in (1) the definition of $Q^{init}$ and $I$, and (2) in the
method of updating $Q$ and $I$ for each link performed by $\smooth_{\path}$.

Let $\boxx=\{\coord i0|~i\in[n]\}\cup\{\coord 0i|~0\le i\le n\}\cup\{\coord{n+1}i|~0\le i\le n\}$.
Let $Q^{init}=P\cup\boxx$. Unlike in \S\,\oldref{sub:general
step}, we do not have a closed form formula for the interval function
$I$ at every time. We instead describe how it is updated after each
link. Initially, we set $I(u_{0}).\min=0$, $I(u_{0}).\max=n+1$,
and $I(u_{i}).\min=I(u_{i}).\max=i$, for $i>0$.

We keep invariants (1),(2), and (4) from \S\,\oldref{sub:common framework}
unchanged. We strengthen Invariant (3) to Invariant (3')\footnote{Invariant (3) states that for every node $v$ with parent
$u$, $\{\coord{v.x}i\mid u.y<i<v.y\}\cap Q=\emptyset$. Invariant
(3') states more strongly that for every node $v$, $\{\coord{v.x}i\mid0<i<v.y\}\cap Q=\emptyset$. }. We also introduce one new invariant. 
\begin{itemize}
\item \textbf{Invariant 3' (No points added below nodes): }For every point
$p\in Q\setminus\boxx$ and node $v$, if $p.x=v.x$, then $p.y\ge v.y$. 

\item \textbf{Invariant 5 (Intervals include all non-box points): } For
every node $v$, $I(v)\supseteq(\min_{q\in(Q_{v}\setminus\boxx)}q.x,\max_{q\in(Q_{v}\setminus\boxx)}q.x)$.
\end{itemize}
Observe that earlier, Invariant (5), although not stated explicitly, was true \emph{with equality}, by the definition of the intervals.
Here we have a relaxed statement, to allow for box points being included
in the intervals. 
\begin{lem}
The invariants hold initially.\label{lem:inv init smooth}\end{lem}
\begin{proof}
Observe that the initial intervals form a laminar family whose structure
is exactly $star(P)$, so Invariant (1) holds. Next, for every $i\in[n]$,
$u_{i}$ has only $u_{0}$ as an ancestor. For every point $p\in Q_{y=0}^{init}$
and $q\in Q_{y=i}^{init}$, we have that $\square_{pq}$ is satisfied
as $Q_{y=0}^{init}\supseteq\{\coord i0\mid0\le i\le n+1\}$. So Invariant
(2) holds. Invariants (3'), (4), and (5) hold trivially. \qedd
\end{proof}
Next, we specify how to add points to $Q$ and update the interval
function $I$. There are two cases: whether we are in the two-neighbor
phase (\ref{lem:inv maintained two-neighbor}) or in the one-neighbor
phase (\ref{lem:inv maintained one-neighbor}). We argue these cases separately
in the following subsections.
%
%

\subsubsection{Adding points in two-neighbor phases}

During the the two-neighbor phase, let $T,Q$ and $I$ be the current
monotone tree, point set, and interval function respectively. Suppose
that the invariants hold for $T,Q$ and $I$. We argue that the
invariants hold after just one link. Let $T'$ be obtained from $T$
by a stable link operation performed by $\smooth_{\path}$. We first describe how we update $Q$ and $I$ to obtain
$Q'$ and $I'$. Then, we show that the invariants hold. Finally, we
show that the number of new points in $Q$ is between 1 and 3. By
applying this argument repeatedly for each link until the phase ends,
we conclude:
\begin{lem}
For the two-neighbor phase of every round, the invariants are maintained,
and the number of points added into $Q$ is proportional to the number
of links performed in the phase. \label{lem:inv maintained two-neighbor}
\end{lem}

\paragraph{Updating $Q$ and $I$ after one link.}

Suppose that we are in the $r$-th round. Let $b$ be the local maximum
child of $u_{r-1}$, with two neighbors. Let $a$ and $c$ be the
left and right neighbors of $b$ respectively. By symmetry, assume
$a.y>c.y$. So $b$ is linked to $a$. To update $Q$ and $I$, we
set $Q'=Q\cup\{d,e,f\}$ where 
\begin{align*}
d & =\coord{I(a).\max}{b.y},\\
e & =\coord{I(c).\min}{b.y},\\
f & =\coord{I(c).\min}{a.y}.
\end{align*}
We set 
\begin{align*}
I'(a) & =(I(a).\min,I(c).\min),\\
I'(b) & =(I(a).\max,I(c).\min)
\end{align*}
and $I'(x)=I(x)$ for all nodes $x\neq a,b$. See \ref{fig:inv_smooth_two_neighbor}.

\begin{figure}
	\begin{centering}
		\includegraphics[scale=0.18]{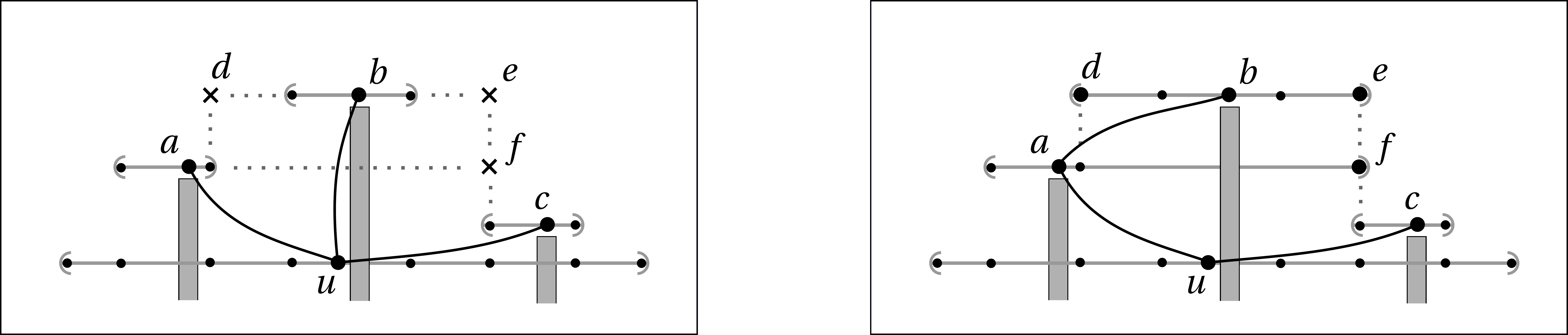}
		\par\end{centering}
	\caption{A two-neighbor link and the maintenance of invariants. \textbf{(left)} Before the transformation: $a$, $b$, and $c$ are children of $u$ in $T$. Crosses mark locations of newly added points $d$, $e$, and $f$. \textbf{(right)} After the transformation: $b$ is a child of $a$ in $T'$. Intervals shown as horizontal segments, respecting tree-structure (Invariant (1) and Invariant (5)); labels omitted. Shaded rectangles indicate empty areas (Invariant (3')). Smaller dots indicate interval endpoints and their projections (Invariant (4)). Dotted lines indicate horizontal and vertical alignment, some lines omitted for clarity. Non-emptiness of newly created rectangles (Invariant (2)) is guaranteed by newly added points and the points of Invariant (4).
		\label{fig:inv_smooth_two_neighbor}}
\end{figure}

\paragraph{Invariants maintained after one link.}

For Invariant (1), observe that $I'(b)\subseteq I'(a)$, and for every
child $a'$ of $a$, $I'(b)\cap I'(a')=\emptyset$ because $I'(a').\max=I(a').\max\le I(a).\max=I'(b).\min$.
Similarly, $I'(a)\subseteq I'(u)$ and $I'(b)\subseteq I'(u)$. As
there are no other changes in the monotone tree $T$ (other than $a$
gaining $b$ as a child), Invariant (1) is maintained.

For Invariant (2), after linking, $a$ becomes a new ancestor of $b$
and of every descendant of $b$. 
This introduces several new pairs of $(p,q)$ where we need to check if $\square_{pq}$ is satisfied.
It is clear that this holds if either $p\in\boxx$ or $q\in\boxx$.
The other cases are as follows:
\begin{enumerate}
\item If $p=d$ and $q$ is below $p$ (i.e. $q.y < p.y$), then $\coord{I(a).\max}{a.y}\in\square_{pq}$. If $p=d$ and $q$ is above $p$, then $\coord{I(b).\min}{b.y}\in\square_{pq}$. Note that $\coord{I(a).\max}{a.y},\coord{I(b).\min}{b.y}\in Q$ by
Invariant (4). 
\item If $p=e$ and $q$ is below $p$, then $f \in \square_{pq}$. If $p=e$ and $q$ is above $p$, then $\coord{I(b).\max}{b.y}\in\square_{pq}$.
Note that $\coord{I(b).\max}{b.y}\in Q$ by Invariant (4). 
\item If $p=f$ and $q$ is below $p$, then $\coord{f.x}{u.y}\in\square_{pq}$. If $p=f$ and $q$ is above $p$, then $e \in \square_{pq}$. 
Note that $\coord{I(b).\max}{b.y}\in Q$ by Invariant (4). 
\item If $p\in Q_{a}\setminus\boxx$ is an old point, then $\coord{I(a).\max}{a.y}\in\square_{pq}$ where $q \in Q_b$. 
\end{enumerate}
The aboves cases exhaust all the cases because only $d,e,f$ are new points, and only $a$ becomes a new ancestor of $b$.
This implies that Invariant (2) holds. 


For Invariant (3'), we show that the points $d,e,f$ are not below
any node of the monotone tree. Let $v_{d}$ be a node such that $v_{d}.x=I(a).\max=d.x$.
Now, by Invariant (4), we have $\coord{I(a).\max}{a.y}\in Q$ and
actually we have $\coord{I(a).\max}{a.y}\in Q\setminus\boxx$. By
Invariant (3'), $v_{d}.y\le a.y<b.y=d.y$. So $d$ does not violate
Invariant (3'). Next, let $v_{e}$ be a node such that $v_{e}.x=I(c).\min=e.x=f.x$.
Again, $\coord{I(c).\min}{c.y}\in Q\setminus\boxx$ by Invariant (4).
Then, Invariant (3') implies $v_{e}.y\le c.y<f.y<e.y$. So $e$ and
$f$ do not violate Invariant (3') either.

For Invariant (4), as $I'(u)=I(u)$ for $u\neq a,b$ we only need
to argue about the endpoints of $I'(a)$ and $I'(b)$. As $I'(a).\min=I(a).\min$,
we have $\coord{I'(a).\min}{a.y}\in Q$ and $\coord{I'(a).\min}{u.y}\in Q$
by Invariant (4) before adding points. As $I'(a).\max=I(c).\min$,
we have $\coord{I'(a).\max}{a.y}=f\in Q'$ and $\coord{I'(a).\max}{u.y}=Q$
by Invariant (4) before adding points. Next, we have $\coord{I'(b).\min}{b.y}=d\in Q'$
and $\coord{I'(b).\min}{a.y}=\coord{I(a).\max}{a.y}\in Q$ by Invariant
(4) before adding points. Finally, we have $\coord{I'(b).\max}{b.y}=e\in Q'$
and $\coord{I'(b).\max}{a.y}=f\in Q'$.

For Invariant (5), only $Q_{a}\setminus\boxx$ and $Q_{b}\setminus\boxx$
change. For $Q_{a}\setminus\boxx$, we have set $I'(a).\max$ 
to be $f.x$. For $Q_{b}\setminus\boxx$, $I'(b)=(d.x,e.x)$. So Invariant
(5) is maintained as well. This concludes that the invariants hold
during the two-neighbor phase.

\paragraph{Counting points added after one link.}

We first observe the following:
\begin{claim}
If $I(b).\min=I(b).\max=b.x$ , then $I(a).\max<b.x<I(c).\min$.\label{claim:singleton get gap}\end{claim}
\begin{proof}
Recall that $b.y>a.y>c.y$. Suppose that $I(a).\max=b.x$. Then, $\coord{I(a).\max}{a.y}\in Q$
by Invariant (4) and actually $\coord{I(a).\max}{a.y}\in Q\setminus\boxx$.
As $\coord{I(a).\max}{a.y}=\coord{b.x}{a.y}$, Invariant (3') implies
that $a.y\ge b.y$ which is a contradiction. So we have $I(a).\max<I(b).\min$.
Symmetrically, suppose that $I(c).\min=b.x$. Then $\coord{I(c).\min}{c.y}\in Q\setminus\boxx$
by Invariant (4). As $\coord{I(c).\min}{c.y}=\coord{b.x}{c.y}$, Invariant
(3') implies that $c.y\ge b.y$ which is a contradiction. \qedd
\end{proof}
Now, we get:
\begin{claim}
$I(a).\max<I(c).\min$. \label{claim:gap ac} 
\begin{proof}
There are two cases: if $I(b).\min<I(b).\max$, then this is easy;
by Invariant (1), $I(a).\max\le I(b).\min<I(b).\max\le I(c).\min$.
Else, \ref{claim:singleton get gap} directly implies the current claim. \qedd
\end{proof}
\end{claim}
We claim that we add at least one and at most three new point into
$Q$:
\begin{claim}
$f\notin Q$. Hence, $1\le|Q'\setminus Q|\le3$.
\begin{proof}
As $I(a).\max<I(c).\min$ and $f.y=a.y$, by Invariant (5), we have
that $f\notin Q_{a}\setminus\boxx$, and so $f\notin Q$. \qedd
\end{proof}
\end{claim}

\subsubsection{Adding points in one-neighbor phases}

Before the first link in the one-neighbor phase, let $T,Q,I$ be the
current tree, point set and interval function respectively, and assume
that the invariants hold for $T$, $Q$, and $I$. Let $T'$ be obtained from
$T$ at the end of the one-neighbor phase. 

After the entire phase, we show how to update $Q$ and $I$ to obtain
$Q'$ and $I'$, and we argue that the invariants hold. Finally,
we bound the total number of points added and the links performed over
all one-neighbor phases. This subsection is for proving the following.
\begin{lem}
For each one-neighbor phase of every round, the invariants are maintained.
The total number of link operations and the total number of points
newly added to $Q$ in all one-neighbor phases over $n$ rounds are
both at most $2n$.\label{lem:inv maintained one-neighbor}
\end{lem}

\paragraph{Update $Q$ and $I$ after one phase.}

\begin{figure}
	\begin{centering}
		\includegraphics[scale=0.2]{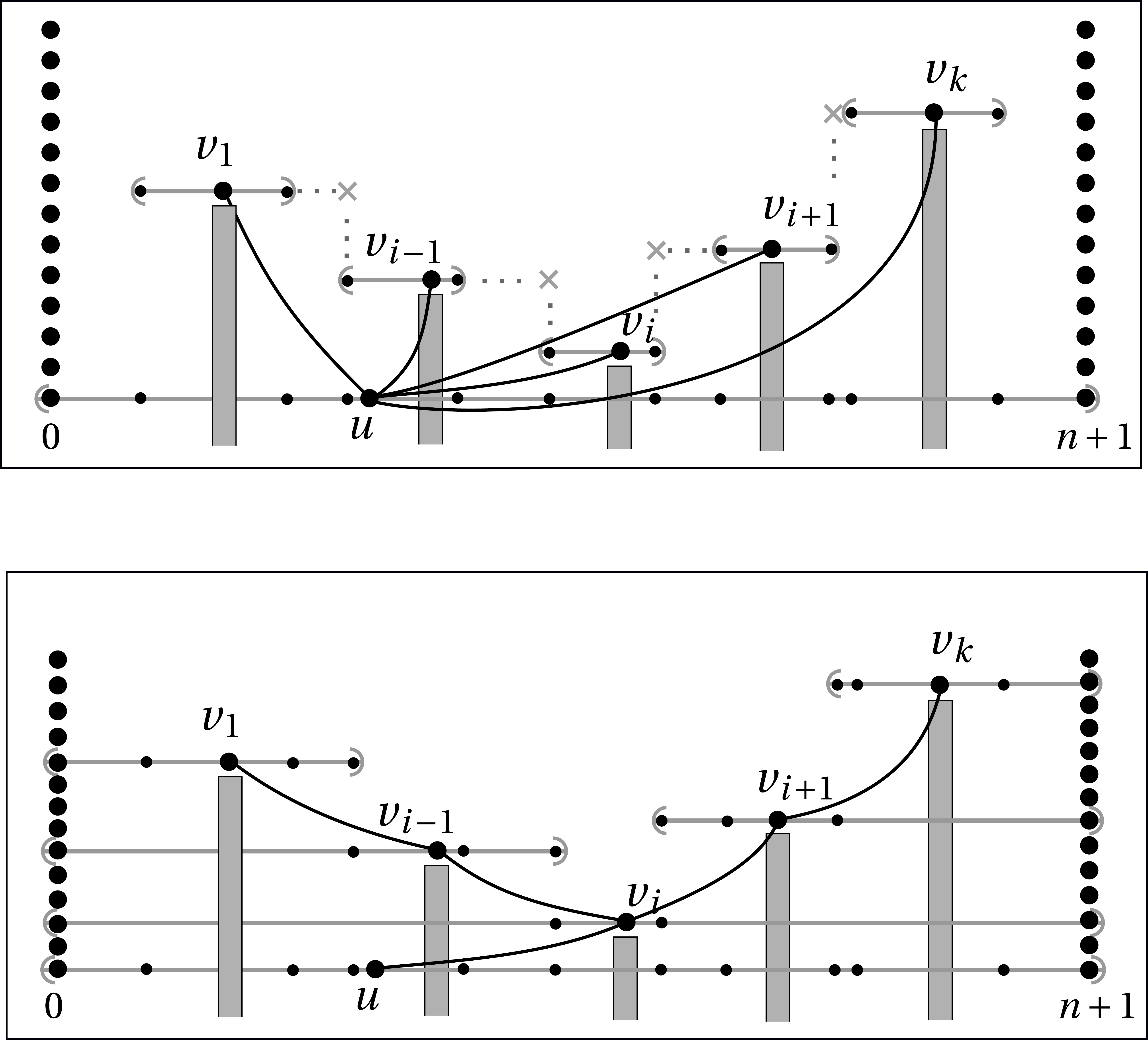}
		\par\end{centering}
	\caption{Links of a one-neighbor phase and the maintenance of invariants. \textbf{(above)} Before the transformation: $v_1, \dots, v_k$ are children of $u$ in $T$. Crosses mark locations of newly added points. \textbf{(below)} After the transformation: $v_1, \dots,v_{i-1}$ and $v_{i+1}, \dots, v_k$ form two paths rooted at $v_i$ in $T'$. Intervals shown as horizontal segments, respecting tree-structure (Invariant (1) and Invariant (5)); labels omitted. Observe that intervals may extend to include box points. Shaded rectangles indicate empty areas (Invariant (3')). Smaller dots indicate interval endpoints and their projections (Invariant (4)). Dotted lines indicate horizontal and vertical alignment, some lines omitted for clarity. Non-emptiness of newly created rectangles (Invariant (2)) is guaranteed by newly added points and the points of Invariant (4). Dots on positions $x = 0$ and $x=n+1$ are points from $\boxx$. \label{fig:inv_smooth_one_neighbor}}
\end{figure}

Suppose that we are at the beginning of the one-neighbor phase of
the $r$-th round. Let $v_{1},\dots,v_{k+1}$ be the current children
of $u_{r-1}$ ordered from left to right. By \ref{prop:structure one-neighbor},
we obtain $T'$ by performing $k$ links to the current local maximum
with one-neighbor (i.e. the current rightmost or leftmost child).
There is some $i$ such that, at the end of the round, $v_{j}$ is
linked to $v_{j+1}$ for all $j<i$, and $v_{j}$ is linked to $v_{j-1}$
for all $j>i$.

For $j<i$, set 
\[
Q'=Q\cup\coord{I(v_{j+1}).\min}{v_{j}.y}\mbox{ and }I'(v_{j})=(0,I(v_{j+1}).\min).
\]

For $j>i$, set 
\[
Q'=Q\cup\coord{I(v_{j-1}).\max}{v_{j}.y}\mbox{ and }I'(v_{j})=(I(v_{j-1}).\max,n+1).
\]
For $j=i$, we do not add points, and we set $I'(v_{j})=(0,n+1)$.
See \ref{fig:inv_smooth_one_neighbor}. Note that we do not claim that
all points that we include are new. That is, possibly, $\coord{I(v_{j+1}).\min}{v_{j}.y}\in Q$
for $j<i$ and $\coord{I(v_{j-1}).\max}{v_{j}.y}\in Q$ for $j>i$.
But we do have that the number of new points is bounded.
\begin{prop}
In each one-neighbor phase, the number of \emph{new }points added
is at most the number of link operations in the phase. \label{prop:new point one neighbor phase}
\end{prop}

\paragraph{Invariants maintained after one phase.}

For each of the invariants, we only argue about $j$ where $j<i$ or
$j=i$. The other case is symmetric.

For Invariant (1), it is enough to show that, for $j<i$, $I'(v_{j})\subseteq I'(v_{j+1})$
and $I'(v_{j})\cap I'(x)=\emptyset$ for any other child $x$ of $v_{j+1}$.
First, $I'(v_{j})\subseteq I'(v_{j+1})$ because if $j+1=i$, then
$I'(v_{j+1})=(1,n)$ (i.e.\ the entire range), and otherwise $j+1<i$,
then $I'(v_{j})=(1,I(v_{j+1}).\min)\subseteq(1,I(v_{j+2}).\min)=I'(v_{j+1})$.
Next, we have $I'(v_{j}).\max=I(v_{j+1}).\min\le I(x).\min=I'(x).\min$
for any other child $x$ of $v_{j+1}$ where the inequality holds by Invariant
(1).

For Invariant (2), we only argue about points outside $\boxx$, as
for every $p\in\boxx$ and every $q$, $\square_{pq}$ is satisfied. Consider, for every $j<i$, $p\in Q'_{v_{j}}=Q{}_{v_{j}}\cup\{\coord{I(v_{j+1}).\min}{v_{j}.y}\}$
and $q\in Q'_{v_{j'}}$. If $j'<j$, then there is another point $\coord{I(v_{j}).\min}{v_{j-1}.y}\in\square_{pq}\cap Q'$.
Else, $j'>j$, and then we have that $\coord{I(v_{j+1}).\min}{v_{j}.y}\in\square_{pq}\cap Q'$.
For $j=i$, for every $p\in Q'_{v_{i}}=Q{}_{v_{i}}$ and $q\in Q'_{v_{j'}}$
where $j'\neq i$, there is $\coord{I(v_{i}).\min}{v_{i-1}.y}\in\square_{pq}\cap Q'$
if $j'<i$ and $\coord{I(v_{i}).\max}{v_{i+1}.y}\in\square_{pq}\cap Q'$
if $j'>i$. 
This implies that Invariant (2) holds.

For Invariant (3'), for every $j<i$, let $v'$ be a node where $v'.x=I(v_{j+1}).\min$.
As $\coord{I(v_{j+1}).\min}{v_{j+1}.y}\in Q$ by Invariant (4) and
actually $\coord{I(v_{j+1}).\min}{v_{j+1}.y}\in Q\setminus\boxx$,
Invariant (3') implies $v'.y\le v_{j+1}.y<v_{j}.y$. So the new point
$\coord{I(v_{j+1}).\min}{v_{j}.y}$ does not violate Invariant (3'). 

For Invariant (4), for every $j<i$, observe that $\coord{I'(v_{j}).\max}{v_{j}.y}=\coord{I(v_{j+1}).\min}{v_{j}.y}\in Q'$
which is newly added, and $\coord{I'(v_{j}).\max}{v_{j+1}.y}=\coord{I'(v_{j+1}).\min}{v_{j+1}.y}\in Q$
by Invariant (4). $\coord{I'(v_{j}).\min}{v_{j}.y},\coord{I'(v_{j}).\min}{v_{j+1}.y}\in\boxx$
because $I'(v_{j})=0$. For $j=i$, as $I(v_{i}).\min=0$ and $I(v_{i}).\max=n+1$,
so $\coord{I'(v_{i}).\min}{v_{i}.y},\coord{I'(v_{i}).\max}{v_{i}.y}\in\boxx$.

Invariant (5) clearly holds. We conclude that the invariants
hold after the one-neighbor phase.

\paragraph{Counting links and points added over all rounds.}

We show that over all $n$ rounds, the number of link operations
performed in all one-neighbor phases is at most $2n$. By \ref{prop:new point one neighbor phase},
it follows that the total number of added points in one-neighbor phases
is $2n$ as well.

Let $w$ be a local maximum with one-neighbor $v$ where $w.y>v.y$,
and $w$ is linked to $v$. If $w.x<v.x$, then we say thet $v$ gets a \emph{one-neighbor
link from the left}. Else $v.x<w.x$, then we say $v$ gets a \emph{one-neighbor
link from the right}. 
\begin{claim}
Every node $v$ gets a one-neighbor link from the left at most
once, and a one-neighbor link from the right at most once.\end{claim}
\begin{proof}
Suppose $v$ gets a one-neighbor link from the left. It must
be the case that $I(v).\min>0$ before linking. Because $v$ has a
left neighbor $w$ and $0<w.x\le I(w).\max\le I(v).\min$. However,
after linking, we have $I(v).\min=0$. So $v$ can never get a one-neighbor
link from the left again. The case from the right is argued symmetrically. \qedd
\end{proof}
It follows that there are at most $2n$ links in
one-neighbor phases over all $n$ rounds. This concludes the proof
of \ref{lem:inv maintained one-neighbor}.

\subsubsection{$\protect\smooth_{\protect\path}$ fills $\protect\add$ gadgets}

In order to draw a connection to $\greedy$, we need the following lemma.
Recall the $\add$ gadget from \ref{def:add gadget}.
We claim that each added point \emph{fills} some $\add$ gadget in
$Q$. 
\begin{lem}
\label{lem:smooth add to gadget}For each new point $x$ added to
$Q$ by either \ref{lem:inv maintained two-neighbor} or \ref{lem:inv maintained one-neighbor},
there is an $\add$ gadget in $Q$ which is filled by $x$.
\end{lem}
To prove the above, we introduce some notation. For every node $v$,
let $v_{\min}=\coord{I(v).\min}{v.y}$ and $v_{\max}=\coord{I(v).\max}{v.y}$
where $I$ is the current interval function before adding points.
For every point $p$, let $v_{p}=\coord{p.x}{v.y}$ be the point $p$
``projected'' to the row $v.y$. Now, there are two cases.

\paragraph{Two-neighbor phases.}

Let $Q$ be the current point
set during a two-neighbor phase. Suppose that there are three nodes
$a,b,c$, where $b$ is a local maximum and $a$ and $c$ are its left
and right neighbors respectively. Assume by symmetry that $a.y>c.y$
and so $b$ is linked to $a$. Then $Q'=Q\cup\{d,e,f\}$ where $d=\coord{I(a).\max}{b.y}$,
$e=\coord{I(c).\min}{b.y}$ and $f=\coord{I(c).\min}{a.y}$. Let $u$
be the parent of $a,b,c$. We define the following as depicted in \Cref{fig:find_gadget_two_neighbor}(a):
\begin{align*}
G_{d} & =(b_{\min},a_{\max},u_{b_{\min}},u_{a_{\max}},u_{c_{\min}}),\\
G_{e} & =(b_{\max},c_{\min},u_{b_{\max}},u_{c_{\min}},u_{a_{\max}}),\mbox{ and}\\
G_{f} & =(a_{\max},c_{\min},u_{a_{\max}},u_{c_{\min}},u_{\min}).
\end{align*}

\begin{figure}
	\begin{centering}
		\includegraphics[scale=0.2]{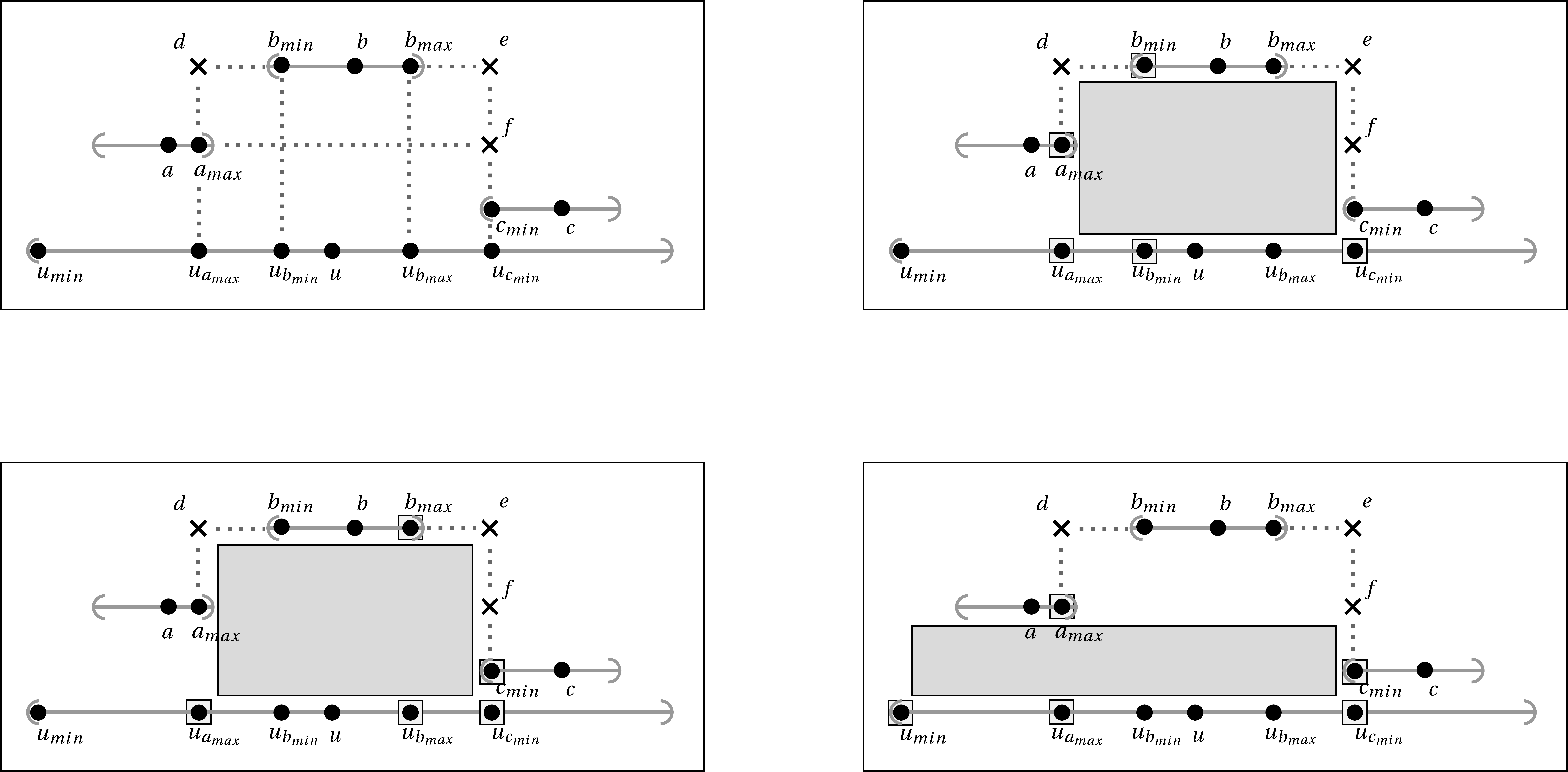}
		\par\end{centering}
	\caption{\textbf{(a)} Points in the proof of \Cref{claim:find_gadget_two_neighbor}. \textbf{(b)}
	 $\protect\add$ gadget $G_{d}$. \textbf{(c)} $\protect\add$
		gadget $G_{e}$. \textbf{(d)} $\protect\add$ gadget $G_{f}$. Shaded rectangles are empty.\label{fig:find_gadget_two_neighbor}}
\end{figure}

\begin{claim}\label{claim:find_gadget_two_neighbor}
If $x\in\{d,e,f\}$ is newly added, i.e., $x\notin Q$,
then $G_{x}$ is an $\add$ gadget in $Q$ and $x$ fills $G_{x}$.\end{claim}
\begin{proof}
Observe that Invariant (4) implies that 
\[
\{a_{\max},b_{\min},b_{\max},c_{\min},u_{a_{\max}},u_{b_{\min}},u_{b_{\max}},u_{c_{\min}},u_{\min}\}\subseteq Q.
\]

(These are the points of $G_d$, $G_e$, $G_f$.)

We first prove the case $x=d$. See \Cref{fig:find_gadget_two_neighbor}(b). We argue that the five
points $b_{\min}$, $a_{\max}$, $u_{b_{\min}}$, $u_{a_{\max}}$, $u_{c_{\min}}$ are distinct by showing that $I(a).\max<I(b).\min<I(c).\min$. To
see this, as $d\notin Q$, we get $I(a).\max<I(b).\min$ (otherwise,
by Invariant (4), $d\in Q$). Also $I(b).\min<I(c).\min$, because if
$I(b).\min=I(b).\max$, then $I(b).\min=b.x<I(c).\min$ by \ref{claim:singleton get gap}.
Else, $I(b).\min<I(b).\max\le I(c).\min$. Next, it is easy to see
that the five points are in the correct relative position for $\add$ gadgets.
It remains to show that $\square_{G_{d}}\cap Q=\emptyset$. Suppose
there is a point $r\in\square_{G_{d}}\cap Q$. Then, let $v_{r}$
be a node such that $v_{r}.y=r.y$. By Invariant (5), we have $r.x\in I(v_{r})$.
By Invariant (1), we obtain a contradiction in one of the following
ways: 1) if $v_{r}$ is a descendant of $a$ or $c$, then $I(v_{r})\not\subseteq I(a)$
or $I(v_{r})\not\subseteq I(c)$ respectively, contradicting Invariant
(1), or 2) if $v_{r}$ is a descendant of $b$, then $r.y\ge b.y$
and so $r\notin\square_{G_{d}}$, or else 3) $v_{r}$ is such that
$I(v_{r})$ is disjoint from $I(a),I(b),I(c)$, and so $r\notin\square_{G_{d}}$.

Next we prove the case $x=e$ in a similar way. See \Cref{fig:find_gadget_two_neighbor}(c).
First, we
argue that the five points $b_{\max},c_{\min},u_{b_{\max}},u_{c_{\min}},u_{a_{\max}}$
are distinct by showing that $I(a).\max<I(b).\max<I(c).\min$. To
see this, as $e\notin Q$, we get $I(b).\max<I(c).\min$ (otherwise,
by Invariant (4), $e\in Q$). Also $I(a).\max<I(b).\max$, because if
$I(b).\min=I(b).\max$, then $I(a).\max<b.x=I(b).\max$ by \ref{claim:singleton get gap}.
Else, $I(a).\max\le I(b).\min<I(b).\max$. Next, it is easy to see
that the five points are in the correct relative position for $\add$ gadgets.
It is left to show that $\square_{G_{e}}\cap Q=\emptyset$. Suppose
there is a point $r\in\square_{G_{e}}\cap Q$. Then, let $v_{r}$
be a node such that $v_{r}.y=r.y$. By Invariant (5), we have $r.x\in I(v_{r})$.
By Invariant (1), we obtain a contradiction in one of the following
ways: 1) if $v_{r}$ is a descendant of $a$ or $c$, then $I(v_{r})\not\subseteq I(a)$
or $I(v_{r})\not\subseteq I(c)$ respectively, or 2) if $v_{r}$ is
a descendant of $b$, then $r.y\ge b.y$ and so $r\notin\square_{G_{e}}$,
or else 3) $v_{r}$ is such that $I(v_{r})$ is disjoint with $I(a),I(b),I(c)$,
and so $r\notin\square_{G_{e}}$.

Remains the case $x=f$. See \Cref{fig:find_gadget_two_neighbor}(d). First,
we argue that the five points $a_{\max},c_{\min},u_{a_{\max}},u_{c_{\min}},u_{\min}$
are distinct by showing that $I(u).\min<I(a).\max<I(c).\min$. To
see this, $I(a).\max<I(c).\min$ by \ref{claim:gap ac}. Also $I(u).\min<I(a).\max$.
There are two cases: if $I(a).\min<I(a).\max$ and we are done because
$I(u).\min\le I(a).\min$. Else, we have $I(a).\min=I(a).\max$, then,
by Invariant (3'), $u_{a_{\max}}\notin Q\setminus\boxx$. As $u_{a_{\max}}\in Q$,
it must be that $u_{a_{\max}}\in\boxx$, i.e. $u.y=0$. Hence, $I(u).\min=0$.
But $I(a).\max>0$. Next, we can see that the five points are in the correct
relative position for $\add$ gadgets. It is left to show that $\square_{G_{f}}\cap Q=\emptyset$.
Suppose there is a point $r\in\square_{G_{f}}\cap Q$. Then, let $v_{r}$
be a node such that $v_{r}.y=r.y$. By Invariant (5), we have $r.x\in I(v_{r})$.
By Invariant (1), we obtain a contradiction in one of the following
ways: 1) if $v_{r}$ is a descendant of $c$, then $I(v_{r})\not\subseteq I(c)$,
or 2) if $v_{r}$ is a descendant of $a$, then $r.y\ge a.y$ and
so $r\notin\square_{G_{f}}$, or else 3) $v_{r}$ is such that 
$I(a).\max < I(v_r).\min \leq I(v_r).\max < I(c).\min$. \qedd
\end{proof}

\paragraph{One-neighbor phases.}
Suppose that we are
at the beginning of the one-neighbor phase of the $r$-th round. Let
$v_{1},\dots,v_{k+1}$ be the current children of $u=u_{r-1}$ ordered
from left to right. Let $i$ be the index such that, at the end of
the round, $v_{j}$ is linked to $v_{j+1}$ for all $j<i$, and $v_{j}$
is linked to $v_{j-1}$ for all $j>i$. By symmetry, we only consider
$j$ where $j<i$ from now. A point $q{}_{j}=\coord{I(v_{j-1}).\max}{v_{j}.y}$
is included into $Q'$ (if it is not already in $Q$). Let $G_{q_{j}}=(v_{j-1,\max},v_{j,\min},u_{v_{j-1,\max}},u_{v_{j,\min}},u_{\min})$. See \Cref{fig:find_gadget_one_neighbor}.

\begin{figure}
	\begin{centering}
		\includegraphics[scale=0.23]{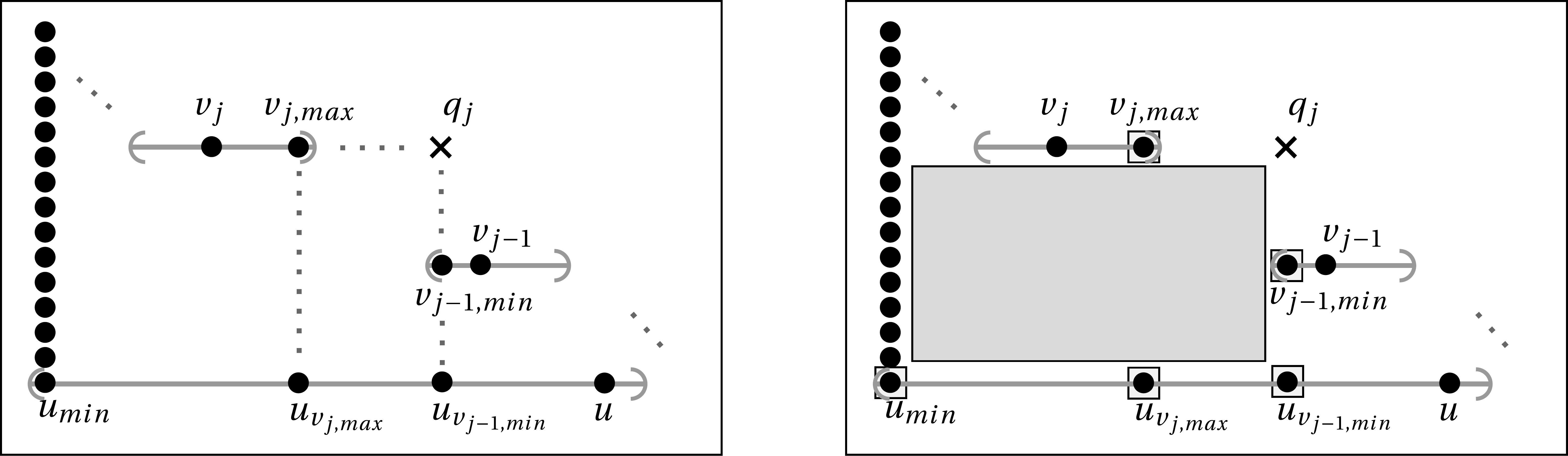}
		\par\end{centering}
	\caption{\textbf{(a)} Points in the proof of \Cref{claim:find_gadget_one_neighbor}. \textbf{(b)}
	 $\protect\add$ gadget $G_{q_{j}}$. Shaded rectangle is empty. \label{fig:find_gadget_one_neighbor}}
\end{figure}

\begin{claim}\label{claim:find_gadget_one_neighbor}
If $q{}_{j}$ is newly added, i.e.\ $q_{j}\notin Q$, then $G_{q_{j}}$
is an $\add$ gadget in $Q$ and $q_{j}$ fills $G_{q_{j}}$.\end{claim}
\begin{proof}
As $q_{j}\notin Q$, then $I(v_{j-1}).\max<I(v_{j}).\min$ (otherwise
$q_{j}\in Q$ by Invariant (4)). So the five points $v_{j-1,\max},v_{j,\min},u_{v_{j-1,\max}},u_{v_{j,\min}},u_{\min}$
are distinct. Also, they are all in $Q$ by Invariant (4), and are
in the correct relative position for an $\add$ gadget. It is left
to argue that $\square_{G_{q_{j}}}\cap Q=\emptyset$. Suppose there
is a point $r\in\square_{G_{p_{j}}}\cap Q$. Then, let $v'$ be a
node such that $v'.y=r.y$. By Invariant (5), we have $r.x\in I(v')$.
By Invariant (1), we obtain a contradiction in one of the following
ways: 1) if $v'$ is a descendant of $v_{j}$, then $I(v')\not\subseteq I(v_{j})$,
or 2) if $v'$ is a descendant of $v_{j-1}$, then $v'.y\ge v_{j-1}.y$
and so $r\notin\square_{G_{q_{j}}}$, or else 3) $v_{r}$ is such
that $I(v_{j-1}).\max<I(v').\min$ and so $r\notin\square_{G_{q_{j}}}$. \qedd
\end{proof}

\subsubsection{Finishing the proof}

Finally, we conclude with the proof of \ref{lem:smooth tran geo}. 
\begin{proof}
[Proof of \ref{lem:smooth tran geo}]First, observe that $|Q|=\Theta(|\smooth_{path}(P)|)$
by \ref{lem:inv maintained two-neighbor} and \ref{lem:inv maintained one-neighbor}. 
We partition $Q=Q_{2}\cup Q_{1}\cup Q^{init}$ where
$Q_{2}$ and $Q_{1}$ are points newly added to $Q$ in two-neighbor
and one-neighbor phases respectively, and $Q^{init}=P\cup\boxx$ is
the initial point set. Similarly, we write $S=\smooth_{path}(P)$
and partition $S=S_{2}\cup S_{1}$ where $S_{2}$ and $S_{1}$ correspond
to the links in two-neighbor and one-neighbor phases respectively.
By \ref{lem:inv maintained two-neighbor}, we have $|S_{2}|\le|Q_{2}|\le3|S_{2}|$,
and, by \ref{lem:inv maintained one-neighbor}, we have $|Q_{1}|,|S_{1}|\le2n$
, and so $|Q_{1}\cup Q^{init}|\le6n$. There are two cases: if $|Q_{2}|\le20n$,
then $|Q|=\Theta(n)$ and $|S|=\Theta(n)$. Else, we have $|Q|=\Theta(|Q_{2}|)=\Theta(|S_{2}|)=\Theta(|S|)$.
Therefore, $|Q|=\Theta(|\smooth_{path}(P)|)$.

Next, we claim that $\greedy(P)=Q\setminus\boxx$ exactly. By \ref{lem:smooth add to gadget},
$Q$ is constructed by repeatedly filling some $\add$ gadget in $Q$.
We claim that this continues until there is no $\add$ gadget in $Q$.
This is because, by \ref{lem:inv init smooth}, \ref{lem:inv maintained two-neighbor}
and \ref{lem:inv maintained one-neighbor}, the invariants hold. So,
by \ref{lem:inv end}, at the end $Q$ is a satisfied superset of
$P\cup\boxx$. Hence, by \ref{lem:cannot add imply sat}, there is
no $\add$ gadget left in $Q$. Recall \ref{def:nondet greedy}. We
can see that $Q$ is produced exactly by the non-deterministic \greedy.
Therefore $Q\setminus\boxx=\greedy_{nondet}(P)=\greedy(P)$, by \ref{thm:greedy nondet}.
So we can conclude $|\greedy(P)|=\Theta(|\smooth_{path}(P)|)$.\qedd\end{proof}

\newpage
\section{Discussion and open questions}
\label{sec:open}
The work presented in this paper raises several new questions. We discuss those that we find the most interesting.

\paragraph{Further study of the stable heap model.}

The fact that a simple modification in the definition of the link operation implies a general connection between heaps and BSTs is intriguing. We believe
that the stable model is interesting in its own right.
 
\begin{itemize}
	\item \textbf{Stable vs.\ non-stable heaps:} How restrictive is stable linking? Are there examples showing separations between the various stable and
	non-stable algorithms? 
	\item \textbf{Instance-specific lower bounds for general sequences:} We
	have transferred instance-specific lower bounds from BSTs to stable heaps when only extract-min operations are used (\S\,\oldref{sub:dynopt for heap}). 
	Can one extend this study to general sequences including extract-min, insert, and decrease-key
	operations?
	\item \textbf{Algorithms (besides smooth heap) in the stable heap model:} What is the complexity
	of the stable variants of pairing heaps described in \ref{fig_heap_algos}?
	How good are their BST counterparts?
\end{itemize}

\paragraph{Further connection between heaps and BSTs.}

By \ref{thm:general tran}, we can obtain an offline BST algorithm
from any stable heap algorithm. The following directions seem promising.
\begin{itemize}
	\item \textbf{Beyond stable heaps: }Can one do the same for (some class of)
	non-stable heap algorithms?
	\item \textbf{Online BSTs: }Our connections yield offline BST algorithms,
	except for $\greedyfuture$, which can be made online~\cite{DHIKP09}.
	Can one derive other online BST algorithms from heaps?
	\item \textbf{Characterization:} Is it true that \emph{every} BST algorithm can be
	obtained from some stable heap algorithm? If not, what is the subclass of BST algorithms for which this is possible? Does it include algorithms such as Splay, Move-to-root,
	Tango, a static BST, etc.? If yes, what stable heap algorithms generate
	these? 
\end{itemize}

\paragraph{Further analysis of smooth heaps.}
In \S\,\oldref{sec:produce sat set}, we gave an upper bound on the amortized cost of extract-min for smooth heaps, and two instance-specific upper bounds. There are several further directions.
\begin{itemize}
	\item What is the amortized cost of \textbf{insert} and \textbf{decrease-key}? The
	techniques developed in~\cite{our_wads} for Greedy may be useful.
	As discussed in~\S\,\oldref{sec:intro}, matching the bounds of Fibonacci heaps by a simpler algorithm is a long-standing problem. Can smooth heaps achieve this?
	\item Does smooth heap satisfy \textbf{working-set-like bounds} similar
	to those studied in~\cite{queaps,ElmasryWS,ElmasryEtc}?
	\item Are smooth heaps (with a careful implementation) efficient \textbf{in practice}? 
\end{itemize}

\paragraph{Status of the conjectured optimality of $\greedyfuture$.}

In \S\,\oldref{sub:consequence greedy} we present results conditional on the the conjectured optimality of $\greedyfuture$ (i.e.\ \ref{conj:greedy opt}). Settling the following questions would give more insight 
towards the proof or refutation of the conjecture, a significant open problem of the field.
\begin{itemize}
	\item \textbf{Optimality of smooth heap:} Is smooth heap (in sorting-mode) an instance-optimal
	stable heap algorithm? By our results, the competitive ratio of the smooth heap is not worse than the competitive ratio of Greedy. Is it better?
	\item \textbf{Inversion and reversion:} Is it true that for every permutation
	$X$, 
	\[
	|\greedyfuture(X)|=\Theta(|\greedyfuture(X')|)=\Theta(|\greedyfuture(X^{r})|)?
	\]
	How about $\smooth(\cdot)$, $\opt_{\stable}(\cdot)$, and $\opt_{\bstins}(\cdot)$?
	If any of these are false, then so is \ref{conj:greedy opt}.
	\item \textbf{Approximation algorithms:} Are there polynomial-time algorithms
	for $O(1)$-approximating $\opt_{\bst}$ and $\opt_{\stable}$? Both
	of these are captured by clean geometric problems (i.e.\ the satisfied
	superset and the star-path problem).

\end{itemize}

\section*{Acknowledgement}
This project has received funding from the European Research Council (ERC) under the European Unions Horizon 2020 research and innovation programme under grant agreement No 715672 and through the ERC consolidator grant No 617951, and from the Israeli Centers of Research Excellence (I-CORE) program (Center No.\ 4/11), and from  the Swedish
Research Council (Reg. No. 2015-04659).

\newpage{}

{\reffontsize
\bibliographystyle{alpha}
\bibliography{paper}	
}
\newpage

\appendix

\section{Adaptive sorting}\label{app_sort}

Several measures of pre-sortedness have been considered in the literature on adaptive sorting. Most, but not all, of these are known to be subsumed by properties of the BST optimum. A full survey of such measures is out of our scope. We briefly discuss four examples. In all cases, we consider an input permutation $X = (x_1, \dots, x_n) \in [n]^n$, and we express the number of comparisons a sorting algorithm performs as a function of $X$.\\

\textbf{Number of inversions.} Perhaps the most natural measure of pre-sortedness of $X$ is the number of inversions $Inv(X)$. It is equal to the number of pairs $i,j$, such that $1 \leq i < j \leq n$, and $x_i > x_j$. Observe that for increasing $X$ we have $Inv(X) = 0$, and in general, $Inv(X) \leq {n \choose 2}$, attained by the decreasing sequence.

There are sorting algorithms that achieve a running time of $O\bigl(n \cdot \log{\bigl(\frac{Inv(X)}{n} + 1 \bigr)}\bigr)$, and this is, in general, well-known to be optimal (see e.g.\ \cite{Mehlhorn79}).
\\

\textbf{Number of runs.} The number of runs, denoted $Run(X)$ is the number of pairs $x_i,x_{i+1}$, such that $x_i > x_{i+1}$. This value is between $0$ (for the increasing sequence) and $n-1$ for the decreasing sequence. A running time of $O\bigl(n \cdot \log{\bigl( Run(X) + 1 \bigr)}\bigr)$ is achievable and, in general, optimal~\cite{Mannila}.\\

\textbf{Number of oscillations.} The number of oscillations of $X$ is defined in~\cite{Levcopoulos89} as follows:
$$ 
Osc(X) = \sum_{i=1}^n{\bigl|\bigl\{j: \min{\{x_i,x_{i+1}\}} < x_j < \max{\{x_i,x_{i+1}\}}\bigr\}\bigr|}.
$$
In words, $Osc(X)$ measures the number of oscillations, i.e.\ the number of times that a consecutive pair ``brackets'' some item in the sequence. Clearly, $0 \leq Osc(X) \leq n^2/2$.

Levcopoulos and Petersson~\cite{Levcopoulos89} show that their sorting algorithm based on the Cartesian tree achieves $O\bigl(n \cdot \log{\bigl(\frac{Osc(X)}{n} + 1 \bigr)}\bigr)$, and that this is, in general, optimal. Furthermore, they show that $Osc(X) \leq 4\cdot Inv(X)$, and $Osc(X) \leq 2n \cdot Run(X) + n$. In consequence, an algorithm that is optimal with respect to the number of oscillations is also optimal (up to a constant factor) with respect to the number of inversions and the number of runs.\\

\textbf{Shuffled monotone sequences.} The quantity $SMS(X)$, introduced in~\cite{Levcopoulos}, is the minimum number of monotone subsequences (not necessarily contiguous) that cover the sequence $X$. Levcopoulos and Petersson~\cite{Levcopoulos} show that their \emph{Slabsort} algorithm achieves the running time $O\bigl(n \cdot \log{\bigl(SMS(X) + 1\bigr)}\bigr)$, which is, in general, optimal. As $SMS(X) \leq Run(X)+1$ clearly holds, the ``shuffled monotone sequences'' measure subsumes the ``number of runs'' measure. \\

In the following we discuss the relationship between the four measures and the BST optimum.

We first show that the BST optimum subsumes the ``number of oscillations'' measure. The argument is based on an unpublished note by Eppstein~\cite{Eppstein-blog}. Then we show that the BST optimum \emph{almost} subsumes the ``shuffled monotone sequences'' measure.

Let $DF(X)$ denote the \emph{dynamic finger} bound, i.e.\ $$DF(X) = \sum_{i=2}^n{\log{(|x_{i+1} - x_i| + 1)}}.$$ It is known that $\opt_{\bst}(X) \leq O(DF(X))$, furthermore, the Splay and \greedy algorithms also match this bound~\cite{finger1, finger2, LI16}.

\begin{lem}[Eppstein~\cite{Eppstein-blog}]
$$DF(X) \leq O\left(n \cdot \log{\left(\frac{Osc(X)}{n} + 1 \right)}\right).$$
\end{lem}
\begin{proof}
We denote $M_i = \max{\{x_i,x_{i+1}\}}$ and $m_i = \min{\{x_i,x_{i+1}\}}$. We thus have $DF(X) = \sum_{i=1}^{n-1}{\log{(M_i - m_i + 1)}}$. 
We say that $i$ \emph{brackets} $j$ if $m_i < x_j < M_i$.

Consider the quantity:
$$Q = \sum_{i=1}^{n-1} {\sum_{~j: ~i \mbox{~brackets~} j~} {\frac{1}{M_i - x_j}}}.$$

Let $H_n$ denote the $n$-th harmonic number, i.e.\ $H_n = \sum_{i=1}^n{\frac{1}{n}}$. Since all items bracketed by $i$ have distinct ranks between $m_i$ and $M_i$, we have:
$$
Q = \sum_{i=1}^{n-1} {\sum_{j=m_{i}+1}^{M_i-1} {\frac{1}{M_i - j}}} = \sum_{i=1}^{n-1} H_{(M_i-m_i - 1)} = \Theta(DF(X)).$$

On the other hand we have:

$$
Q = \sum_{j=1}^{n} {\sum_{~i: ~i \mbox{~brackets~} j~} {\frac{1}{M_i - x_j}}}.
$$
Let $Osc_j(X)$ denote the number of items that bracket $j$. We have $Osc(X) = \sum_{j=1}^n{Osc_j(X)}$.

Since for all $i$ that bracket $j$, at most two have the same value $M_i$, we have $$
Q \leq \sum_{j=1}^n {2 H_{Osc_j(X)}}.
$$
By convexity, this quantity is maximized if all $Osc_j(X)$ are equal, and thus $$Q \leq O\left(n \cdot \log{\left(\frac{Osc(X)}{n} + 1 \right)}\right). \qedd$$
\let\qed\relax
\end{proof}

We now discuss the connection between the BST optimum and the ``shuffled monotone sequences'' measure.

Let $k$ denote the quantity $SMS(X)$. In~\cite{FOCS15} we study the restricted variant of $SMS(X)$ in which all monotone subsequences are increasing. We show (from more general results for pattern-avoiding sequences) that $\opt_{\bst}(X) \leq n \cdot 2^{O(k^2)}$. Furthermore, this bound is achieved by \greedy. In~\cite{landscape} we show, building on~\cite{LI16}, the stronger $\opt_{\bst}(X) \leq O(n \cdot k \log{k})$, which is at this time not known to be matched by online BST algorithms. (In~\cite{landscape} the bound is stated for shuffled \emph{increasing} sequences, but the argument easily extends to shuffled \emph{monotone} sequences.)

Thus, we observe that the BST optimum is known to match the $O(n \log{k})$ bound from the adaptive sorting literature only for small values of $k$. Whether, in fact, it holds that $\opt_{\bst}(X) \leq O(n \cdot \log{k})$, i.e.\ whether the BST optimum subsumes the $SMS$ measure is an interesting open question. We note that, in general, we cannot expect the BST optimum to subsume \emph{every} structural measure for adaptive sorting; for every input sequence there is \emph{some} sorting algorithm that can sort that particular sequence with a linear number of comparisons, whereas $\opt_{\bst} = \Omega(n \log{n})$ for most sequences.

\section{Cartesian tree sorting}\label{app_cartesian}
Levcopoulos and Petersson~\cite{Levcopoulos89} show that the quantity 
$n \cdot \log{\bigl(\frac{Osc(X)}{n} + 1 \bigr)}$ is not just an upper bound, but (up to a small constant factor) a closed-form expression for the exact running time of the Cartesian sort algorithm.

We show an example in which sorting with the smooth heap is significantly faster than sorting with the Cartesian tree algorithm. The example is the ``tilted grid'' permutation, shown in~\ref{fig:tilted}. 
A simple inductive argument shows that the cost of Greedy on both this permutation and its inverse is linear~\cite{FOCS15, landscape}. As a consequence, by our results in this paper, the cost of sorting the permutation by smooth heap is $O(n)$. On the other hand, we observe that at least half of the items are ``bracketed'' by $\Theta(\sqrt{n})$ other items, yielding $Osc(X) = \Omega(n\sqrt{n})$. The running time of Cartesian sort on this sequence is thus $\Omega(n \log{n})$.

\begin{figure}[H]
	\begin{center}
\includegraphics[width=0.45\textwidth]{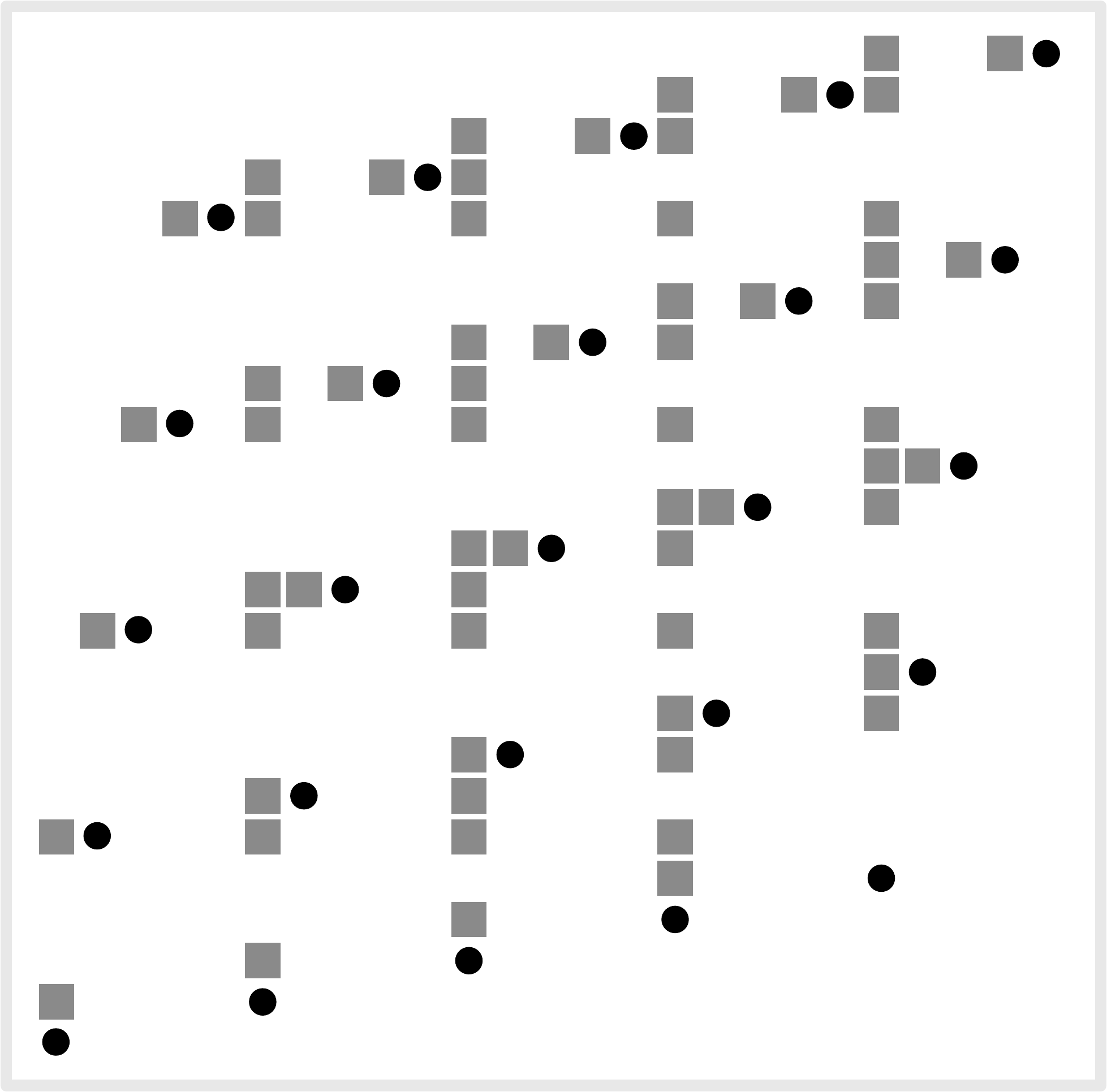}
	\end{center}
\caption{Tilted grid sequence. Points denote input permutation point set, squares denote \greedy output. The sequence is defined as follows. Let $n=t^2$ and consider the point set $\mathcal{P} = \{(i t + (j-1),jt + i -1)~|~i,j \in [t]\}$. It is not hard to verify that no two points align on either $x$ or $y$ coordinates, therefore $\mathcal{P}$ defines a permutation of $[n]$. The example is a special case of a more general family of ``perturbed grid'' sequences, amenable to similar analysis~\cite{landscape}.
\label{fig:tilted}}
\end{figure}

As the running time of Cartesian sort is subsumed by the dynamic finger bound, an example on which Cartesian tree sorting is asymptotically faster than smooth heap sorting would contradict the dynamic optimality of \greedy (\ref{conj:greedy opt}). In fact, it would contradict even the weaker conjecture that $\greedy(X)$ and $\greedy(X')$ are within a constant factor of each other.

\section{Geometry of insert-only execution}\label{app_insert}

\begin{thm}
	[Geometry of BSTs, insertion-compatible (Restatement of \ref{thm:geo bst})] Let
	$X\in[n]^{n}$ be an arbitrary permutation, and let $Q\subset[n]^{2}$.
	$Q$ is a insertion-compatible satisfied superset of $P^{X}$ iff
	there is an offline BST algorithm $\cA_{\bst}$ in insert-only mode
	such that $\cA_{\bst}(X)=Q$.\end{thm}
\begin{proof}
	($\Rightarrow$) This direction is similar to \cite[Lemma 2.2]{DHIKP09}.
	In their proof, each node is assigned with a priority which is the
	next touch time, and then the min-treap is maintained. By the property
	of this treap, at time $i$, the keys from $K_{i}=\{p.x\mid p\in Q_{y=i}\setminus\{p_{y=i}\}\}$
	form a connected subtree containing the current root of the treap.
	So at time $i$, we just touch exactly $K_{i}$ which contains both
	the current predecessor and successor of $p_{y=i}.x$. Then, we insert
	$p_{y=i}.x$ into the subtree formed by $K_{i}$ and rearrange the
	subtree with next touch time as priority again. This implies that
	each time $i$, we touch exactly $K_{i}$. Hence, the execution trace
	is exactly $Q$. 
	
	($\Leftarrow$) This is exactly \cite[Lemma 2.1]{DHIKP09} but with
	an additional observation that, for any key $x$, $x$ is never touched
	before it is inserted. So $Q$ is insertion-compatible for $P^{X}$. \qedd \end{proof}

\section{Detailed pseudocode of smooth heap}\label{smooth_heap_ps}

\begin{figure}[H]
	\begin{center}
		\parbox{5.5in}{
\begin{framed}
\hspace{0.1in}
			\begin{algorithm}[H]
				\newcommand{\llWhile}[2]{{\let\par\relax\lWhile{#1}{#2}}}
				\newcommand{\llIf}[2]{{\let\par\relax\lIf{#1}{#2}}}
				\newcommand{\llElse}[1]{{\let\par\relax\lElse{#1}}}
				\DontPrintSemicolon
				\SetAlgoLined
				\NoCaptionOfAlgo
				\KwIn{A list of nodes, with $x$ initially pointing to the leftmost node}
				
				\vspace{0.1in}
				\itemindent=30pc
				\nlset{(LTR)}\While{$x.\NEXT \neq \NULL$}{
					\llIf{$(x.\KEY < x.\NEXT.\KEY)$}{$x \gets x.\NEXT$ \;}
					\SetAlgoNoLine
					
					\Else(\Comment*[f]{Local maximum found}  ){ 
						\SetAlgoLined 
						\While{$(x.\PREV \neq \NULL)$}{
							
							\SetAlgoNoLine
							\If{$(x.\PREV.\KEY > x.\NEXT.\KEY)$}{
								$ x \gets x.\PREV $ \;
								$\link(x)$ \;
							}
							
							\Else{
								$ x \gets x.\NEXT $ \;
								$\link(x.\PREV)$ \;
								\textbf{continue} \textbf{\textsc{\scriptsize(LTR)}}
							}
						}
						$\link(x)$ \;
					}

				}
				\SetAlgoLined
				
				\vspace{0.1in}
				
				\nlset{(RTL)}\While{$x.\PREV \neq \NULL$}{
					$x \gets x.\PREV$ \;
					$\link(x)$ \;
				}
				
			\end{algorithm}
\end{framed}
			
			\caption{Pseudocode of smooth heap re-structuring}
		}
	\end{center}
\end{figure}

\section{Example in which smooth heap is not optimal}\label{app:nonoptimal}
\begin{figure}[H]
	\begin{centering}
		\includegraphics[width=0.9\textwidth]{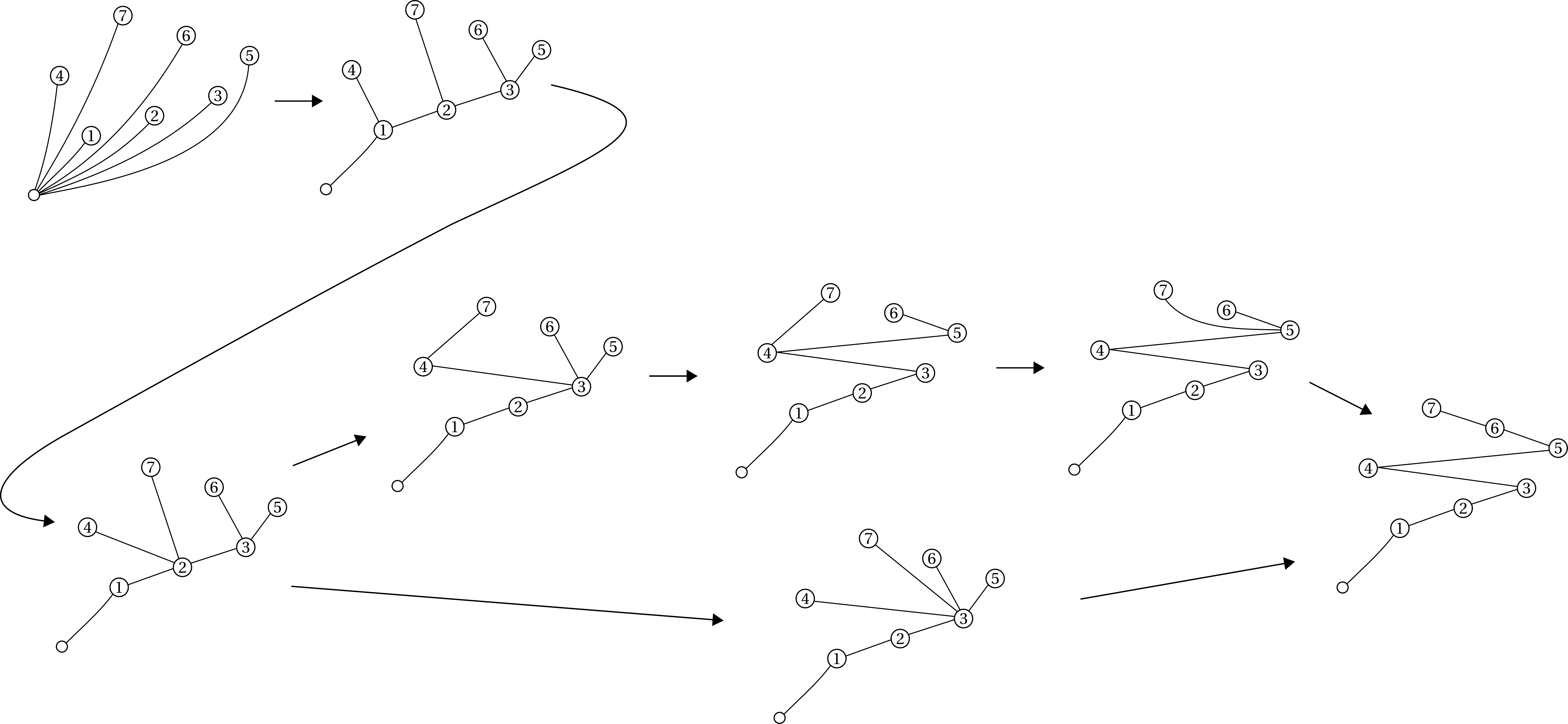}
		\par\end{centering}
	
	\caption{Sequence of operations illustrating that smooth heap is not globally optimal. Top branch: smooth heap execution in sorting-mode for $X=(4,1,7,2,6,3,5)$. Total link-cost is $13$. Bottom branch: an alternative stable heap execution for the same input. Total link-cost is $12$. \label{fig_counter}}
\end{figure}

\end{document}